\documentclass[a4paper,11pt]{article}
\usepackage[a4paper]{geometry}
\usepackage{amssymb,amsmath,graphicx,epsfig,amscd,array,color}
\usepackage{amsthm}
\usepackage{subfigure}
\usepackage{pstricks, pst-text}
\usepackage[round]{natbib}
\usepackage{pgf}
\usepackage{ppfp}
\usepackage{authblk}

\newtheorem{cor}{Corollary}

\newtheorem{obs}{Observation}

\newcommand\norm[1]{\lVert#1\rVert}

\DeclareMathOperator{\Cov} {Cov}
\DeclareMathOperator{\diag}{diag}
\DeclareMathOperator{\ran} {ran}

\title{Specifying Gaussian Markov Random Fields with Incomplete Orthogonal Factorization using Givens Rotations}
\author[1]{Xiangping Hu\footnote{Corresponding author. Email: \texttt{Xiangping.Hu@math.ntnu.no}}}
\author[1]{Daniel Simpson}
\author[1]{H\aa{}vard Rue}
\affil[1]{Department of Mathematical Sciences, Norwegian University of Science and Technology, N-7491 
Trondheim, Norway}
\date{July 4, 2013}
\begin{document}

\maketitle

 \begin{abstract}
In this paper an approach for finding a sparse incomplete Cholesky factor through an incomplete orthogonal factorization with Givens rotations
is discussed and applied to Gaussian Markov random fields (GMRFs). The incomplete Cholesky factor obtained from the incomplete orthogonal factorization is usually
sparser than the commonly used Cholesky factor obtained through the standard Cholesky factorization. On the computational side, 
this approach can provide a sparser Cholesky factor, which gives a computationally more efficient representation of GMRFs. On the theoretical side, this approach 
is stable and robust and always returns a sparse Cholesky factor. Since this approach applies both to square matrices and to rectangle matrices, 
it works well not only on precision matrices for GMRFs but also when the GMRFs are conditioned on a subset of the variables or on observed data. 
Some common structures for precision matrices are tested in order to illustrate the usefulness of the approach.
One drawback to this approach is that the incomplete orthogonal factorization is usually slower than the standard Cholesky factorization implemented in standard libraries and currently it 
can be slower to build the sparse Cholesky factor. 
\end{abstract}

\section{Introduction} \label{sec: cholesky_introduction}
Gaussian Markov random fields(GMRFs) are useful models in spatial statistics due to the Gaussian properties together with Markovian
structures. They can also be formulated as conditional auto-regressions (CARs) models \citep{rue2005gaussian}. GMRFs have applications in many areas, such as spatial statistics,
time-series models, analysis of longitudinal survival data, image analysis and geostatistics. See~\citet[Chapter $1$]{rue2005gaussian} for more
information and literature on how the {GMRFs} can be applied in different areas. From an analytical point of view GMRFs have good properties and can be specified through mean values $\boldsymbol{\mu}$ and covariance matrices $\boldsymbol{\Sigma}$.
While from a computational point of view {GMRFs} can conveniently specified through precision matrices $\boldsymbol{Q}$ (the inverse of the covariance matrices $\boldsymbol{\Sigma}$), which are usually sparse matrices. 
The numerical algorithms for sparse matrices can be exploited for calculations with the sparse precision matrices and hence fast statistical inference is possible~\citep{rue2001fast}. The numerical
algorithms for sparse matrices can be applied to achieve fast simulation of the fields and evaluation the densities (mostly, log-densities) of GMRFs and GMRFs with
conditioning on subset of variables or linear constraints. See~\citet[Chapter 2]{rue2005gaussian} for further details.
These algorithms can also be used to calculate the marginal variances~\citep{rue2005marginal}, and they can be extended to non-Gaussian cases ~\citep{rue2004approximating}.

Precision matrices $\boldsymbol{Q}$ are commonly used to specify GMRFs. This approach is natural due to the
sparsity patterns of the precision matrices in Markovian models. In many situations the Cholesky factors are required and are crucial for simulation and inferences with GMRFs, and the Cholesky factors are normally obtained with
Cholesky factorization routines in standard libraries. See~\citet[Chapter $2$]{rue2005gaussian} for different simulation algorithms for GMRFs using Cholesky factors. In order to
get an even sparser Cholesky factor, with the purposes of saving computational resources, \citet{wist2006specifying} showed that the
Cholesky factor from an incomplete Cholesky factorization can be much sparser than the Cholesky factor from the regular Cholesky factorization. However, they 
provided theoretical and empirical evidence showing that the representation of sparser Cholesky factor was fragile when conditioning
the GMRF on a subset of the variables or on observed data. It means that the sparsity patterns of the sparser Cholesky factors are destroyed when some constraints or observed data are introduced and the computational cost increases. 
Additionally, the sparser Cholesky factor from the incomplete Cholesky factorization is only valid for a specific precision matrix. Their
approach is illustrated in Figure \ref{fig: cholesky_diag} with Routine $1$. 

In this paper a different approach is chosen to solve the problem presented by \citet{wist2006specifying} . The main idea is given in the Figure \ref{fig: cholesky_diag} with Routine $2$. 
In this approach one rectangular matrix $\boldsymbol{A}$ is formulated,

\begin{equation} \label{eq: rectanglar_matrix_A}
\boldsymbol{A} =  \begin{pmatrix} \boldsymbol{L}_1^{\mbox{T}}\\\boldsymbol{L}_2^{\mbox{T}}\end{pmatrix}.
\end{equation}

\noindent It consists of the Cholesky factor $\boldsymbol{L}^{\mbox{T}}_1$ from the precision matrix $\boldsymbol{Q}_1$ of a given GMRF and the Cholesky factor
$\boldsymbol{L}^{\mbox{T}}_2$ of the matrix $\boldsymbol{Q}_2$. The matrix $\boldsymbol{Q}_2$ can be the additional effect
when the GMRF is conditioned on observed data or on a subset of the variables. Both $\boldsymbol{L}_1$ and
$\boldsymbol{L}_2$ are lower triangular matrices.

An incomplete orthogonal factorization is then used to factorize the matrix $\boldsymbol{A}$ in Equation \eqref{eq: rectanglar_matrix_A} to find the sparse Cholesky factor
 for specifying the GMRF. It is shown that by using this approach an upper triangular matrix $\boldsymbol{R}$ which is sparser than the standard Cholesky factor is obtained. 
Furthermore, this approach is applicable when the GMRF is conditioned on a subset of the variables or on observed data. 
Since the upper triangular matrix $\boldsymbol{R}$ is sparser in structure than the common Cholesky factor, it is better for applications.

\begin{figure} 
\centering
\begin{pspicture}(1, 1)(10,10)
\psframe[framearc=0.2, linewidth=0.8pt](1,1)(5,2)
\rput[c](3,1.5){$\boldsymbol{\bar{L}}\boldsymbol{\bar{L}}^{T}$}
\psframe[framearc=0.2, linewidth=0.8pt](1,3)(5,4)
\rput[c](3, 3.5){$\boldsymbol{\bar{L}}_1 \boldsymbol{\bar{L}}_1^{T} + \boldsymbol{Q}_2$}
\rput[c](3,2.5){$\Downarrow$}
\rput[c](3.35,4.5){$(1)$}
\psframe[framearc=0.2, linewidth=0.8pt](1,5)(5,6)
\rput[c](3,5.5){$\bar{\boldsymbol{L}}_1 \bar{\boldsymbol{L}}_1^{T}$}
\rput[c](3,4.5){$\Downarrow$}
\rput[c](6, 7.5){$\Longrightarrow$}
\rput[c](6,7.8){$(2)$}
\psframe[framearc=0.2, linewidth=0.8pt](1,7)(5,8)
\rput[c](3,7.5){$\boldsymbol{L}_1 \boldsymbol{L}_1^{T}$}
\rput[c](3,6.5){$\Downarrow$}
\rput[c](5.1,6.5){\text{\tiny{Incomplete Cholesky factorization}}}
\psframe[framearc=0.2, linewidth=0.8pt](1,9)(5,10)
\rput[c](3, 9.5){\text{Precision matrix} $\boldsymbol{Q}_1$}
\rput[c](3,8.5){$\Downarrow$}
\rput[c](4.4,8.5){\text{\tiny{Cholesky factorization}}}
\psframe[framearc=0.2, linewidth=0.8pt](6.9,7)(10.1,8)
\rput[c](8.5,7.5){$\left[\boldsymbol{L}_1, ~\boldsymbol{L}_2 \right] \left[\boldsymbol{L}_1, ~\boldsymbol{L}_2 \right]^{T}$}
\rput[c](8.5,6.5){$\Downarrow$}
\rput[c](9.1,6.5){\text{\tiny{cTIGO}}}
\psframe[framearc=0.2, linewidth=0.8pt](7,5)(10,6)
\rput[c](8.5,5.5){$\boldsymbol{\tilde{L}}\boldsymbol{\tilde{L}}^{T}$}
\end{pspicture}
\caption{Diagram for the algorithm for finding sparser Cholesky factor by incomplete Cholesky factorization used by \citet{wist2006specifying} (Routine $1$)
and the algorithm used in this paper (Routine $2$).}
\label{fig: cholesky_diag}
\end{figure}
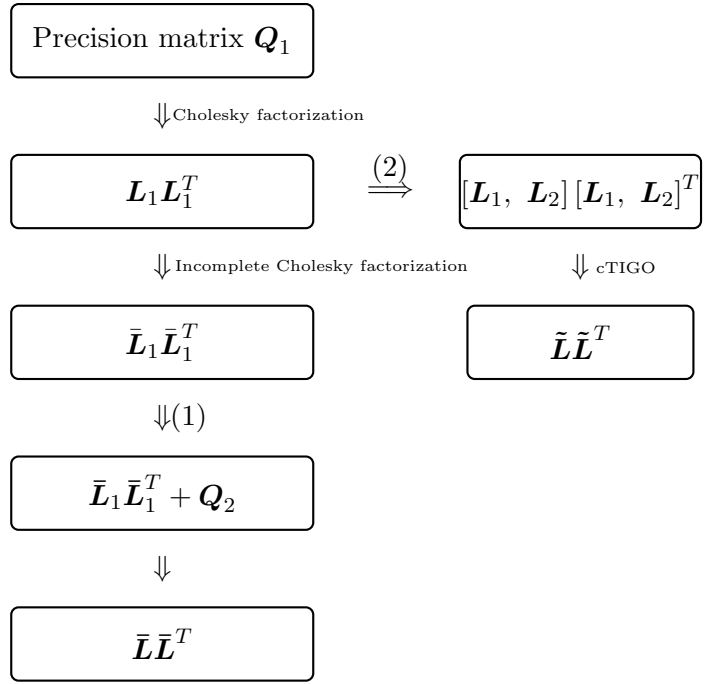

The rest of this paper is organized as follows. In Section \ref{sec: cholesky_background} some basic theory on GMRFs, sparsity patterns for
precision matrices and Cholesky factors of GMRFs are presented. Some basic theories on the orthogonal factorization and the incomplete orthogonal factorization are also introduced in this section.
In Section \ref{sec: cholesky_algorithm} the algorithm for obtaining the sparse Cholesky factor from the
incomplete orthogonal factorization is introduced. A small example is given in order to illustrate how the algorithm works when the GMRFs are
conditioned on a subset of the variables or on observed data. Results for different structures on the precision matrices are given in Section \ref{sec: cholesky_results}. 
Conclusion and general discussion in Section \ref{sec: cholesky_conclusion} ends the paper.

\section{Background and Preliminaries} \label{sec: cholesky_background}

\subsection{Basic theory on GMRFs} \label{sec: cholesky_GMRFs}
A random vector $\boldsymbol{x}=(x_1,x_2,\ldots,x_n)^{\mbox{T}} \in \mathbb{R}^n$ is called a GMRF if it is Gaussian distributed and processes
a Markov property. The structure of a GMRF is usually presented by a labeled graph $\mathcal{G} = (\mathcal{V}, \mathcal{E})$, where $\mathcal{V}$
is the set of vertexes $\{ 1,2,\ldots, n \}$ and $\mathcal{E}$ is the set of edges. The graph $\mathcal{G}$ satisfies the properties that no edge between node $i$ and node $j$ if and only if 
$x_i \perp x_j | \boldsymbol{x}_{-ij}$ \citep{rue2005gaussian}, where  $\{ \boldsymbol{x}_{-ij}; i, j = 1, 2, \dots, n\}$ denotes $\boldsymbol{x}_{-\{i,j\}}$.
If the random vector $\boldsymbol{x}$ has a mean
$\boldsymbol{\mu}$ and a precision matrix $\boldsymbol{Q}_1>0$, the probability density of the vector $\boldsymbol{x}$ is

\begin{equation}\label{eq: cholesky_density_function}
\pi(\boldsymbol{x}|\boldsymbol{\theta}) = 
\left(\frac{1}{2\pi} \right)^{n/2}|\boldsymbol{Q}_1(\boldsymbol{\theta})|^{1/2}\exp\left(-\frac{1}{2}(\boldsymbol{x}-\boldsymbol{\mu})^{\mbox{T}}\boldsymbol{Q}_1(\boldsymbol{\theta})(\boldsymbol{x}-\boldsymbol{\mu})\right),
\end{equation}
with the property 
$$ Q_{ij}\neq 0 \Longleftrightarrow \{i,j\} \in \mathcal{E}
\text{ for all }  i \neq j. $$
The notation $\boldsymbol{Q}_1>0$ means that $\boldsymbol{Q}_1$ is a symmetric positive definite matrix. $\boldsymbol{\theta}$ denotes the parameters in the precision matrix. 
This implies that any vector with a Gaussian distribution and a symmetric positive definite 
covariance matrix is a Gaussian random field (GRF), and GMRFs are GRFs with Markov properties.
The graph $\mathcal{G}$ determines the nonzero pattern of $\boldsymbol{Q}_1$. 
If $\mathcal{G}$ is fully connected, then $\boldsymbol{Q}_1$ is a complete dense matrix. A useful property of GMRF is that we can know whether $x_i$ and $x_j$ are conditionally 
independently or not directly from the precision matrix $\boldsymbol{Q}_1$ and the graph $\mathcal{G}$.
Values of mean $\boldsymbol{\mu}$ do not have any influence on the pairwise conditional independence properties of the GMRFs, and hence we set $ \boldsymbol{\mu} = \boldsymbol{0}$ 
in the following sections unless otherwise specified.
The diagonal elements in the precision matrix $Q_{ii} \, (i = 1, 2, \dots, n)$ are the conditional precisions of $x_{i}$ given all the other nodes $\boldsymbol{x}_{-i}$. The off-diagonal elements 
$Q_{ij} \, (i,j = 1, 2, \dots, n, i \neq j)$ can provide information about the correlations between $x_i$ and $x_j$ given on the nodes $\boldsymbol{x}_{-ij}$.
These are the main differences in the interpretation between the precision matrix $\boldsymbol{Q}_1$ and the covariance matrix $\boldsymbol{\Sigma}_1$. 
The  covariance matrix $\boldsymbol{\Sigma}_1$ contains the marginal variance of $x_i$ and the marginal correlation between $x_i$ and $x_j$. However, with the precision matrix 
the marginal properties are not directly available \citep{rue2005gaussian}.

Since $\boldsymbol{Q}_1$ is symmetric positive definite, there is a unique Cholesky factor $\boldsymbol{L}_1$ where $\boldsymbol{L}_1$ is a lower triangle matrix
satisfying $ \boldsymbol{Q}_1 = \boldsymbol{L}_1\boldsymbol{L}_1^{\mbox{T}} $. If we want to sample from the GMRF $\boldsymbol{x} \thicksim \mathcal{N}(\boldsymbol{\mu}, \boldsymbol{Q}_1^{-1})$, 
the Cholesky factor $\boldsymbol{L}_1$ is commonly used. One algorithm for sampling GMRFs is given in Section \ref{sec: cholesky_sampling}.
More algorithms for sampling GMRFs with different specifications are also available. See \citet[Chapter $2$]{rue2005gaussian} for a detailed discussion on these algorithms.
\citet{rue2005gaussian} showed how to check the sparsity pattern of the Cholesky factor of a GMRF. Define

\begin{equation}\label{eq: cholesky_F}
F(i,j) = {i+1, i+2, \ldots, j-1, j+1, \ldots, n,}
\end{equation}
which is the future of i except j. Then
\begin{equation}\label{eq: cholesky_conditional_indpendent}
x_i \bot x_j\ | \ \boldsymbol{x}_{F(i,j)}  \Longleftrightarrow
L_{ji} = 0,
\end{equation}
 and $F(i,j)$ is called a separating subset of $i$ and $j$. 
 If $i \sim j$ denotes that $i$ and $j$ are neighbors, then $F(i,j)$ cannot be a separating subset for $i$ and $j$ whenever $i \sim j$.
Further, the Cholesky factor of the precision matrix of a GMRF is always equally dense or denser than the lower triangle
part of $\boldsymbol{Q}_1$. 

In many situations there are more nonzero elements in $\boldsymbol{L}_1$ than in the lower triangular part of $\boldsymbol{Q}_1$. 
Denote $n_{\boldsymbol{L}_1}$ and $n_{\boldsymbol{Q}_1}$  the numbers of nonzero elements in the Cholesky factor $\boldsymbol{L}_1$ and the lower triangular part of precision matrix $\boldsymbol{Q}_1$, respectively.
The difference $n_f =  n_{\boldsymbol{L}_1} - n_{\boldsymbol{Q}_1}$ is called the fill-in. The ideal case is $n_f = 0$ or $n_{\boldsymbol{L}_1} = n_{\boldsymbol{Q}_1}$,  
but commonly $n_{\boldsymbol{L}_1} > n_{\boldsymbol{Q}_1}$ or even $n_{\boldsymbol{L}_1} \gg  n_{\boldsymbol{Q}_1}$. 
It is known that the fill-in $n_f$ not only depends on the graph, but also on the order of the nodes in the graph \citep{rue2005gaussian}.
Thus a re-ordering is usually needed before doing a Cholesky factorization. It is desirable to find an optimal 
or approximately optimal ordering of the graph in order to make the Cholesky factor of $\boldsymbol{Q}_1$ sparser and to save computational resources,
but this is not the focus in this paper. We refer to \cite{rue2005gaussian} for more information on why it is desirable to do re-ordering of the graph of a GMRF.

 \subsection{Orthogonal factorization} \label{sec: cholesky_IGO}
 With an $m \times n$ matrix $\boldsymbol{A}$, the orthogonal factorization of $\boldsymbol{A}$ is 

 \begin{equation}\label{eq: cholesky_Orthogonal_factorization}
 \boldsymbol{A} = \boldsymbol{S} \cdot \boldsymbol{R},
 \end{equation}
 where $\boldsymbol{S} \in \mathbb{R}^{m \times m}$ is an orthogonal matrix  and $\boldsymbol{R} \in \mathbb{R}^{m \times n}$ is an
 upper triangular matrix. We assume without loss of the generality that $m \geq n$.  There exist many algorithms for orthogonal factorization,
 such as the standard Gram - Schmidt algorithm or the modified Gram--Schmidt (MGS) algorithm and the Householder orthogonal factorization. We refer to
 \citet{saad2003iterative} and \citet{bjorck1996numerical} for more algorithms. If $\boldsymbol{A}$ has full  column rank, then the first $n$ columns of $\boldsymbol{S}$ forms an
 orthonormal basis of ran($\boldsymbol{A}$), where ran($\boldsymbol{A}$) denotes the range of $\boldsymbol{A}$

 $$ \ran(\boldsymbol{A}) = \left\{ \boldsymbol{y} \in \mathbb{R}^m: \boldsymbol{y} = \boldsymbol{A} \boldsymbol{x} \ \textrm{for} \ \textrm{some} \
 \boldsymbol{x} \in \mathbb{R}^n \right\}.$$
The orthogonal factorization is usually used to find an orthonormal basis for a matrix. 
The orthogonal factorization has many advantages and some of them are given in what follows. 

 \begin{enumerate}
 \item It is numerically stable and robust both with a Householder orthogonal factorization and with a orthogonal factorization using Givens rotations. If the matrix $\boldsymbol{A}$ is non-singular,
 it always produces  an orthogonal matrix $\boldsymbol{S}$ and an upper triangular matrix $\boldsymbol{R}$ which satisfy Equation \eqref{eq: cholesky_Orthogonal_factorization};
 \item It is easy to solve the linear system of equations $\boldsymbol{A}\boldsymbol{x} = \boldsymbol{b} $ using the upper triangular matrix $\boldsymbol{R}$ since $\boldsymbol{S}$ is an orthogonal matrix;
 \item The normal equation has the form $\boldsymbol{A}^{\mbox{T}}\boldsymbol{A} \boldsymbol{x} = \boldsymbol{A}^{\mbox{T}}\boldsymbol{b}$
       and the normal equation matrix is $\boldsymbol{A}^{\mbox{T}}\boldsymbol{A}$, where $\boldsymbol{A}^{\mbox{T}}$ denotes the transpose of $\boldsymbol{A}$.
       Then the triangular matrix $\boldsymbol{R}$ is the Cholesky factor of the normal equations matrix.  
 \end{enumerate}

\subsection{Givens rotations} \label{sec: cholesky_givens}
 A Givens rotation $G(i,j,\theta) \in \mathbb{R}^{m \times n}$ is an identity matrix $\boldsymbol{I}$ except that 

 \begin{displaymath}
   \begin{split}
   G_{ii} & = c,\hspace{6.5mm}   G_{ij}  = s, \\
   G_{ji} &= -s, \hspace{3mm} G_{jj}  = c.
  \end{split}
 \end{displaymath}
 If $c = \cos(\theta)$ and $s = \sin(\theta)$, then $\boldsymbol{y} = G(i,j,\theta) \cdot \boldsymbol{x}$ rotates $\boldsymbol{x}$ clockwise in the $(i,j)$-plane with
 $\theta$ radians, which gives

 \begin{equation}\label{eq: eq6}
 y_l = \left\{ \begin{array} {ll}
 x_l, & \textrm{ when } l \neq i,j,  \\
 cx_i + sx_j, & \textrm{ when } l =i, \hspace{10mm} (1 \leq l \leq m),  \\
 -sx_i + cx_j, & \textrm{ when } l=j.
 \end{array}  \right.
 \end{equation}
If we want to rotate $\boldsymbol{x}$ counterclockwise in the $(i,j)$-plane with $\theta$ radians, then we can set $c = \cos(\theta)$ and $s = -\sin(\theta)$.
It is obvious from Equation \eqref{eq: eq6} that if

 \begin{displaymath}
 s = \frac{x_j}{\sqrt{x_i^2 + x_j^2}} \  \textrm{and} \  c = \frac{x_i}{\sqrt{x_i^2 + x_j^2}}
 \end{displaymath}
  then $y_j = 0$. So the Givens rotations can set the elements in $\boldsymbol{A}$ to zeros one at a time. This is useful when dealing with sparse matrices. At
  the same time, $c$ and $s$ are the only two values which we need for this algorithm. Givens rotations are suitable for structured least squares
  problems such as the problems at the heart of GMRFs.

\subsection{Incomplete factorization algorithms}
There are many algorithms for incomplete factorizations of matrices, such as the incomplete triangular factorization and the incomplete orthogonal factorization. These algorithms are commonly 
used in practical applications \citep{axelsson1996iterative, meijerink1981guidelines, saad1988preconditioning}. 
The incomplete factorizations usually have the form
\begin{equation}\label{eq: cholesky:imcimplete_factorizations}
 \boldsymbol{A} = \boldsymbol{M}_1 \cdot \boldsymbol{M}_2 + \boldsymbol{E},
 \end{equation}
where $\boldsymbol{E}$ is the error matrix, and $\boldsymbol{M}_1$ and $\boldsymbol{M}_2$ are some well-structured matrices.
The incomplete factorization algorithms are usually associated with dropping strategies. A dropping strategy for an incomplete factorization specifies rules
for when elements of the factors should be dropped. We returns to a detailed discussion on the dropping strategies in Section \ref{sec: cholesky_sparse_factor}.

One of the commonly used incomplete factorization algorithms is the incomplete triangular factorization, and it is also called incomplete LU (ILU) factorization since 
$\boldsymbol{M}_1$ is a \emph{lower} triangular matrix and $\boldsymbol{M}_2$ is an \emph{upper} triangular matrix. This algorithm
 is usually applied to the square matrices, and it uses Gaussian elimination together with a predefined dropping strategy. 
Many incomplete orthogonal factorizations can be used both for square matrices and for rectangular matrices,
and these algorithms usually use the modified Gram-Schmidt procedure together with some dropping strategies in order to return a sparse and generally non-orthogonal matrix $\boldsymbol{S}$ and a sparse
upper triangular matrix $\boldsymbol{R}$. \citet{wang1997cimgs} proved the existence and stability of the associated incomplete orthogonal factorization. 
Incomplete orthogonal factorization using Givens rotations was proposed by \citet{bai2001class}. The main idea of the incomplete orthogonal factorization is to 
use the Givens rotations to zero-out the elements in the matrix one at a time. Some predefined dropping strategies are needed in order to achieve the sparsity pattern for the upper triangular matrix $\boldsymbol{R}$.
This algorithm computes a sparse matrix $\boldsymbol{S}$, which is always an orthogonal matrix, together with a sparse upper triangular matrix $\boldsymbol{R}$.   
Since the matrix $\boldsymbol{S}$ is the product of the Givens rotations matrices, it is always an orthogonal matrix.
The incomplete orthogonal factorization has the form

 \begin{equation}\label{eq: cholesky:imcimplete_orthogonal_factorizations}
 \boldsymbol{A} = \boldsymbol{S} \cdot \boldsymbol{R} + \boldsymbol{E}.
 \end{equation}
  This method was originally described and implemented by \citet{jennings1984incomplete}.
  \citet{saad1988preconditioning} described this incomplete orthogonal factorization with the modified Gram--Schmidt process using some numerical dropping
  strategy. Another version of the incomplete orthogonal factorization is given by \citet{bai2001class} with Givens rotations.
\citet{bai2001class}  claimed that this incomplete algorithm inherited the good properties of the orthogonal factorization. 

  \begin{enumerate}
  \item $\boldsymbol{R}$ is a sparse triangular matrix and $\boldsymbol{S}$ is an orthogonal matrix. 
         \citet{bai2009numerical} pointed out that the sparsity pattern of the upper-triangular part of $\boldsymbol{A}$ is inherited by the incomplete upper triangular matrix $\boldsymbol{R}$.
          They also pointed out that the number of nonzero elements in the upper triangular matrix $\boldsymbol{R}$ is less than the number of nonzero elements in the upper-triangular part of $\boldsymbol{A}$.
  \item The error matrix $\boldsymbol{E} = \boldsymbol{A} - \boldsymbol{S} \cdot \boldsymbol{R}$ is ``small'' in some sense and the size of the errors can be controlled by the pre-defined threshold.
  \item The triangular matrix $\boldsymbol{R}$ is non-singular whenever $\boldsymbol{A}$ is not singular. We can always obtain this triangular matrix in the same way as the orthogonal factorization and 
        $\boldsymbol{R}$ will always be an incomplete Cholesky factor for the normal equation matrix $\boldsymbol{A}^{\mbox{T}} \boldsymbol{A}$.
  \item Another merit of the incomplete orthogonal factorization with Givens rotations is that we do not need to form the corresponding normal matrices $\boldsymbol{S}$ since only the $(c, s)$-pair is needed in
        order to find the upper triangular matrix $\boldsymbol{R}$. More information about the Givens rotations and the $(c, s)$-pairs are given in Section \ref{sec: cholesky_givens}
  \end{enumerate}

\citet{papadopoulos2005class} implemented different versions of the algorithm proposed by \citet{bai2001class}. 
There are two main differences between these versions. The first one is the order in which elements in the matrix $\boldsymbol{A}$ are zeroed out, and the second one is the rules for dropping strategies. 
We refer to \citet{bai2001class} and \citet{papadopoulos2005class} for more information about this algorithm and implementations.
There are also more variations for incomplete orthogonal factorization using Givens rotations, such as \citet{bai2009modified} and  \citet{bai2009numerical}.
\citet{bai2009modified} proposed some modified incomplete orthogonal factorization methods and these algorithms have special storage and sparsity-preserving techniques. 
\citet{bai2009modified} showed a way to adopt a diagonal compensation strategy by reusing the dropped elements.
These dropped elements are added to the main diagonal elements of the same rows in the incomplete upper-triangular matrix $\boldsymbol{R}$. 
\citet{bai2009numerical} proposed practical incomplete Givens orthogonalization (IGO) methods for solving large sparse systems of linear equations.
They claimed that these incomplete IGO methods took the storage requirements, the accuracy of the solutions and the coding of the pre-conditioners into consideration.

  In this report, we have chosen the column-wise threshold incomplete Givens orthogonal (cTIGO) factorization
  algorithm for finding the sparse upper triangular matrix $\boldsymbol{R}$. 
 This sparse upper triangular matrix $\boldsymbol{R}$ has sparse structure and can be used for specifying the GMRFs. 
 The matrix $\boldsymbol{S}$ does not need to be stored in our setting since 
 we only need the upper-triangular matrix $\boldsymbol{R}$. The matrix $\boldsymbol{S}$ only needs computed whenever it is explicitly needed.

\section{Specifying GMRFs using sparse Cholesky factors} \label{sec: cholesky_algorithm}
In this section we begin by introducing the background of GMRFs conditioned on a subset of the variables or on observed data. 
A small example is used to illustrate how the cTIGO algorithm works when applied to GMRFs. 

\subsection{GMRFs conditioned on a subset of the variables} \label{sec: cholesky_conditioningsubsets}
\textbf{I. GMRFs with soft constraint} \newline
Let $\boldsymbol{x}$ be a GMRF and assume that we have observed some linear transformation $\boldsymbol{Ax}$ with additional Gaussian distributed noise
\begin{displaymath}
 \boldsymbol{e}|\boldsymbol{x} \sim \mathcal{N}(\boldsymbol{Ax}, \boldsymbol{\boldsymbol{Q}}_{\epsilon}^{-1}),
\end{displaymath}
where $k$ is the dimension of the vector $\boldsymbol{e}$, $\boldsymbol{A}$ is a $k \times n$ matrix with rank $k$ and $k < n$, and
$ \boldsymbol{Q_{\epsilon}} > 0$ is the precision matrix of $\boldsymbol{e}$.
This is called ``soft constraint" by \citet{rue2005gaussian} and the log-density for the model is 
\begin{equation}\label{eq: logdensity(x|e)}
 \log{\pi(\boldsymbol{x}|\boldsymbol{e})} = -\frac{1}{2}(\boldsymbol{x}^{\mbox{T}} - \boldsymbol{\mu}) \boldsymbol{Q}_1 (\boldsymbol{x} - \boldsymbol{\mu})
                                            -\frac{1}{2}(\boldsymbol{e} - \boldsymbol{Ax})^{\mbox{T}} \boldsymbol{Q_{\epsilon}} (\boldsymbol{e}-\boldsymbol{Ax}) + \text{const},
\end{equation}
where $\boldsymbol{\mu}$  and $\boldsymbol{Q}_1$ are the mean and the precision matrix of the GRMF, respectively, and ``const'' is constant.
If $\boldsymbol{x}$ has mean $\boldsymbol{\mu} = \boldsymbol{0}$ then
\begin{equation}
 \boldsymbol{x}|\boldsymbol{e} \sim \mathcal{N}_c (\boldsymbol{A}^{\mbox{T}}\boldsymbol{Q}_{\boldsymbol{\epsilon}}\boldsymbol{e}, \boldsymbol{Q}_1+\boldsymbol{A}^{\mbox{T}}\boldsymbol{Q}_{\boldsymbol{\epsilon}}\boldsymbol{A}).
\end{equation}
Here we use the canonical form $\mathcal{N}_c (\cdot, \cdot)$ for $\boldsymbol{x}|\boldsymbol{e}$. We refer to \citet[Chapter 2.3.2]{rue2005gaussian} for more information about the canonical form for GMRF.
We can notice that for specifying the GMRFs with  ``soft constraint", the Routine $(2)$ as shown in Figure \ref{fig: cholesky_diag} can be applied since 
$\boldsymbol{Q} = \boldsymbol{Q}_1 + \boldsymbol{A}^{\mbox{T}}\boldsymbol{Q}_{\boldsymbol{\epsilon}}\boldsymbol{A}$ with $\boldsymbol{Q_2} = \boldsymbol{A}^{\mbox{T}}\boldsymbol{Q}_{\boldsymbol{\epsilon}}\boldsymbol{A}$. 
\newline
\newline
\textbf{II. Models with auxiliary variables} \newline
Auxiliary variables are crucial in some models to retrieve GMRF full conditionals. We look at binary regression models with auxiliary variables.

Assume that we have Bernoulli observational model for binary responses. The binary responses have latent parameters which is a GMRF $\boldsymbol{x}$, and the GMRF usually 
depends on some hyperparameters $\boldsymbol{\theta}$. We usually choose the logit or probit models in this case, where
\begin{equation}
 y_i \sim \mathcal{B} \left( \eta^{-1}(z_i^{\mbox{T}}\boldsymbol{x}) \right), \hspace{3mm} i = 1,2,\dots, m
\end{equation}
 where $\mathcal{B}(p)$ denotes a Bernoulli distribution with probability $p$ for $1$ and $1-p$ for $0$. $\boldsymbol{z}_i$ is a vector of covariates and we assume it is fixed.
$\eta(\cdot)$ is a link function
\begin{equation}
 \eta(p) = \begin{cases}
            \log\left(p/(1-p)\right) & \text{ for logit link}  \\
            \Phi(p)       & \text{ for probit link}   
           \end{cases}
\end{equation}
where $\Phi(\cdot)$ denotes the cumulative distribution function (CDF) for standard Gaussian distribution.
We can use models with auxiliary variables $\boldsymbol{\omega} = (\omega_1, \omega_2, \dots, \omega_m)$ to represent these models,
\begin{displaymath}
 \begin{split}
 \epsilon_i & \overset{iid} \sim G(\epsilon_i), \\
 \omega_i & = \boldsymbol{z}_i^{\mbox{T}}\boldsymbol{x} +\epsilon_i, \\
  y_i & = \begin{cases}
        1, &  \text{ if } \omega_i > 0,    \\
        0, &  \text{ otherwise},
        \end{cases}
 \end{split}
\end{displaymath}
where $G(\cdot)$ is the CDF of standard logistic distribution in the logit case and $G(\cdot) = \Phi(\cdot)$ in the probit case. We refer to \citet[Chapter $28$]{forbes2011statistical} for more information about 
the standard logistic distribution and its CDF.  
Let $\boldsymbol{x}|\boldsymbol{\theta}$ be a GMRF of dimension $n$ with mean $\boldsymbol{\mu} = \boldsymbol{0}$, and assume that we have $\boldsymbol{z}_i^{\mbox{T}}\boldsymbol{x} = x_i$ and $m = n$. With the probit link 
the posterior distribution is 
\begin{equation}
\pi(\boldsymbol{x},\boldsymbol{\omega},\boldsymbol{\theta} |\boldsymbol{y}) \propto \pi(\boldsymbol{\theta})\pi(\boldsymbol{x}|\boldsymbol{\theta})\pi(\boldsymbol{\omega}|\boldsymbol{x})\pi(\boldsymbol{y}|\boldsymbol{\omega}).
\end{equation}
The conditional distribution of $\boldsymbol{x}$ given the auxiliary variables can then be obtained
\begin{displaymath}
 \pi(\boldsymbol{x}|\boldsymbol{\theta},\boldsymbol{\omega}) \propto \exp \left( -\frac{1}{2}\boldsymbol{x}^{\mbox{T}}\boldsymbol{Q}_1(\boldsymbol{\theta})\boldsymbol{x}-\frac{1}{2}\sum_i(x_i-\omega_i) \right)^2,
\end{displaymath}
and this can be written in the canonical form 
\begin{displaymath}
 \boldsymbol{x}|\boldsymbol{\theta},\boldsymbol{\omega} \sim \mathcal{N}_c(\omega, \boldsymbol{Q}_1(\boldsymbol{\theta}) + \boldsymbol{I}).
\end{displaymath}
A general form for the conditional distribution of $\boldsymbol{x}$ given the auxiliary variables, for this binomial model with a probit link function, is given as
\begin{displaymath}
  \boldsymbol{x}|\boldsymbol{\theta},\boldsymbol{\omega} \sim \mathcal{N}_c(\boldsymbol{Z}^{\mbox{T}} \omega, \boldsymbol{Q}_1(\boldsymbol{\theta}) + \boldsymbol{Z}^{\mbox{T}}\boldsymbol{Z}),
\end{displaymath}
where $\boldsymbol{Z}$ is an $m \times n$ matrix.
Similarly, the conditional distribution of $\boldsymbol{x}$ given the auxiliary variables for the logistic regression model can be written as
\begin{displaymath}
  \boldsymbol{x}|\boldsymbol{\theta},\boldsymbol{\omega} \sim \mathcal{N}_c(\boldsymbol{Z}^{\mbox{T}} \boldsymbol{\Lambda} \omega, \boldsymbol{Q}_1(\boldsymbol{\theta}) + \boldsymbol{Z}^{\mbox{T}} \boldsymbol{\Lambda} \boldsymbol{Z}),
\end{displaymath}
where $\boldsymbol{\Lambda} = \diag(\boldsymbol{\lambda})$, and $\lambda_i$ is from the model specification. See more discussions on these models in \citet[Chapter $4.3$]{rue2005gaussian}.

In all the examples in this section, the models are suitable for use Routine (2) in Figure \ref{fig: cholesky_diag} to find the sparse Cholesky factors of the precision matrices of GMRFs.

\subsection{GMRFs conditioned on data} \label{sec: cholesky_GMRFs_conditioning_data}
As mentioned in Section ~\ref{sec: cholesky_GMRFs}, if a vector $\boldsymbol{x}$ is a GMRF with precision matrix $\boldsymbol{Q}_1$ and mean vector $\boldsymbol{\mu}$, then
the density of the vector is given by Equation \eqref{eq: cholesky_density_function}. In practical applications it is common to set $\boldsymbol{\mu} = \boldsymbol{0}$ \citep{rue2005gaussian,gneitingmatern},
which gives the probability density function 
\begin{equation} \label{eq: cholesky_GMRF_zeromean}
\pi(\boldsymbol{x}|\boldsymbol{\theta})
= \left(\frac{1}{2\pi} \right)^{n/2}|\boldsymbol{Q}_1(\boldsymbol{\theta})|^{1/2}\exp\left(-\frac{1}{2}\boldsymbol{x}^{\mbox{T}}\boldsymbol{Q}_1(\boldsymbol{\theta})\boldsymbol{x}\right).
\end{equation}
Assume that the data are of dimension $k$ and defined as a $k$-dimensional random vector

\begin{displaymath}
 \boldsymbol{y}|\boldsymbol{x}, \boldsymbol{\theta} \sim \mathcal{N} \left( \boldsymbol{Ax}, \boldsymbol{Q}_{\boldsymbol{\epsilon}}^{-1} \right)
\end{displaymath}
and has the probability density function

\begin{equation} \label{eq: cholesky_density_data}
  \pi(\boldsymbol{y|\boldsymbol{x},\boldsymbol{\theta}}) = \left(\frac{1}{2\pi} \right) ^{k}|\boldsymbol{Q}_{\boldsymbol{\epsilon}}|^{1/2}
               \exp \left(-\frac{1}{2} (\boldsymbol{y}-\boldsymbol{Ax})^{\mbox{T}} \boldsymbol{Q}_{\boldsymbol{\epsilon}} (\boldsymbol{y}-\boldsymbol{Ax}) \right),
\end{equation}
where $\boldsymbol{A}$ is a $k \times n$ matrix used to select the data location. The precision matrix $\boldsymbol{Q}_{\boldsymbol{\epsilon}}$ for the noise process is a positive definite matrix with dimension $k \times k$. 
Notice that the density function $\pi(\boldsymbol{y|\boldsymbol{x},\boldsymbol{\theta}})$ is not dependent on the $\boldsymbol{\theta}$, and hence the probability density function $\pi(\boldsymbol{y|\boldsymbol{x},\boldsymbol{\theta}})$
can be written as $\pi(\boldsymbol{y|\boldsymbol{x}})$.
The probability density function of ${\boldsymbol{x}|\boldsymbol{y}, \boldsymbol{\theta}}$ can be found from Equations \eqref{eq: cholesky_GMRF_zeromean} and \eqref{eq: cholesky_density_data} through
\begin{equation} \label{eq: cholesky_x|y,theta}
\begin{split}
 \pi(\boldsymbol{x}|\boldsymbol{y}, \boldsymbol{\theta}) & \propto \pi({\boldsymbol{x}, \boldsymbol{y} | \boldsymbol{\theta}})   \\
          & = \pi(\boldsymbol{x}|\boldsymbol{\theta}) \pi(\boldsymbol{y}|\boldsymbol{x}, \boldsymbol{\theta}) \\
          & \propto \exp \left( -\frac{1}{2} \left[ x^{\mbox{T}} (\boldsymbol{Q}_1(\boldsymbol{\theta})
            + \boldsymbol{A}^{\mbox{T}} \boldsymbol{Q}_{\boldsymbol{\epsilon}} \boldsymbol{A}) \boldsymbol{x} 
           - 2\boldsymbol{x}^{\mbox{T}} \boldsymbol{A}^{\mbox{T}}\boldsymbol{Q}_{\boldsymbol{\epsilon}} \boldsymbol{y} \right] \right).
 \end{split}
\end{equation}
Similarly, the density function \eqref{eq: cholesky_x|y,theta} can be written in the canonical form as

\begin{equation} \label{eq: cholesky_canonical_form}
 {\boldsymbol{x}| \boldsymbol{y}, \boldsymbol{\theta}} \sim \mathcal{N} \left( \boldsymbol{\mu}_c (\boldsymbol{\theta}), \boldsymbol{Q}_c (\boldsymbol{\theta}) \right).
\end{equation}
where $\boldsymbol{\mu}_c (\boldsymbol{\theta}) =  \boldsymbol{Q}_c(\boldsymbol{\theta})^{-1} \boldsymbol{A}^{\mbox{T}} \boldsymbol{Q}_{\boldsymbol{\epsilon}} \boldsymbol{y} $, 
and $\boldsymbol{Q}_c (\boldsymbol{\theta}) = \boldsymbol{Q}_1(\boldsymbol{\theta}) + \boldsymbol{A}^{\mbox{T}} \boldsymbol{Q}_{\boldsymbol{\epsilon}} \boldsymbol{A}$. Now we can notice that the precision matrix for the GMRF conditional on data has the form 
$\boldsymbol{Q} = \boldsymbol{Q}_1(\boldsymbol{\theta}) + \boldsymbol{Q}_2 $ with $\boldsymbol{Q}_2 = \boldsymbol{A}^{\mbox{T}} \boldsymbol{Q}_{\boldsymbol{\epsilon}} \boldsymbol{A}$, where $\boldsymbol{Q}_2 $ does not
depend on $\boldsymbol{\theta}$. Since $\boldsymbol{Q}$ has the same form as given in Routine ($2$) in Figure \ref{fig: cholesky_diag}, 
it is possible to use the proposed routine to find the sparse Cholesky factor of the precision matrix of the GMRF conditioned on data.

Even though it is not the focus of this paper, it might be useful to point out that using Equations \eqref{eq: cholesky_GMRF_zeromean} - \eqref{eq: cholesky_x|y,theta}, we can find the analytical formula for the posterior
density function of $(\boldsymbol{\theta} | \boldsymbol{y})$ through Bayes' formula. It is given by
\begin{equation} \label{eq: cholesky_log(theta|y)}
\begin{split}
 \log(\pi(\boldsymbol{\theta}|\boldsymbol{y})) = & \text{ const.} + \log(\pi(\boldsymbol{\theta})) + \frac{1}{2}\log(|\boldsymbol{Q}_1(\boldsymbol{\theta})|) \\
             & - \frac{1}{2}\log(|\boldsymbol{Q}_c(\boldsymbol{\theta})|) + \frac{1}{2} \boldsymbol{\mu}_c(\boldsymbol{\theta})^{\mbox{T}} \boldsymbol{Q}_c(\boldsymbol{\theta}) \boldsymbol{\mu}_c(\boldsymbol{\theta}).
\end{split}
\end{equation}
We refer to \citet{hu2012multivariate} for detailed information about this log-posterior density function. The log-posterior density function  $ \log(\boldsymbol{\theta} | \boldsymbol{y})$ 
is crucial when doing statistical inference in Bayesian statistics.

The sparse structure of $\boldsymbol{Q}_2$ depends both on the structures of $\boldsymbol{A}$ and of $\boldsymbol{Q}_{\boldsymbol{\epsilon}}$. 
In most cases, the $\boldsymbol{Q}_{\boldsymbol{\epsilon}}$ is a diagonal matrix and the matrix $\boldsymbol{A}$ has sparse structure.
Therefore $\boldsymbol{Q}_2$ should also have a sparse structure. When the observations are conditional independent, 
but have a non-Gaussian distribution, then we can use a GMRF approximation to obtain a sparse structure of $\boldsymbol{Q}_2$ as presented in Section \ref{sec: cholesky_GMRF_approximation}.

\subsection{GMRF approximation} \label{sec: cholesky_GMRF_approximation}
Suppose there are $n$ conditionally independent observations $y_1,y_2, \ldots, y_n$ from a non-Gaussian distribution and that $y_i$ is an
indirect observation of $x_i$. $\boldsymbol{x}$ is a GMRF with mean $\boldsymbol{\mu} =\boldsymbol{0}$ and precision matrix $\boldsymbol{Q}_1$.
The full conditional $\pi(\boldsymbol{x}|\boldsymbol{y}, \boldsymbol{\theta})$ then has the form

\begin{equation}\label{eq: cholesky_nongaussian_data}
\pi \left(\boldsymbol{x}|\boldsymbol{y}, \boldsymbol{\theta} \right) \propto \exp\left(-\frac{1}{2}\boldsymbol{x}^{\mbox{T}} {\boldsymbol{Q}_1} \boldsymbol{x} + \sum_{i=1}^{n}\log{\pi(y_i|x_i)} \right).
\end{equation}
Apply a second-order Taylor expansion of $\sum_{i=1}^{n} \log{\pi(y_i|x_i)}$ around $\boldsymbol{\mu}_0$. In other words, construct a suitable GMRF proposal density
$\widetilde{\pi}(\boldsymbol{x}|\boldsymbol{y}, \boldsymbol{\theta})$

\begin{equation} \label{eq: cholesky_GMRF_proposal_density}
 \begin{split}
\widetilde{\pi}(\boldsymbol{x}|\boldsymbol{y}, \boldsymbol{\theta}) & \propto \exp\left( -\frac{1}{2}{\boldsymbol{x}}^{\mbox{T}} {\boldsymbol{Q}_1(\boldsymbol{\theta})}{\boldsymbol{x}} + \sum_{i=1}(a_i + b_i x_i -\frac{1}{2}c_ix_i^2) \right) \\
 & \propto \exp\left( -\frac{1}{2}\boldsymbol{x}^{\mbox{T}} \left({\boldsymbol{Q}_1(\boldsymbol{\theta}) + \diag(\boldsymbol{c})}\right) \boldsymbol{x} + \boldsymbol{b}^{\mbox{T}}  {\boldsymbol{x}} \right).
\end{split}
\end{equation}

\noindent $c_i$ should set to zero when $c_i < 0$. $\boldsymbol{b}$ and $\boldsymbol{c}$ depend on $\boldsymbol{\mu}_0$.
 The canonical parametrization of $\widetilde{\pi}(\boldsymbol{x}|\boldsymbol{y}, \boldsymbol{\theta})$ has the form

\begin{displaymath}
\mathcal{N}_{c}\left(\boldsymbol{b}, \boldsymbol{Q}_1(\boldsymbol{\theta}) + \diag(\boldsymbol{c})\right).
\end{displaymath}
In this case $\boldsymbol{Q}_2$ has a diagonal structure. An important feature of \eqref{eq: cholesky_GMRF_proposal_density} is that it inherits the \emph{Markov}
property of the prior on $\boldsymbol{x}$, which is useful for sampling GMRF. When $\boldsymbol{\mu} \neq \boldsymbol{0}$, the canonical parametrization of the 
$(\boldsymbol{x}|\boldsymbol{y},\boldsymbol{\theta})$ is changed to 
\begin{displaymath}
\mathcal{N}_{c}\left(\boldsymbol{Q} \boldsymbol{\mu} + \boldsymbol{b}, \boldsymbol{Q}_1(\boldsymbol{\theta}) + \diag(\boldsymbol{c})\right),
\end{displaymath}
and does not change the matrix $\boldsymbol{Q} = \boldsymbol{Q}_1 +\boldsymbol{Q}_2$.

As it was pointed out in Section \ref{sec: cholesky_GMRFs}, to sample from the GMRFs, the Cholesky factor $\boldsymbol{L}$ is one of most important factors. In order to save
computational resources, a sparse Cholesky factor is preferable if the approximated precision matrix is ``close'' to the
original precision matrix, where ``close'' means both in structure and the elements.

\subsection{Theoretical background} \label{sec: cholesky_algorithm_orthogonal}
It has been mentioned in Section \ref{sec: cholesky_GMRFs} that the sparsity pattern of the Cholesky factor is determined by the graph
$\mathcal{G}$, and it is unnecessary to calculate the zero elements in the Cholesky factor. 
In this section, we are going to introduce the theoretical background for finding the Cholesky factor from the orthogonal factorization when the GMRF is conditioned on 
observed data or a subset of the variables.

Let $\boldsymbol{y}$ be the observed data and assume $ \boldsymbol{y} = (y_1, y_2, \ldots, y_n)$ has the Gaussian distribution, then the density of $\boldsymbol{x}$ conditioned on $\boldsymbol{y}$ has the form in 
\eqref{eq: cholesky_x|y,theta}. In the discussed situations in Section \ref{sec: cholesky_conditioningsubsets} - Section \ref{sec: cholesky_GMRF_approximation},
the precision matrix $\boldsymbol{Q}$ can be split into two parts, the precision matrix $\boldsymbol{Q}_1$ of the GMRF  $(\boldsymbol{x}|\boldsymbol{\theta})$ and the matrix $\boldsymbol{Q}_2$
which is the additional effect. The matrix $\boldsymbol{Q}_2$ is usually a diagonal matrix or another type of sparse matrix. 
If the data is not Gaussian distributed, then we can apply the GMRFs approximation given in \eqref{eq: cholesky_GMRF_proposal_density} and it returns the precision matrix $\boldsymbol{Q}_1$ 
with a diagonal matrix $\boldsymbol{Q}_2$ added. This structure satisfies the Routine (2) in Figure \ref{fig: cholesky_diag}. 

Let $\boldsymbol{Q} = \boldsymbol{Q}_1 + \boldsymbol{Q}_2$ 
and assume that the Cholesky factors for $\boldsymbol{Q}_1$, $\boldsymbol{Q}_2$ and $\boldsymbol{Q}$ are $\boldsymbol{L}_1$, $\boldsymbol{L}_2$ and $\boldsymbol{L}$, respectively. 
The Cholesky factors $\boldsymbol{L}_1$ and $\boldsymbol{L}_2$ are assumed to be known. We have the following results.

 \begin{obs} \label{thm_Cholesky}
  Let $\boldsymbol{x} \in \mathbb{R}^n$  be a zero mean GMRF with precision matrix £$\boldsymbol{Q}_1$. Assume that 
  the precision matrix has the form $\boldsymbol{Q} = \boldsymbol{Q}_1 + \boldsymbol{Q}_2 \in \mathbb{R}^{n\times n}$
  when conditioned on observed data or a subset of the variables. 
  Let the Cholesky factors for $\boldsymbol{Q}_1$ and $\boldsymbol{Q}_2$ be $\boldsymbol{L}_1$ and $\boldsymbol{L}_2$, respectively. Form
  $$\boldsymbol{A} =  \begin{pmatrix} \boldsymbol{L}_1^{\mbox{T}}\\\boldsymbol{L}_2^{\mbox{T}}\end{pmatrix}.$$
  Then ${\boldsymbol{A}^{\mbox{T}}} {\boldsymbol{A}}$ is the precision matrix $\boldsymbol{Q}$.
 \end{obs}

 \begin{proof}
 ${\boldsymbol{A}^{\mbox{T}}} {\boldsymbol{A}} =
 {\begin{pmatrix} \boldsymbol{L}_1 \ \boldsymbol{L}_2 \end{pmatrix}} {\begin{pmatrix}
 \boldsymbol{L}_1^{\mbox{T}} \\ \boldsymbol{L}_2^{\mbox{T}} \end{pmatrix}} =\boldsymbol{L}_1 \boldsymbol{L}_1^{\mbox{T}} + \boldsymbol{L}_2 \boldsymbol{L}_2^{\mbox{T}} = \boldsymbol{Q}_1 +
 \boldsymbol{Q}_2 = \boldsymbol{Q}$.
 \end{proof}

From Observation \ref{thm_Cholesky} the following corollaries are established. Sketched proofs for these corollaries are given.
We refer to \citet{simpson2008krylov} for numerical examples with Corollary \ref{cor_sample}.

 \begin{cor} \label{cor_sample}
 Let $\boldsymbol{X}$ be a zero mean GMRF with precision matrix $\boldsymbol{Q} = \boldsymbol{Q}_1 + \boldsymbol{Q}_2 \in \mathbb{R}^{n\times n}$, and
 let $\boldsymbol{A}$ have the form given in Observation \ref{thm_Cholesky}. 
 Let $\boldsymbol{z} \in \mathbb{R}^{2n}$ be a vector of independent and identically distributed (i.i.d.) standard Gaussian random variables.
 Then the solution  of the least squares problem

  \begin{equation} \label{eq10}
  \boldsymbol{x} = \mathop{\arg \min}_{y} \norm{ \boldsymbol{A} \boldsymbol{y} - \boldsymbol{z}}_2
  \end{equation}

\noindent  is a sample from the GMRF $\boldsymbol{X}$.
  \end{cor}

  \begin{proof}
  $\boldsymbol{Q} = \boldsymbol{A}^{\mbox{T}}\boldsymbol{A}$ from Observation \ref{thm_Cholesky} is the starting
  point to prove this Corollary.  Denote $\boldsymbol{A}^\dag$ the Moore-Penrose pseudo-inverse of $\boldsymbol{A}$, and then the solution to the
  least squares problem is $\boldsymbol{x} = \boldsymbol{A}^\dag \boldsymbol{z}$ \citep{bjorck1996numerical}.  
  From the definition of the pseudo-inverse, $\boldsymbol{x} = \boldsymbol{W}\boldsymbol{S}^\dag \boldsymbol{U}^{\mbox{T}}\boldsymbol{z}$,
  where $\boldsymbol{A} = \boldsymbol{U}\boldsymbol{S}\boldsymbol{W}^{\mbox{T}}$ is a singular value decomposition of $\boldsymbol{A}$ and $\boldsymbol{S}^\dag\in \mathbb{R}^{d\times 2d}$ is the matrix with the
  reciprocals of the non-zero singular values on the diagonal. We can verify that $\boldsymbol{x}$ has the required distribution, and it is
  sufficient to check the first two moments since $\boldsymbol{x}$ has a Gaussian distribution, being linear in $\boldsymbol{z}$.  It is clear that $\mathbb{E}(\boldsymbol{x}) = 0$. Furthermore,
  \begin{align*}
  \mathbb{E}(\boldsymbol{x}\boldsymbol{x}^{\mbox{T}}) &= \boldsymbol{A}^\dag \mathbb{E}(\boldsymbol{z}\boldsymbol{z}^{\mbox{T}}) (\boldsymbol{A}^\dag)^{\mbox{T}}\\
        &=\boldsymbol{W}\boldsymbol{S}^\dag \boldsymbol{U}^{\mbox{T}}\boldsymbol{U}(\boldsymbol{S}^\dag)^{\mbox{T}} \boldsymbol{W}^{\mbox{T}} \\
        &=\boldsymbol{W}\boldsymbol{S}^\dag(\boldsymbol{S}^\dag)^{\mbox{T}} \boldsymbol{W}^{\mbox{T}} \\
        &= \boldsymbol{W} (\boldsymbol{S}^{\mbox{T}}\boldsymbol{S})^\dag \boldsymbol{W}^{\mbox{T}}.
  \end{align*} Calculations yield $\boldsymbol{Q} = \boldsymbol{W}\boldsymbol{S}^{\mbox{T}}\boldsymbol{S}\boldsymbol{W}^{\mbox{T}}$ and, hence,
  $\boldsymbol{x}\sim MVN (0,\boldsymbol{Q}^\dag)$.
  \end{proof}

  Therefore, it is possible to sample from a GMRF  by solving the sparse least squares problem given in \eqref{eq10} with some conditions on GMRFs.

  \begin{cor} \label{cholesky_rectangle}
  The upper triangular matrix $\boldsymbol{R}$ from the orthogonal factorization of the rectangular matrix
  \begin{equation}  \label{cholesky_A_matrix}
   \boldsymbol{A} =  \begin{pmatrix} \boldsymbol{L}_1^{\mbox{T}}\\\boldsymbol{L}_2^{\mbox{T}}\end{pmatrix}
  \end{equation}
  is the Cholesky factor of the precision matrix $\boldsymbol{Q} = \boldsymbol{Q}_1 + \boldsymbol{Q}_2$.
  \end{cor}

  \begin{proof}
  Since the upper triangular matrix from the orthogonal factorization is the Cholesky factor for the normal equations matrix,
  this is obvious from Observation \ref{thm_Cholesky}.
  \end{proof}

 By using the orthogonal factorization of the rectangular matrix $\boldsymbol{A}$, it is possible to get samples from the GMRFs when they are conditioned on data or a subset of the variables by using
 Corollary \ref{cor_sample} or the Cholesky factor from Observation \ref{thm_Cholesky} together with the sampling algorithms discussed in \citet[Chapter 2]{rue2005gaussian}.

\subsection{Dropping strategies}  \label{sec: cholesky_sparse_factor}
In this section the dropping strategy for the incomplete orthogonal factorization is introduced in order to find the incomplete Cholesky factor for matrix $\boldsymbol{A}^{\mbox{T}}\boldsymbol{A}$. 
Together with some dropping strategy for the incomplete orthogonal factorization of the rectangular matrix $\boldsymbol{A}$, a sparse upper triangular matrix $\boldsymbol{R}$ can be
obtained. From Corollary \ref{cholesky_rectangle} and the discussion in Section \ref{sec: cholesky_IGO}, we know that 
$\boldsymbol{R}$ is an incomplete Cholesky factor or sparse Cholesky factor for the precision matrix $\boldsymbol{Q}$.
This sparse Cholesky factor can then be used to specify the GMRF. 
The dropping strategies are important when doing the incomplete orthogonal factorization. Generally speaking, there are two kinds of dropping strategies.
\begin{enumerate}
 \item Drop fill-ins based on sparsity patterns. Before doing the incomplete orthogonal factorization, the sparsity pattern of the
 upper triangular matrix is predefined and fixed. If the factorization based only on the sparsity pattern of the original matrix, we
 drop all the elements which are pre-defined to be zeros. The algorithm does not consider the actual numerical values of the elements during the factorizations.
 \item Drop fill-ins by using a numerical threshold. This  strategy only includes the elements in $\boldsymbol{R}$ if they are bigger than a predefined
 threshold value. \citet{munksgaard1980solving} presented one way to select the value of the threshold parameter. His strategy drops
 the elements which are smaller than the diagonal elements of their rows and columns, multiplied by some predefined small value (called dropping tolerance). In this report,
 a slightly different dropping strategy is chosen. During the incomplete orthogonal factorization using Givens rotations,
 or the column-wise threshold incomplete Givens orthogonal (cTIGO) factorization \citep{papadopoulos2005class}, 
 we drop the elements according to their magnitudes with some predefined dropping tolerance.
 The nonzero pattern of $\boldsymbol{R}$ is determined dynamically.
 \end{enumerate}

 Both the fixed sparsity pattern strategy and the dynamic strategy are useful in applications. The fixed sparsity pattern strategy is the candidate when the computation resources are low.
 It is usually faster but sometimes returns unsatisfactory results. The dynamic strategy will in most cases return satisfactory results by choosing proper dropping tolerances
 but it is usually more expensive both in time and computations.

 There are different versions of orthogonal factorizations. We refer to \citet{saad2003iterative}, \citet{golub1996matrix} and \citet{trefethen1997numerical} for more information.
  Based on the research of \citet{bai2001class}, \citet{papadopoulos2005class} and \citet{bai2009numerical}, we choose the incomplete orthogonal factorization using Givens rotations 
  to find the sparse Cholesky factor. This algorithm is stable and robust and always returns a sparse matrix. This algorithm inherits the advantages of orthogonal 
  factorization. \citet{bai2001class} commented that there is little attention given to incomplete orthogonal factorization with Givens rotations, which is actually useful in many numerical problems. 

  In order to use Givens rotations for incomplete orthogonal factorization, the following nonzero patterns needs to be defined,

  \begin{equation*}
  \begin{aligned}
  N_{\boldsymbol{Q}} &= \{(i,j)~| ~Q_{ij} \neq 0, 1 \leq {i,j} \leq n\}, \\
  N_{\boldsymbol{Q},l} &= \{(i,j)~| ~Q_{ij} \neq 0, i \geq j , 1 \leq {i,j} \leq n\}, \\
  N_{\boldsymbol{Q},u} &= \{(i,j)~| ~Q_{ij} \neq 0, i \leq j , 1 \leq {i,j} \leq n\}, \\
  N_{\boldsymbol{L}_1} &= \{(i,j)~| ~L_{1_{ij}} \neq 0, 1 \leq {i,j} \leq n\}, \\
  N_{\boldsymbol{L}_2} &= \{(i,j)~| ~L_{2_{ij}} \neq 0, 1 \leq {i,j} \leq n\}, \\
  N_{\boldsymbol{L}} &= \{(i,j)~| ~L_{ij} \neq 0, 1 \leq {i,j} \leq n\}, \\
  N_{\boldsymbol{A}} &= \{(i,j)~| ~A_{ij} \neq 0, 1 \leq {i,j} \leq n\}, \\
  N_{\boldsymbol{R}} &= \{(i,j)~| ~R_{ij} \neq 0, 1 \leq {i,j} \leq n\}, \\
  \end{aligned}
  \end{equation*}

 \noindent  where $N_{\boldsymbol{Q}}$ is the nonzero pattern of the matrix $\boldsymbol{Q}$, and $N_{\boldsymbol{Q},u}$, $N_{\boldsymbol{Q},l}$ are the nonzero patterns of the
  upper and lower triangular parts of the matrix $\boldsymbol{Q}$, respectively. $N_{\boldsymbol{L}_1}$, $N_{\boldsymbol{L}_2}$, $N_{\boldsymbol{L}}$ $N_{\boldsymbol{A}}$ and $N_{\boldsymbol{R}}$
  are the nonzero patterns of the lower triangular matrix $\boldsymbol{L}_1$, the lower triangular matrix $\boldsymbol{L}_2$, the lower triangular matrix $\boldsymbol{L}$, the matrix $\boldsymbol{A}$  and the matrix $\boldsymbol{R}$, respectively. 
  These matrices are already formulated in previous sections.

  In order to use the cTIGO algorithm, the rectangular matrix $\boldsymbol{A}$ in \eqref{cholesky_A_matrix} is formed. The sparsity pattern of  matrix $\boldsymbol{A}$
  is already known beforehand. However, since the dynamic strategy is chosen, there will be some fill-in during Givens rotations process, and the sparsity pattern of the 
  sparse Cholesky factor $\boldsymbol{R}$ will depend on the dropping tolerance and usually $N_{\boldsymbol{R}} < N_{\boldsymbol{A}}$. 
  For more information about cTIGO algorithm, see \citet{bai2001class} for theoretical issues and \citet{papadopoulos2005class} for implementations.

\subsection{A small example} \label{sec: cholesky_example}
A small example is explored in this section to illustrate how to use the cTIGO algorithm to find the sparse Cholesky factor $\boldsymbol{R}$. 
For simplicity and without loss of generality, we assume that $\boldsymbol{Q}_1$ is the precision matrix for a zero mean GMRF $\boldsymbol{x} \sim \mathcal{N}(\boldsymbol{0}, \boldsymbol{Q}_1)$, and that the data are normally distributed, i.e.,
$\boldsymbol{y} \sim \mathcal{N}(\boldsymbol{0}, \boldsymbol{I})$ and hence $\boldsymbol{Q}_2 = \boldsymbol{I}$. Assume that these matrices are given as follows

  \begin{displaymath}
  \boldsymbol{Q_1} =
  \left( \begin{array}{cccccc}
   5 & -1 & 0 & \ldots & 0 &-1 \\
  -1 & 5  & -1 & 0 & \ldots & 0 \\
   0 & -1 & 5 & -1 & \ldots & 0 \\
   \vdots & & \ddots & \ddots & & \vdots \\
   -1 & 0 & \ldots & & -1 & 5
  \end{array} \right) _{9 \times 9}
  \end{displaymath}

  \noindent and

  \begin{displaymath}
  \boldsymbol{Q}_2 = \boldsymbol{I}_{9 \times 9} =
  \left( \begin{array}{ccccc}
  1 & 0 & 0 & \ldots & 0 \\
  0 & 1 & 0 & \ldots & 0 \\
  \vdots & \ddots & \ddots & \ddots\\
  0 & \ldots\ & & & 1
  \end{array} \right) _{9 \times 9}.
  \end{displaymath}

  \noindent Let $\boldsymbol{L}_1$ and $\boldsymbol{L}_2$  denote the Cholesky factor of the two matrices $\boldsymbol{Q}_1$ and $\boldsymbol{Q}_2$, respectively,  
  with the sparsity patterns given in Figure \ref{fig: cholesky_small_L1} and Figure \ref{fig: cholesky_small_L2}.
  The rectangular matrix $\boldsymbol{A}$ can then be formed as given in \eqref{cholesky_A_matrix} with the sparsity pattern given in Figure \ref{fig: cholesky_small_A}. 
  Apply the cTIGO algorithm to the rectangular matrix $\boldsymbol{A}$ with a dropping tolerance  of $0.0001$ to find the sparse incomplete Cholesky factor $\boldsymbol{R}$.
  The sparsity pattern of $\boldsymbol{R}$ is given in Figure \ref{fig: cholesky_small_R}. The sparsity pattern of the Cholesky factor $\boldsymbol{L}$ from the standard Cholesky factorization
  of the precision matrix $\boldsymbol{Q}$ is given in Figure \ref{fig: cholesky_small_L}. 

  \begin{figure}[htb]
    \centering
    \subfigure[]{ \includegraphics[width=0.3\textwidth,height=0.3\textwidth]{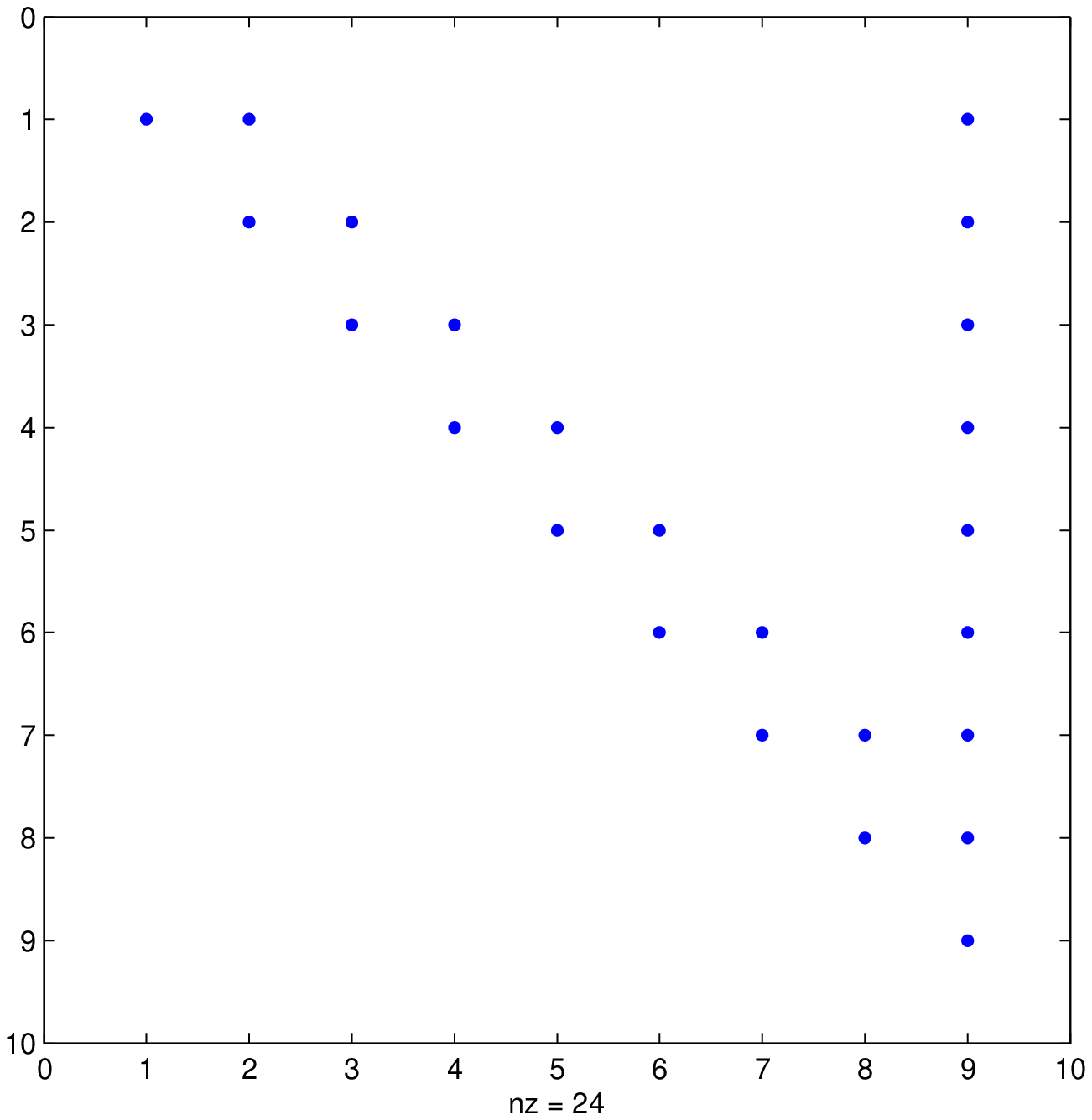} \label{fig: cholesky_small_L1} }
    \subfigure[]{ \includegraphics[width=0.3\textwidth,height=0.3\textwidth]{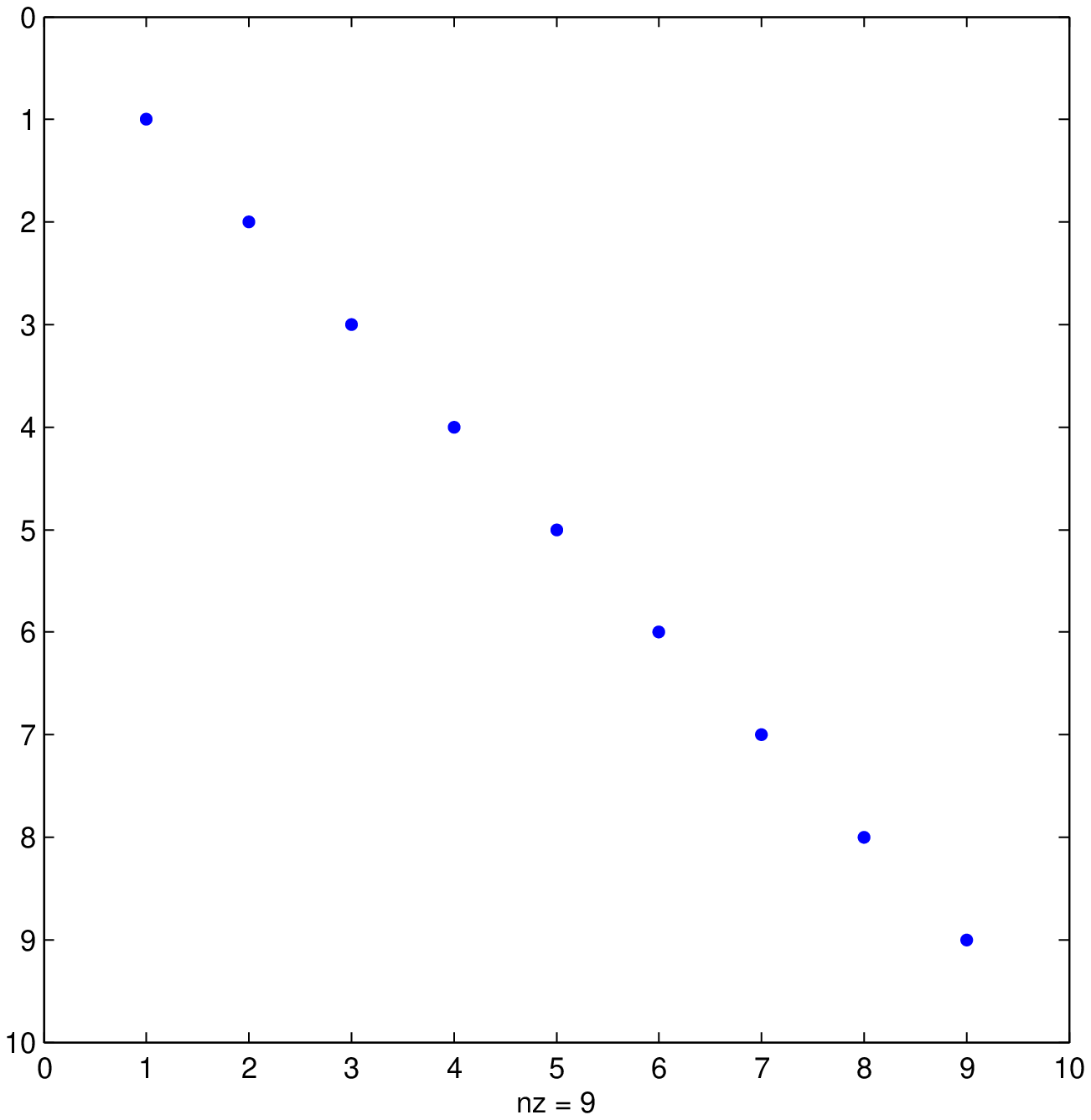} \label{fig: cholesky_small_L2} }
    \subfigure[]{ \includegraphics[width=0.3\textwidth,height=0.3\textwidth]{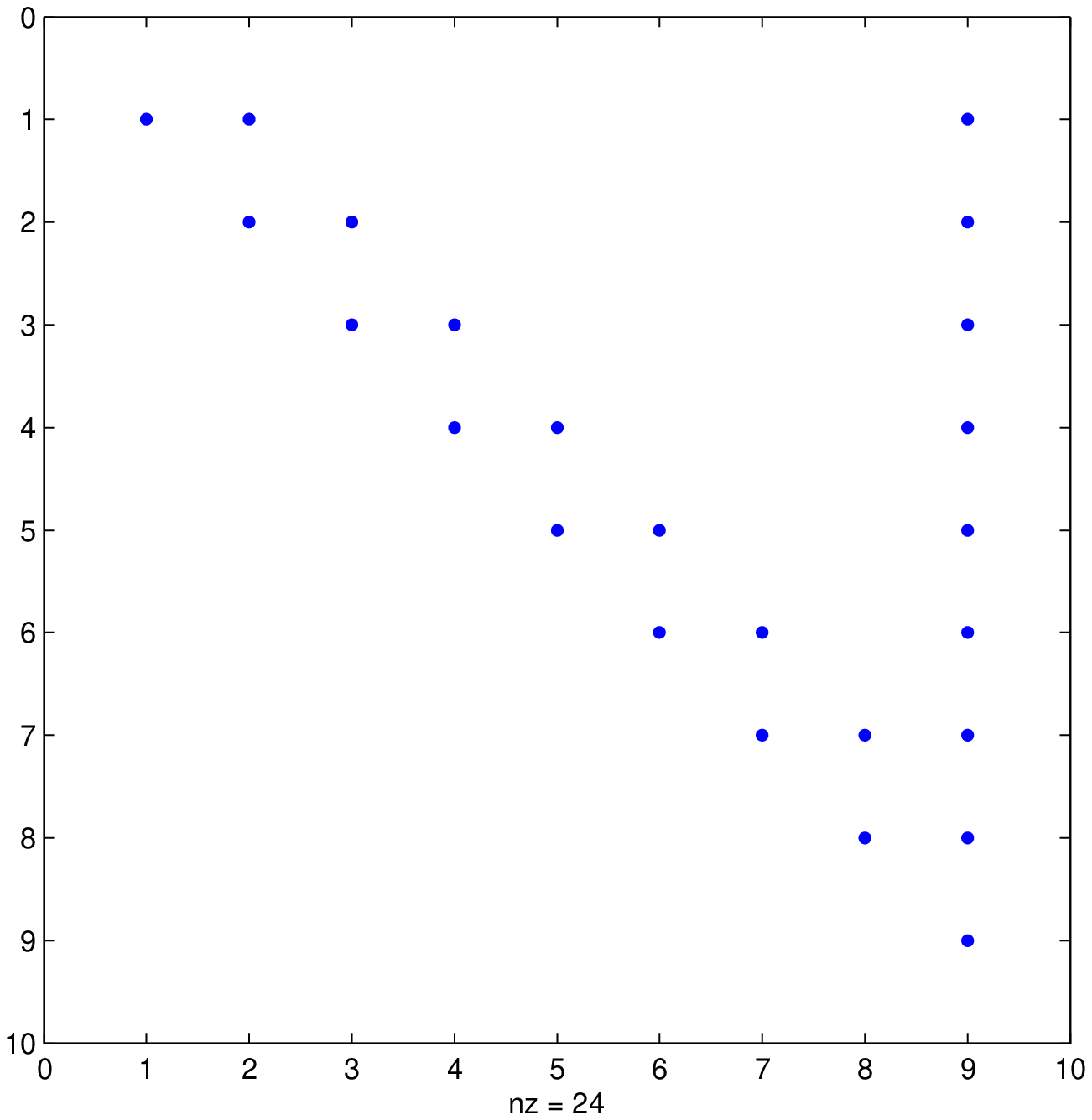} \label{fig: cholesky_small_L} }
    \subfigure[]{ \includegraphics[width=0.2\textwidth,height=0.4\textwidth]{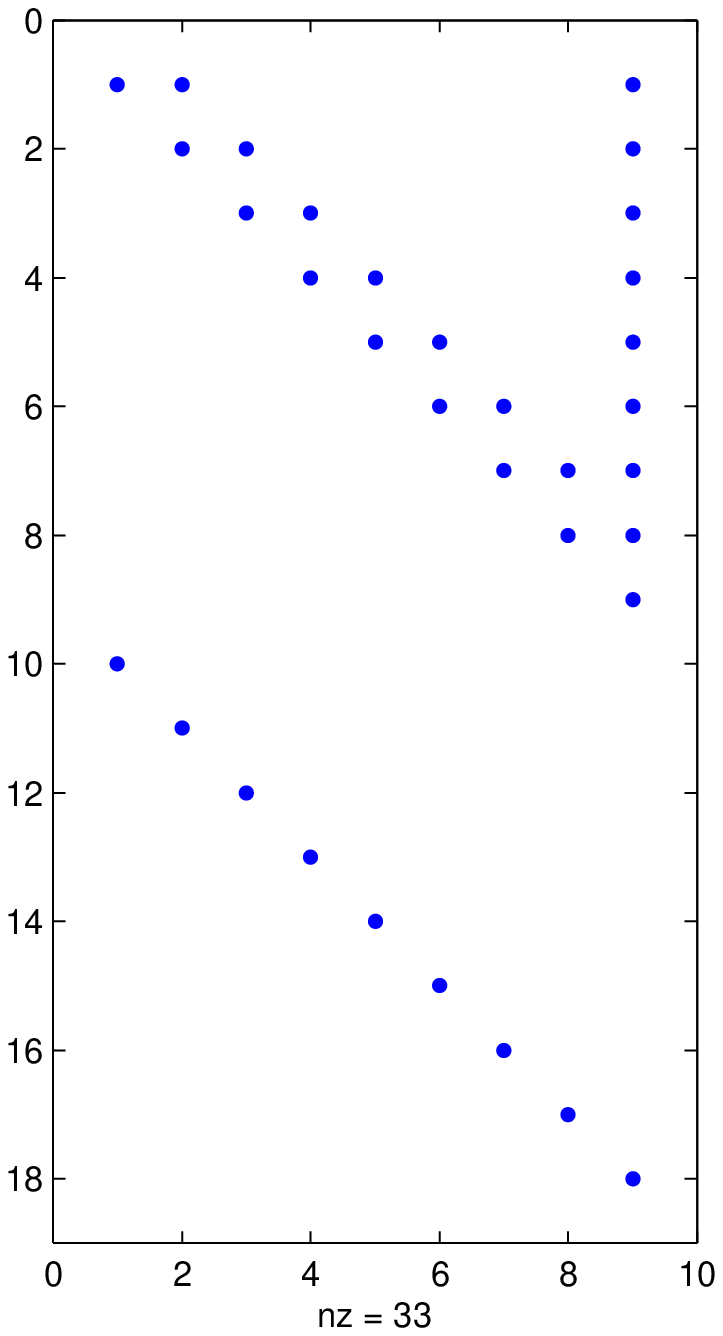}  \label{fig: cholesky_small_A} }
    \subfigure[]{ \includegraphics[width=0.2\textwidth,height=0.4\textwidth]{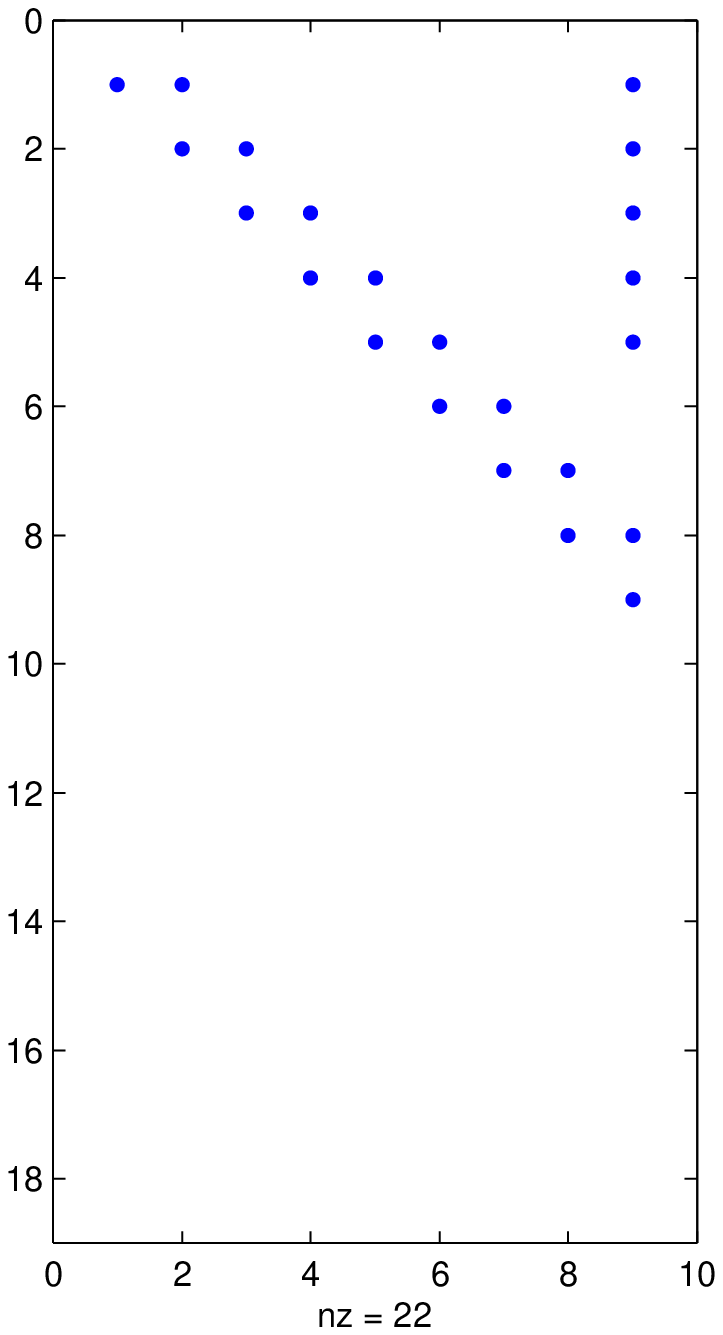}  \label{fig: cholesky_small_R} }
    \caption{The sparsity patterns of $\boldsymbol{L}_1$ (a), $\boldsymbol{L}_2$ (b), $\boldsymbol{L}$ (c), $\boldsymbol{A}$ (d) and $\boldsymbol{R}$ (e)} \label{fig: SmallExample}
  \end{figure}

  We notice that the precision matrix $\boldsymbol{Q}_1$ is quite similar to the tridiagonal matrix except the values at
  two of the corners. However, there is a lot of fill-in in the Cholesky factor $\boldsymbol{L}_1$. This is a common structure for the precision matrix of a GMRF,
  for instance, a GMRF on a torus. The same comments can be given for $\boldsymbol{Q}$ and $\boldsymbol{L}$. Note that the upper triangular matrix $\boldsymbol{R}$ has less nonzero elements than
  $\boldsymbol{L}_1$, $\boldsymbol{L}$ and $\boldsymbol{A}$, $N_{\boldsymbol{R}} < N_{\boldsymbol{L}_1}, N_{\boldsymbol{R}} < N_{\boldsymbol{L}} \text{ and } N_{\boldsymbol{R}} < N_{\boldsymbol{A}}$. 
  The sparsity pattern of $\boldsymbol{R}$ depends on the dropping tolerance and also the elements of the matrices $\boldsymbol{Q}_1$ and $\boldsymbol{Q}_2$, but we are not going deeper here.

  As discussed in Section \ref{sec: cholesky_IGO} and Section \ref{sec: cholesky_sparse_factor}, the sparse upper triangular matrix $\boldsymbol{R}$ is an
  incomplete Cholesky factor for the precision matrix $\boldsymbol{Q}$ of the GMRF when it is conditioned on data. The error matrix $\boldsymbol{E}$ 
  between the true precision matrix $\boldsymbol{Q} = \boldsymbol{Q}_1 + \boldsymbol{Q}_2$ and the approximated precision matrix $\boldsymbol{\widetilde{Q}} = \boldsymbol{R}^{\mbox{T}} \boldsymbol{R}$ is given by
  \begin{equation}
  \boldsymbol{E} = \left(\boldsymbol{Q}_1 + \boldsymbol{Q}_2\right) - \boldsymbol{\widetilde{Q}}.
  \end{equation}
  
  \noindent The sparsity patterns of the precision matrix $\boldsymbol{Q}$ and its approximation $\boldsymbol{\widetilde{Q}}$
   are shown in Figure \ref{fig: cholesky_small_Q} and Figure \ref{fig: cholesky_small_Qq}, respectively. In order to compare the difference between the approximated covariance matrices (inverse 
   of the approximated precision matrix) $\widetilde{\boldsymbol{\Sigma}} = \widetilde{\boldsymbol{Q}}^{-1}$
   and the true covariance matrix (inverse of the true precision matrix) $\boldsymbol{\Sigma} = \boldsymbol{Q}^{-1}$, we calculate the error matrix $\widetilde{\boldsymbol{E}}$,

 \begin{equation}
  \widetilde{\boldsymbol{E}} = \boldsymbol{\Sigma} - \widetilde{\boldsymbol{\Sigma}}.  
 \end{equation}
  The images of $\boldsymbol{\Sigma}$, $\boldsymbol{\widetilde{\Sigma}}$ and $\widetilde{\boldsymbol{E}}$ are given in Figure \ref{fig: ImageSQ_Qq_Err}, and they
  show that the difference between $\boldsymbol{\Sigma}$ and $\widetilde{\boldsymbol{\Sigma}}$ is quite small. 
  By chosen different dropping tolerance, the error can be made smaller and become negligible.

    \begin{figure}[htb]
    \centering
    \subfigure[]{\includegraphics[width=0.4\textwidth,height=0.4\textwidth]{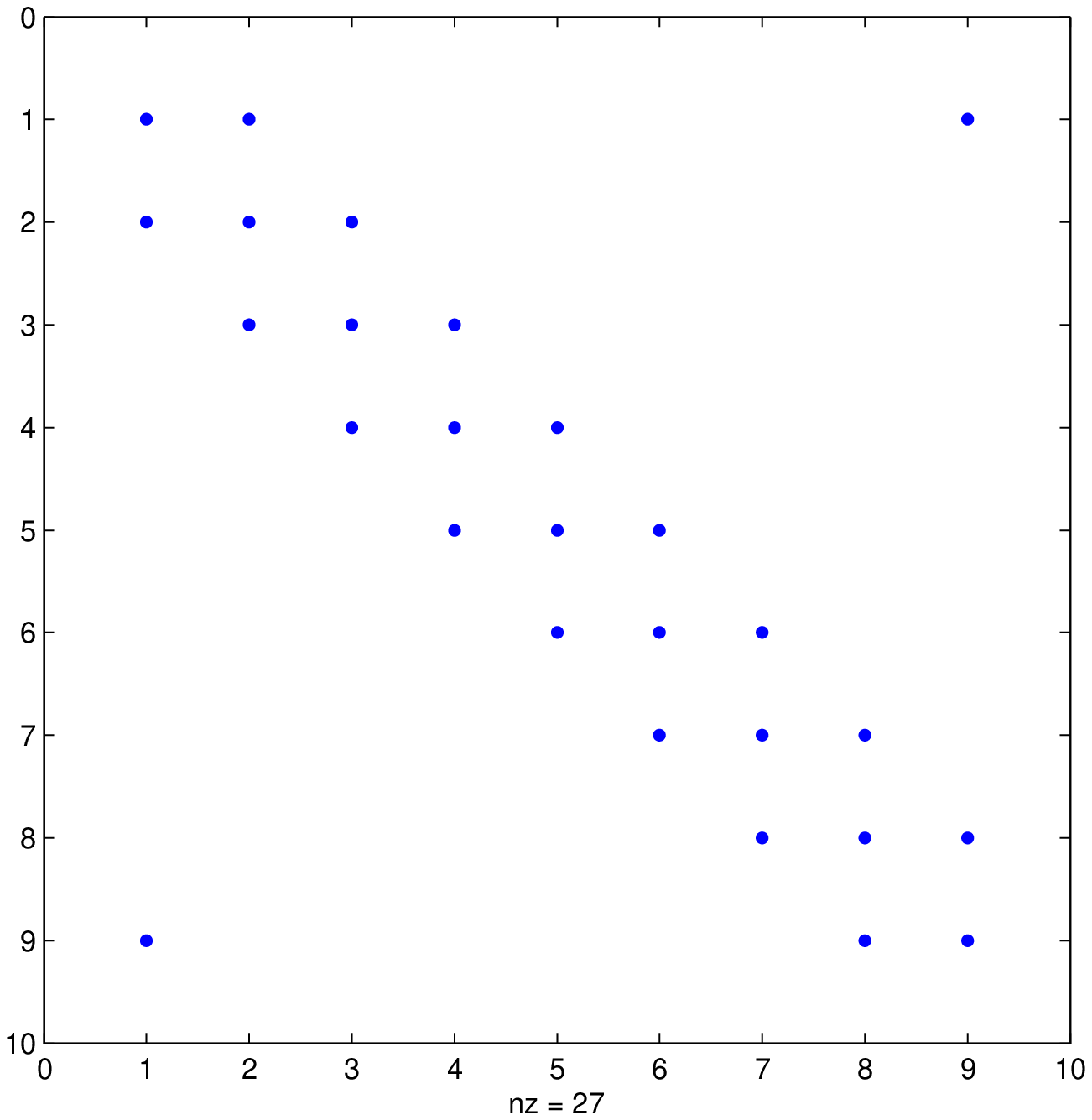}  \label{fig: cholesky_small_Q}}
    \subfigure[]{\includegraphics[width=0.4\textwidth,height=0.4\textwidth]{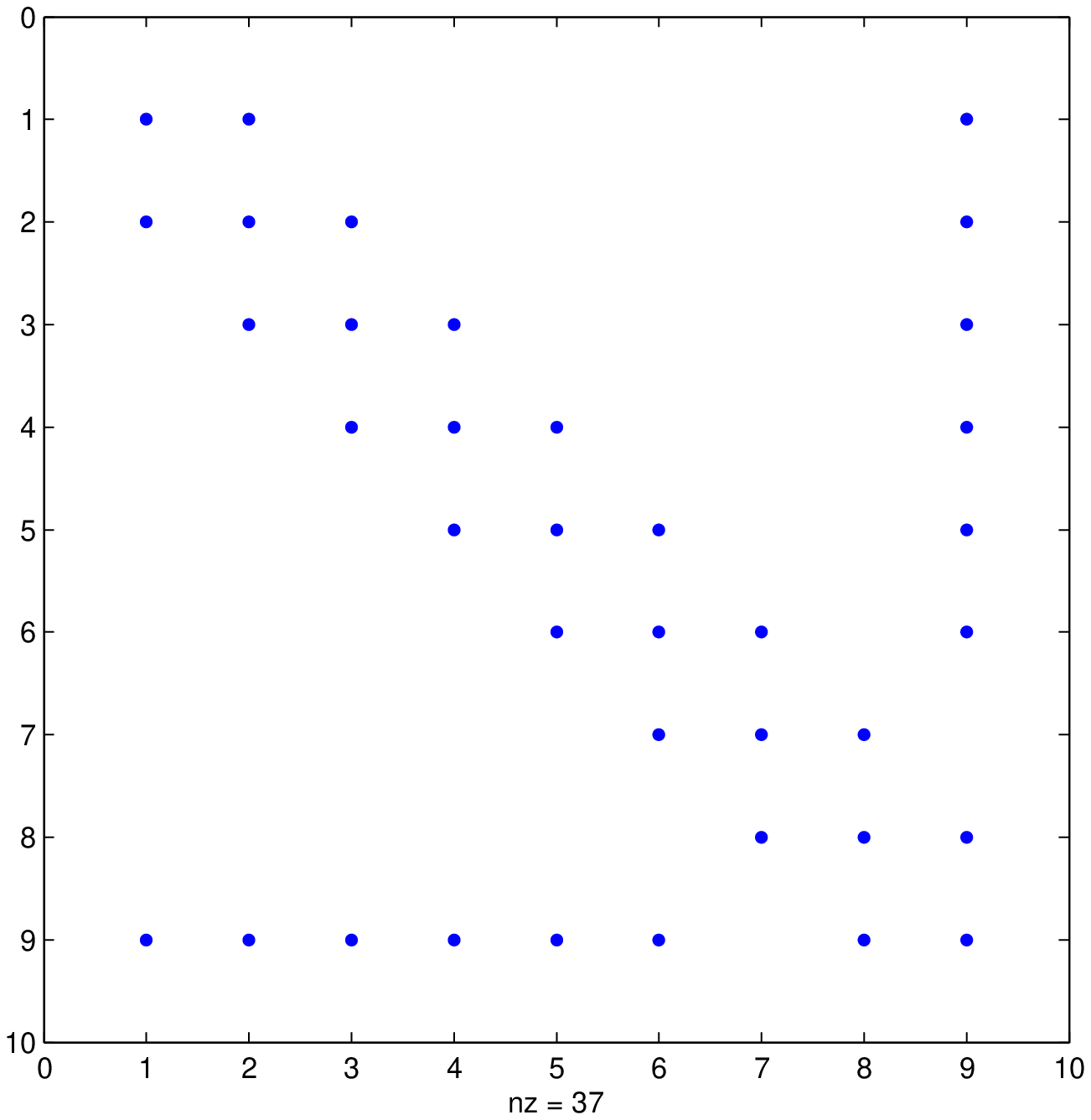}  \label{fig: cholesky_small_Qq}}
    \caption{Sparsity patterns of the true precision matrix $\boldsymbol{Q}$ (a) and the approximated precision matrix $\boldsymbol{\widetilde{Q}}$ (b) }
    \label{fig: cholesky_SparseQ_Qq}
  \end{figure}
 
    \begin{figure}[htb]
    \centering
    \subfigure[]{\includegraphics[width=0.3\textwidth,height=0.3\textwidth]{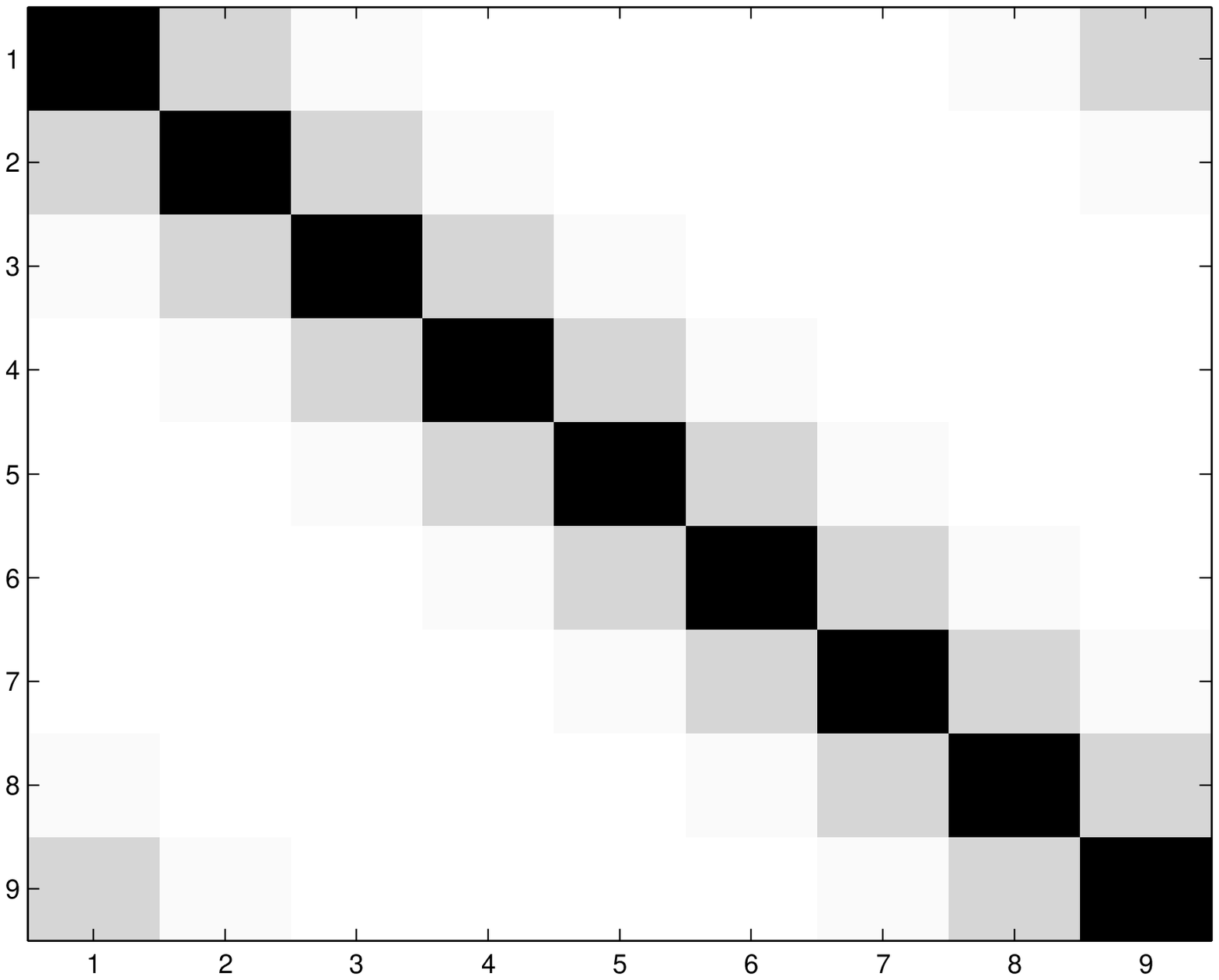}}
    \subfigure[]{\includegraphics[width=0.3\textwidth,height=0.3\textwidth]{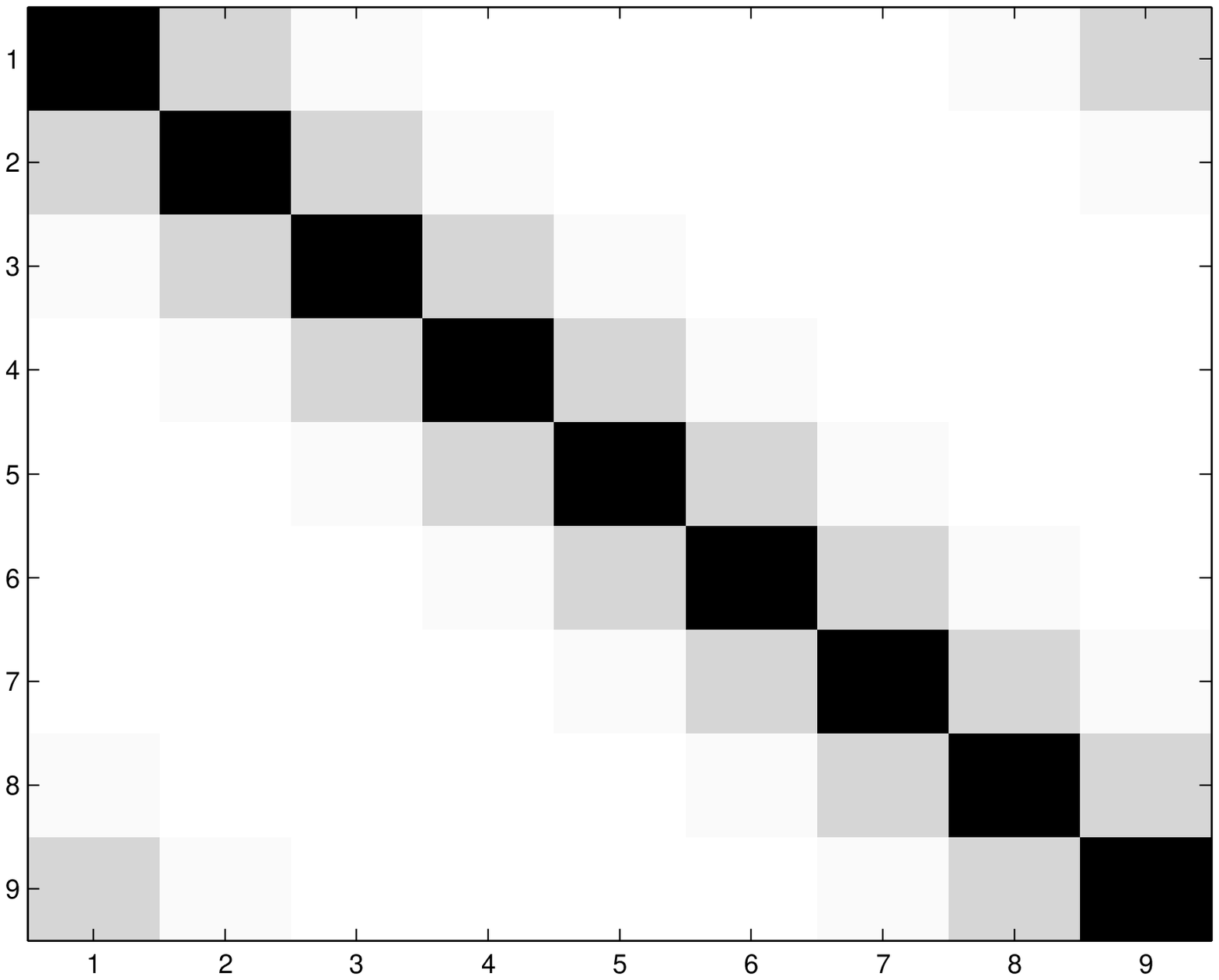}}
    \subfigure[]{\includegraphics[width=0.3\textwidth,height=0.315\textwidth]{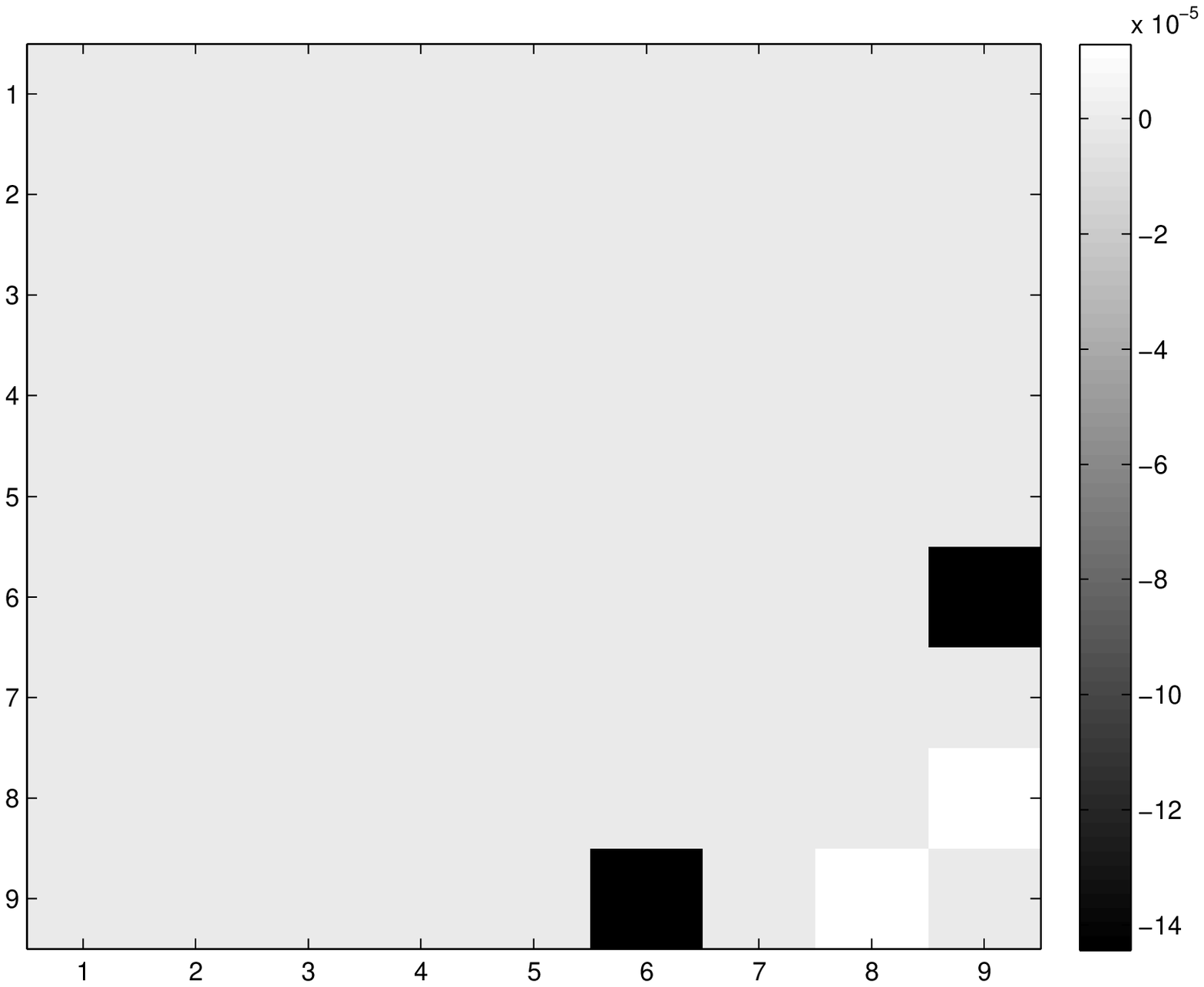}}
    \caption{ Images of the true covariance matrix $\boldsymbol{\Sigma}$ (a), approximately covariance matrix $\boldsymbol{\widetilde{\Sigma}}$ (b) and 
     the error matrix $\widetilde{\boldsymbol{E}}$ (c)} \label{fig: ImageSQ_Qq_Err}
   \end{figure}

\section{Simulation Results with cTIGO algorithm} \label{sec: cholesky_results}
   Using the incomplete orthogonal factorization with Givens rotations, it leads to a sparse upper triangular matrix $\boldsymbol{R}$, which is a sparse incomplete Cholesky factor for the precision matrix $\boldsymbol{Q}$ and
   can be used to specify the GMRF. Hence it has the potential possibility to reduce the computational cost.
   We first apply the cTIGO algorithm to some commonly used structures of the precision matrices in Section \ref{sec: cholesky_result_band_matrices}. 
   In Section \ref{sec: cholesky_result_spde}, we apply the cTIGO algorithm to 
   precision matrices which are generated from the stochastic partial differential equations (SPDEs) discussed in \citet{lindgren2011explicit} and \citet{fuglstad2011spatial}.  

   \subsection{Simulation results for precision matrices with commonly used structures} \label{sec: cholesky_result_band_matrices}
    It is known that if the precision matrix $\boldsymbol{Q} > 0$ is a band matrix with bandwidth $p$, then its Cholesky factor $\boldsymbol{L}$ (lower triangular matrix)
    has the same bandwidth $p$. See \citet{golub1996matrix} (Theorem 4.3.1) for a direct proof and \citet[Chapter 2.4.1]{rue2005gaussian} for more information on
    how to finding Cholesky factor efficiently in this case with Algorithm $2.9$. \citet{wist2006specifying} pointed out that if the original precision
    matrix $\boldsymbol{Q}$ is a band matrix, then the incomplete Cholesky factor $\boldsymbol{\widetilde{L}}$ from the incomplete Cholesky factorization will also be a band
    matrix with the same bandwidth $p$. 

    In this section we consider some commonly used structures for the precision matrices. The first two examples are band matrices with different bandwidths.
    Let $x$ be Gaussian auto-regressive processes of order $1$ or $2$, 
    and then the precision matrix for the process will be a band matrix with bandwidth $p = 2$ or $p = 3$, respectively. The precision matrices for the first-order Random Walk (RW$1$) and the second-order Random
    Walk (RW$2$) models have bandwidths $p = 2$ and $p = 3$. Since these models are intrinsic GMRFs, the precision matrices are not of full rank. 
    We fix this by slightly modifying the elements in the precision matrices for the RW$1$ and RW$2$ models but we still called them as the precision matrices for 
    the RW$1$ and the RW$2$ models. For more information about intrinsic GMRFs and 
    the RW$1$ and RW$2$ models, see, for example, \citet[Chapter 3]{rue2005gaussian}.

    Assume that the data are Gaussian distributed. Then from Section \ref{sec: cholesky_GMRFs_conditioning_data} the matrix $\boldsymbol{Q}_2$ is a diagonal matrix when $\boldsymbol{A} = \boldsymbol{I}$. For simplicity and without lost of generality, 
    assume the data $\boldsymbol{y}$ $\sim$ $\mathcal{N}(\boldsymbol{0},\boldsymbol{I})$, then the matrix $\boldsymbol{Q}_2$ and its Cholesky factor $\boldsymbol{L}_2$ are identity matrices.
    Since we know exactly what the sparsity patterns of the precision matrices $\boldsymbol{Q}_1$ and $\boldsymbol{Q}_2$ and the Cholesky factors $\boldsymbol{L}_1$ and $\boldsymbol{L}_2$ are, 
    the sparsity pattern of $\boldsymbol{A}$ is known beforehand and can be taken advantage of in the implementation.
    By applying the cTIGO algorithm to the matrix $\boldsymbol{A}$ with dropping tolerance $\tau = 0.0001$, the sparse upper triangular matrix $\boldsymbol{R}$ can be obtained. 
    The sparsity patterns of the matrices $\boldsymbol{L}_1$, $\boldsymbol{L}_2$, $\boldsymbol{L}$, $\boldsymbol{A}$ and $\boldsymbol{R}$ are given in Figure \ref{fig: cholesky_RW1}.
    The sparsity patterns of the true precision matrix $\boldsymbol{Q}$ and the approximated precision matrix $\boldsymbol{\widetilde{Q}}$ are given in Figure \ref{fig:cholesky_RW1_QandQq}.
    The image of the true covariance matrices $\boldsymbol{\Sigma}$, 
    the approximated covariance matrix $\boldsymbol{\widetilde{\Sigma}}$ and the error matrix $\widetilde{\boldsymbol{E}}$ 
    for the RW$1$ model are shown in Figure \ref{fig: cholesky_RW1ImQ12QqR}. Note that the order of the numerical values in the error matrix $\widetilde{\boldsymbol{E}}$  
    is $10^{-8}$, which is essentially zero in practice applications.

    Similarly for the RW$2$ model we apply the cTIGO algorithm to the matrix $\boldsymbol{A}$ with the dropping tolerance $\tau = 0.0001$. 
    The results in this case are quite similar to the results for the RW$1$ model. We only show the 
    images of the true covariance matrix $\boldsymbol{\Sigma}$, the approximated covariance matrix $\boldsymbol{\widetilde{\Sigma}}$, and the error matrix $\boldsymbol{\widetilde{E}}$. 
    The results are given in Figure \ref{fig: cholesky_RW2_ImQ12QqR}. Note that the order of the numerical values in the error matrix $\widetilde{\boldsymbol{E}}$ is $10^{-8}$ as for the RW$1$ model.
    See Section \ref{sec: cholesky_result_comparsion} from more simulation results for the RW$1$ and RW$2$ models and discussions. We can notice that the sparseness of $\boldsymbol{R}$ is the same as $\boldsymbol{L}$. Hence in
    these two cases, we do not save computational resources. However, this approach is still have the potential to be used in applications since it is robust.
    \begin{figure}[htbp]
    \centering
    \subfigure[]{\includegraphics[width=0.4\textwidth,height=0.4\textwidth]{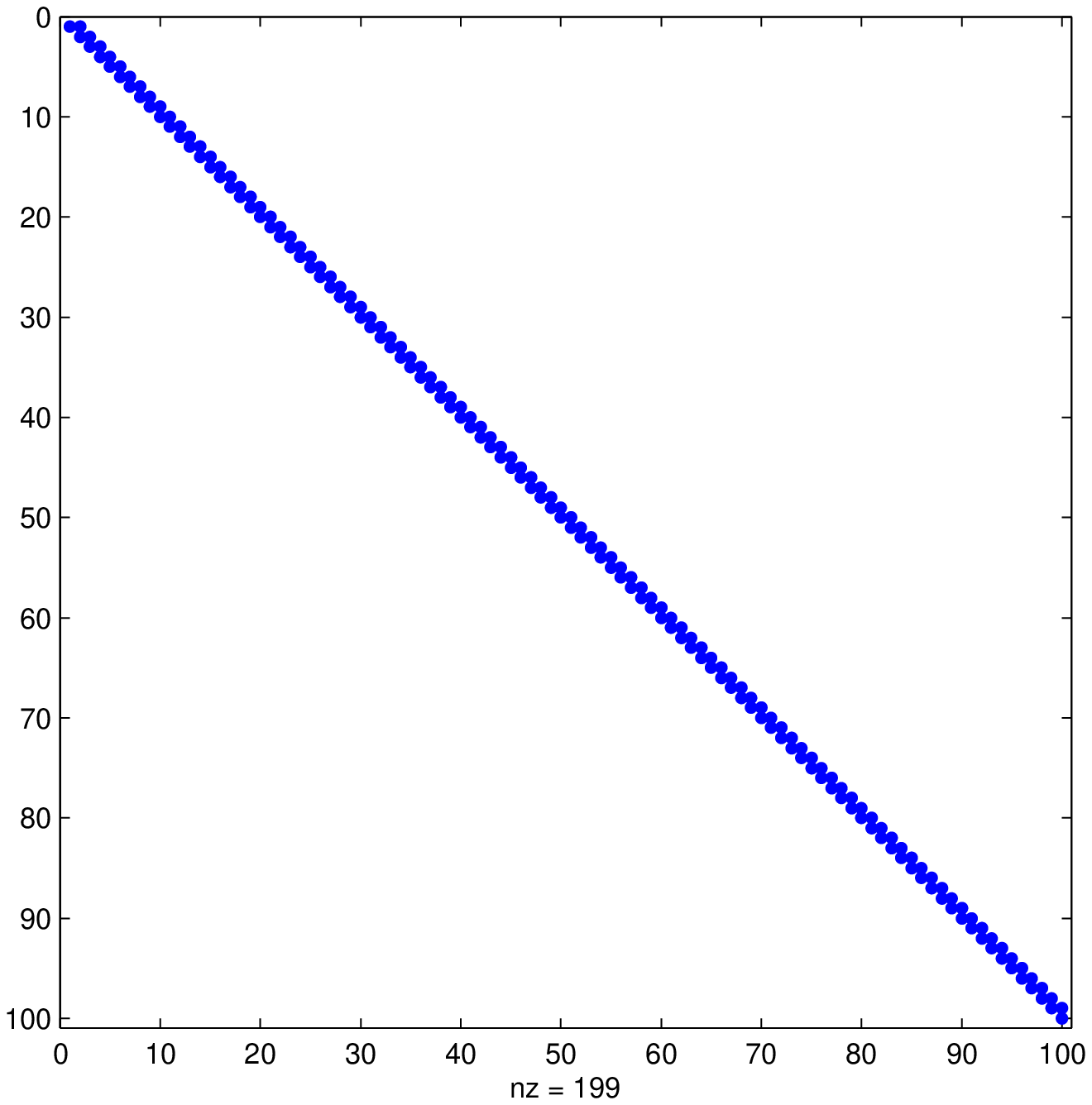} \label{fig: cholesky_RW1_L1} }
    \subfigure[]{\includegraphics[width=0.4\textwidth,height=0.4\textwidth]{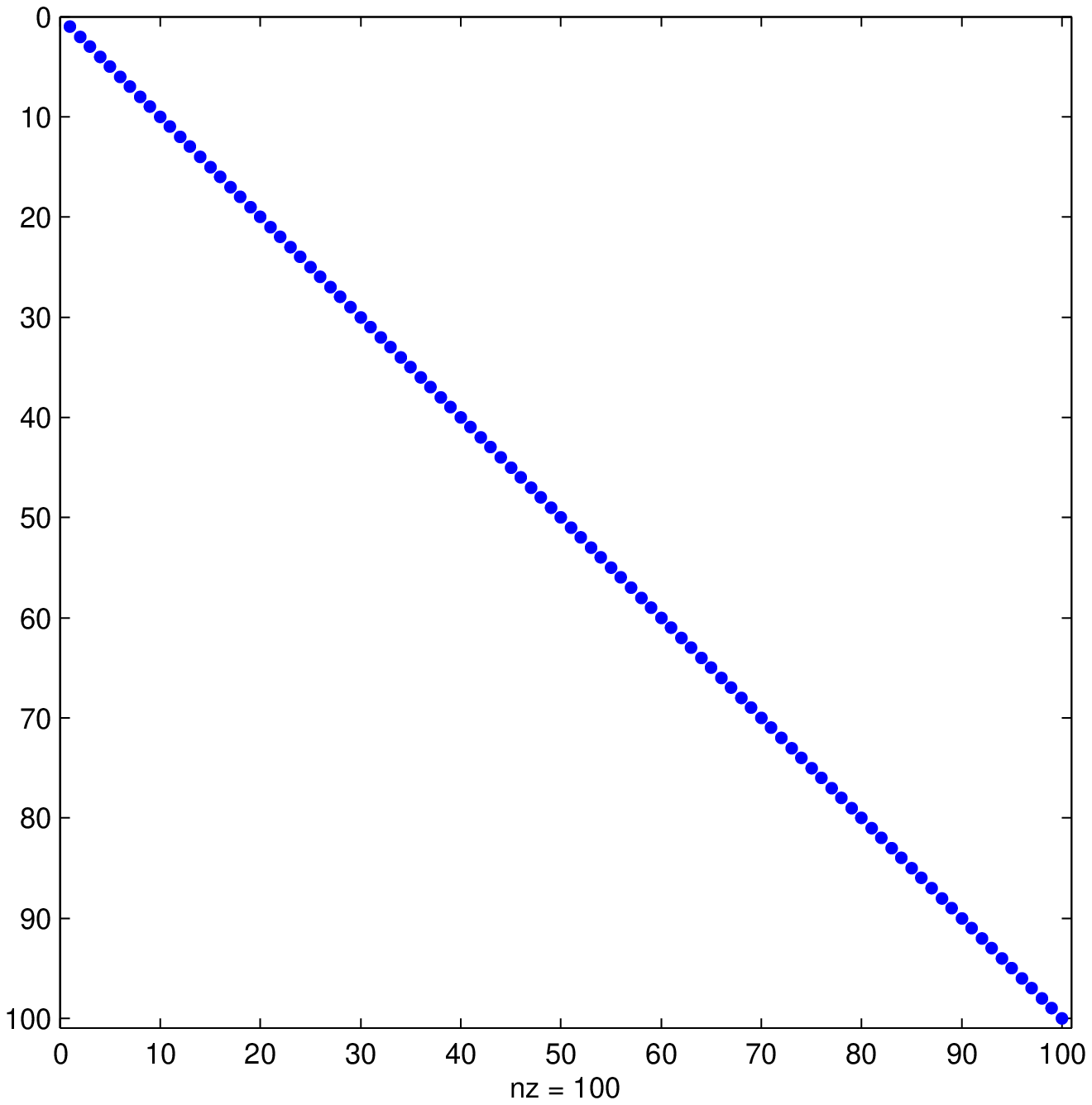} \label{fig: cholesky_RW1_L2} } 
    \subfigure[]{\includegraphics[width=0.4\textwidth,height=0.4\textwidth]{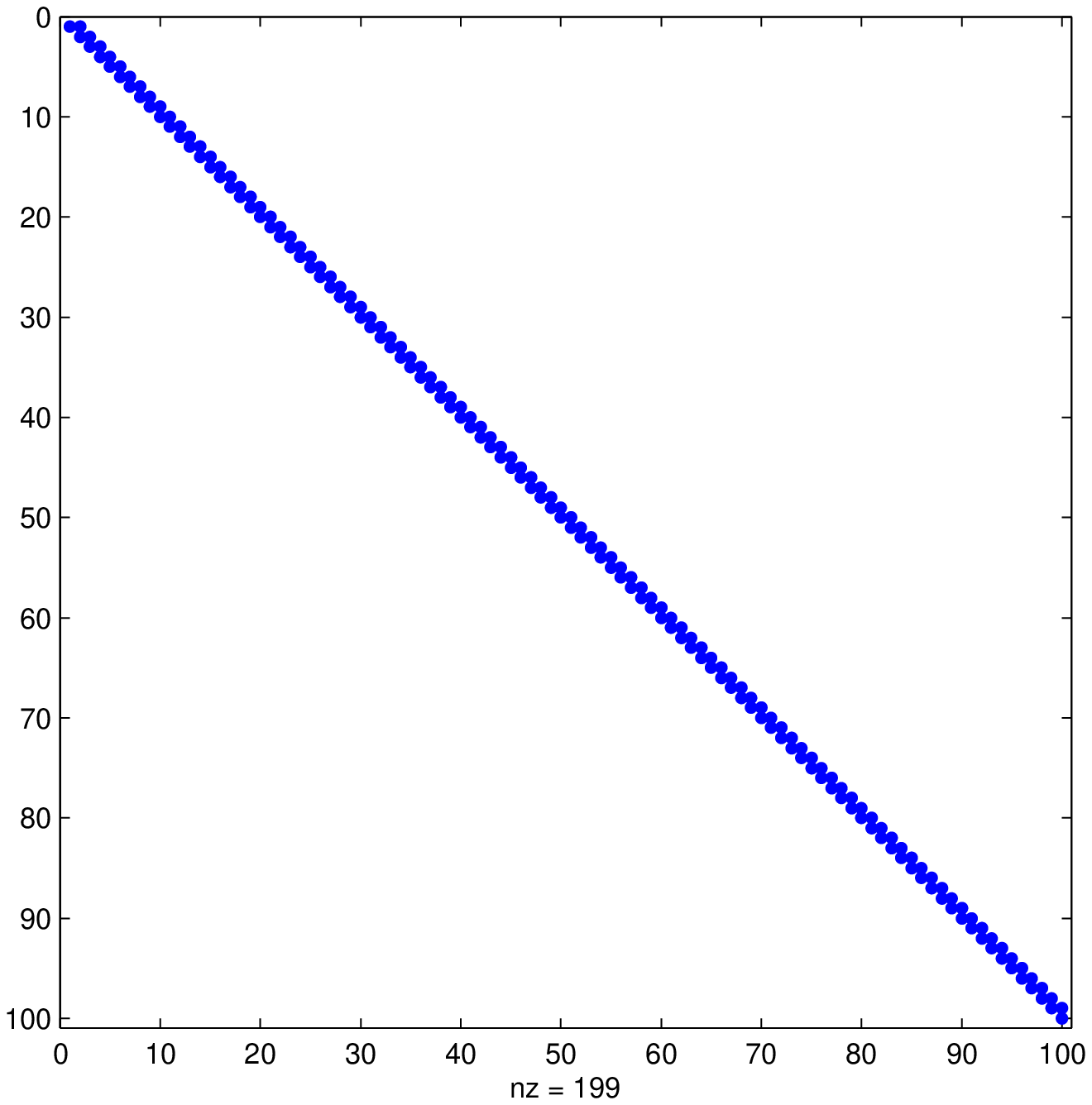}  \label{fig: cholesky_RW1_L} }
    \subfigure[]{\includegraphics[width=0.2\textwidth,height=0.4\textwidth]{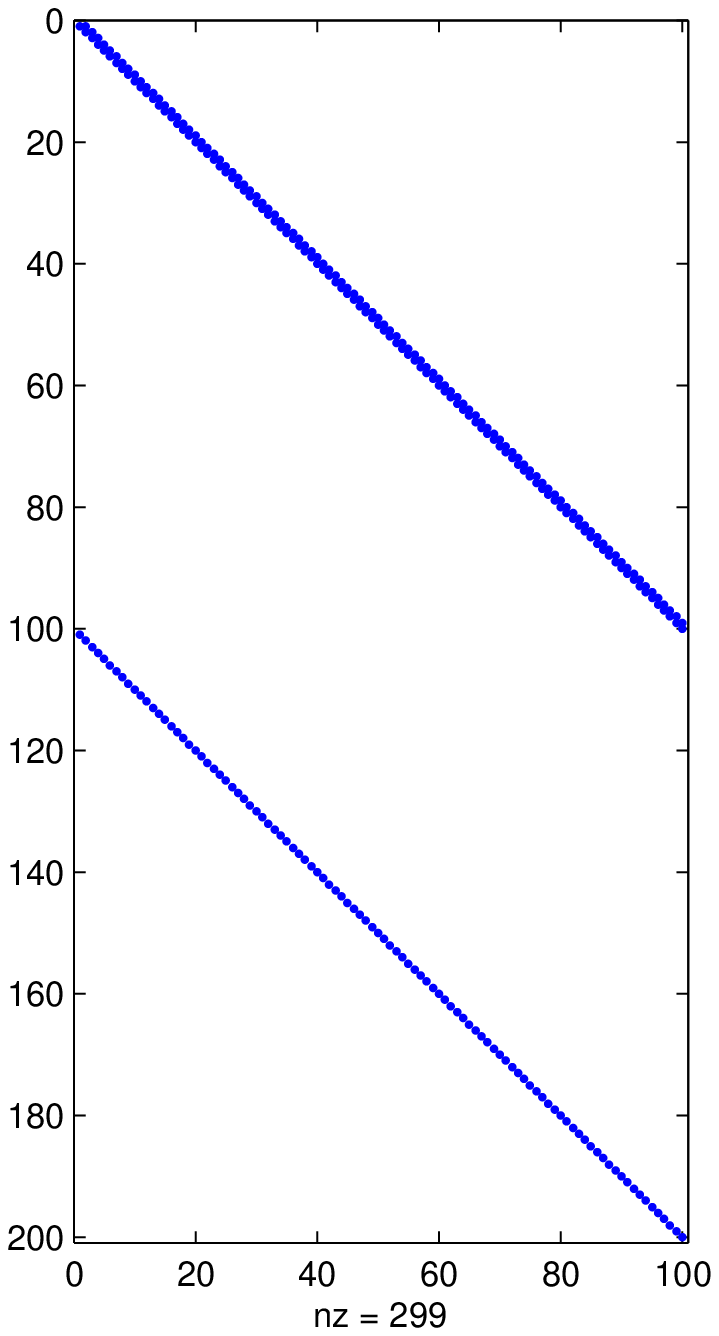} \label{fig: cholesky_RW1_A} }
    \subfigure[]{\includegraphics[width=0.2\textwidth,height=0.4\textwidth]{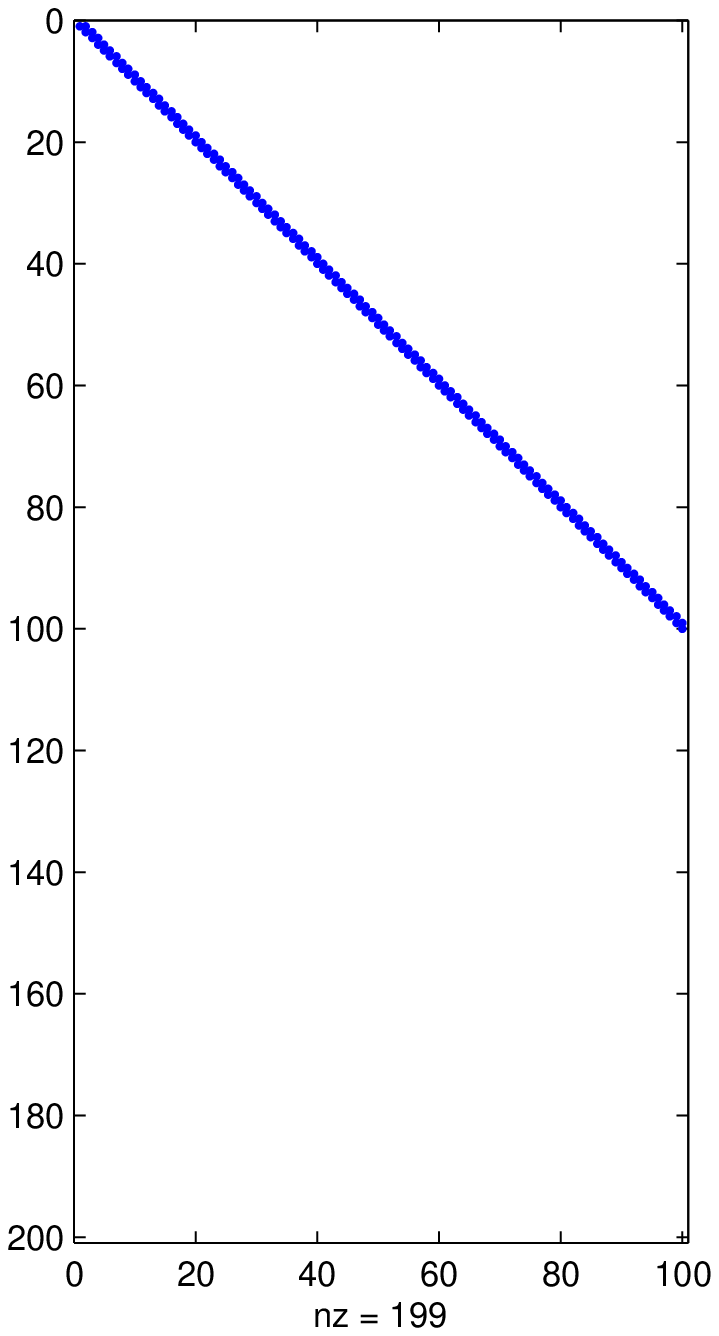} \label{fig: cholesky_RW1_R} }
    \caption{Sparsity patterns of $\boldsymbol{L}_1$ (a), $\boldsymbol{L}_2$ (b), $\boldsymbol{L}$ (c), $\boldsymbol{A}$ (d) and $\boldsymbol{R}$ (e) for the RW$1$ model} \label{fig: cholesky_RW1}
    \end{figure}   

    \begin{figure}[htb]
    \centering
    \subfigure[]{\includegraphics[width=0.4\textwidth,height=0.4\textwidth]{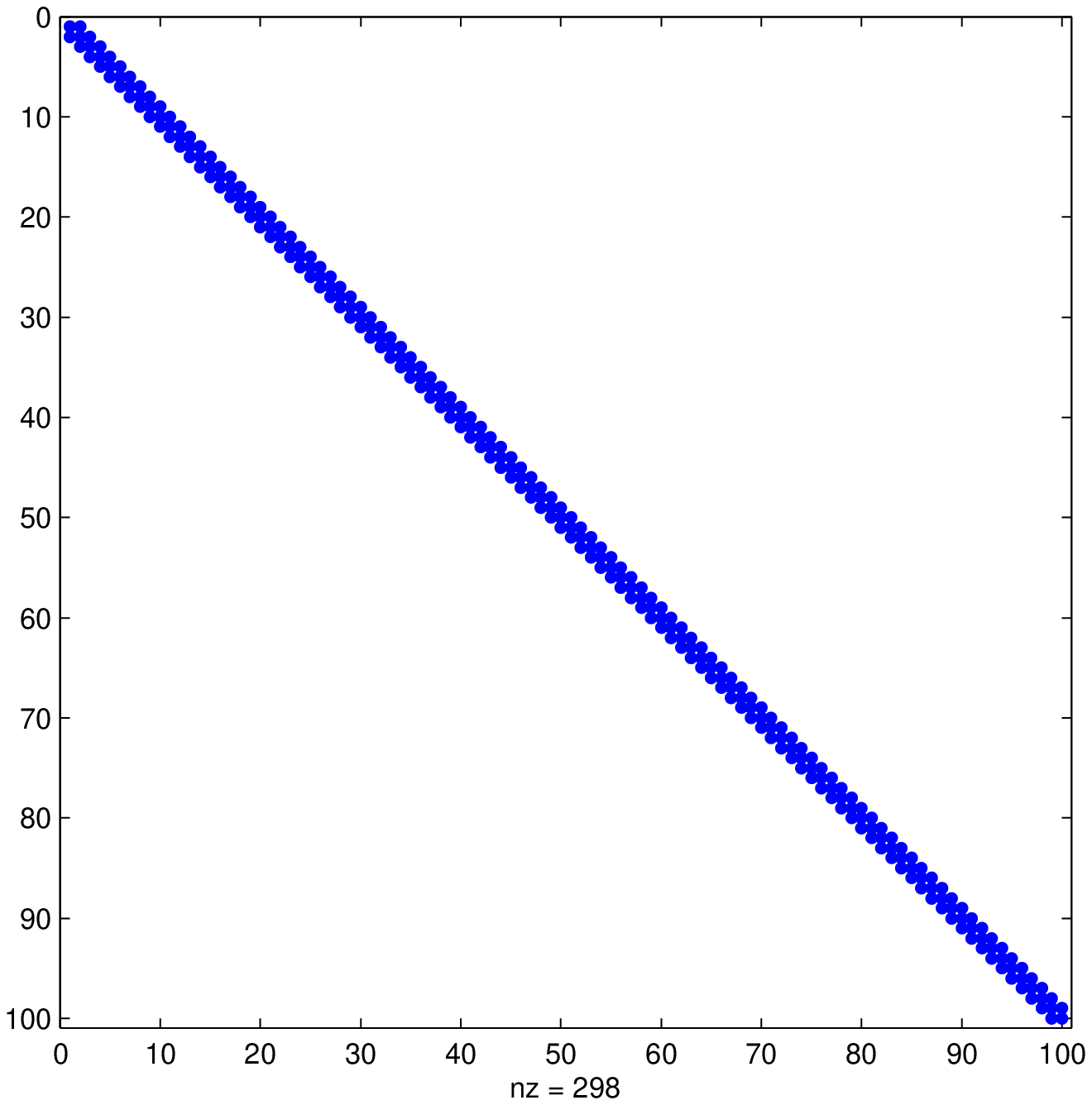} \label{fig: cholesky_RW1_Q} } 
    \subfigure[]{\includegraphics[width=0.4\textwidth,height=0.4\textwidth]{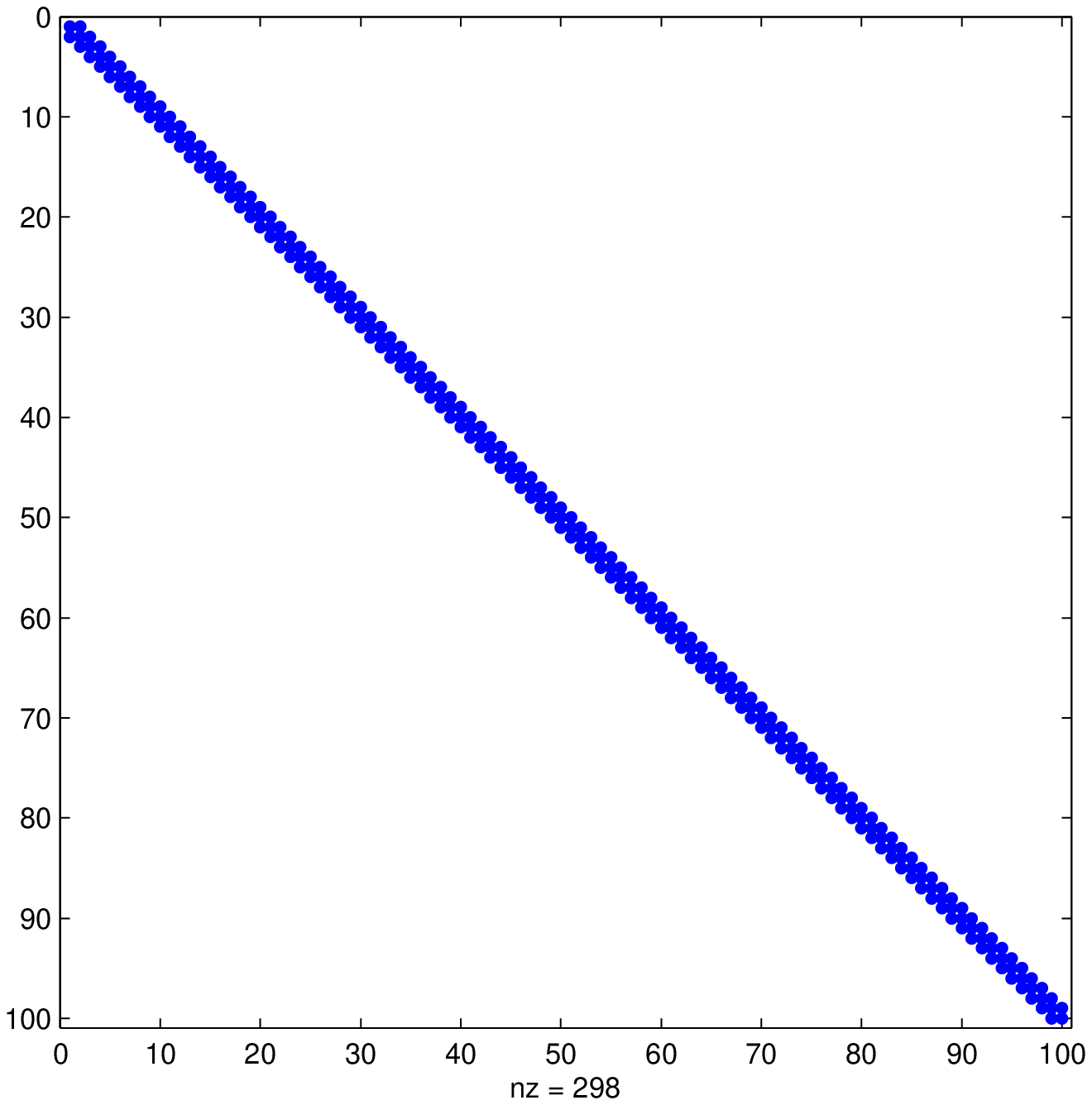} \label{fig: cholesky_RW1_Qq} }
    \caption{ Sparsity patterns for the true precision matrix $\boldsymbol{Q}$ (a) and the approximated precision matrix $\boldsymbol{\widetilde{Q}}$ (b) for the RW$1$ model} \label{fig:cholesky_RW1_QandQq}
    \end{figure}
   
    \begin{figure}[htb]
    \centering
    \subfigure[]{\includegraphics[width=0.3\textwidth,height=0.3\textwidth]{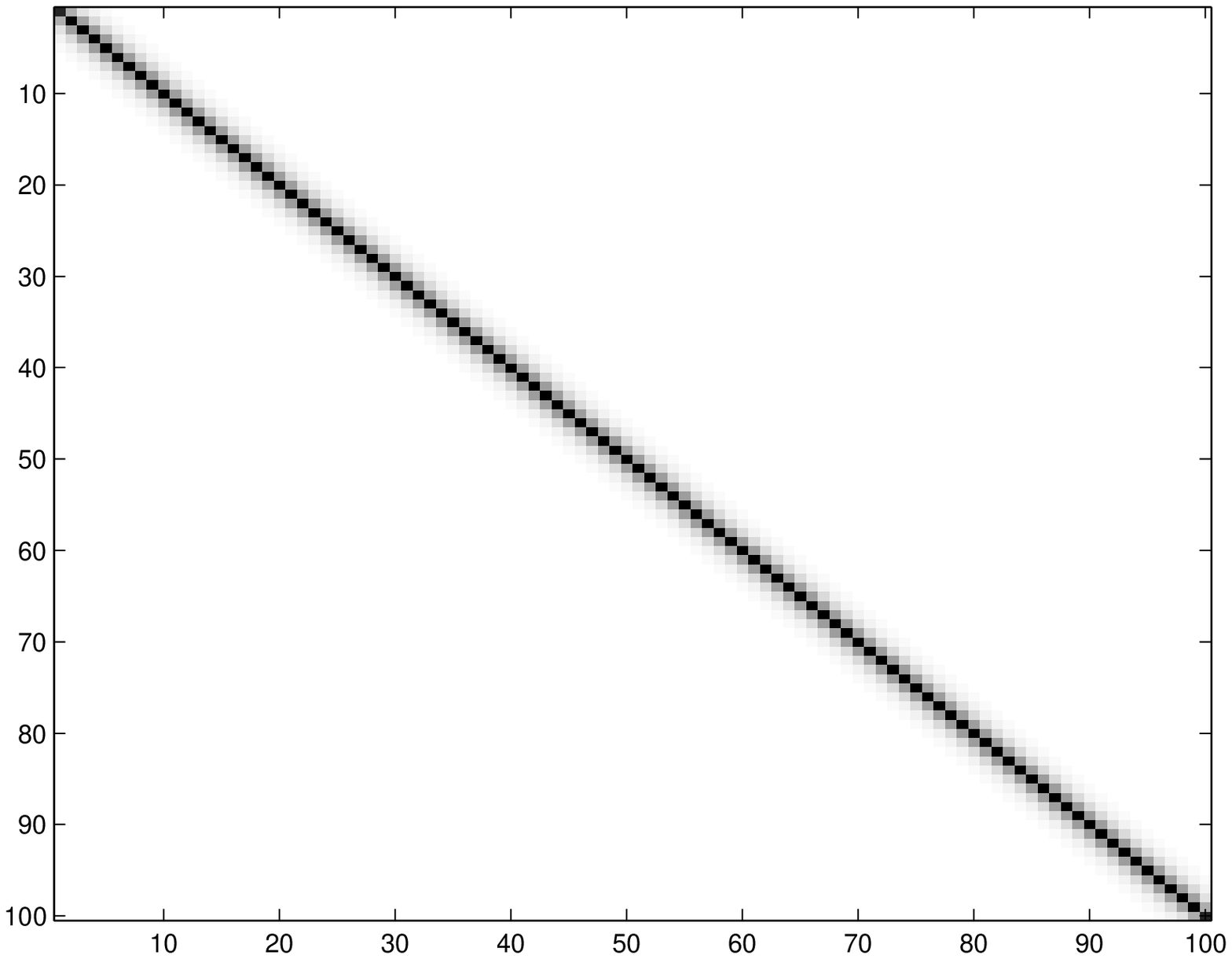} \label{fig: cholesky_RW1_InverseQ} } 
    \subfigure[]{\includegraphics[width=0.3\textwidth,height=0.3\textwidth]{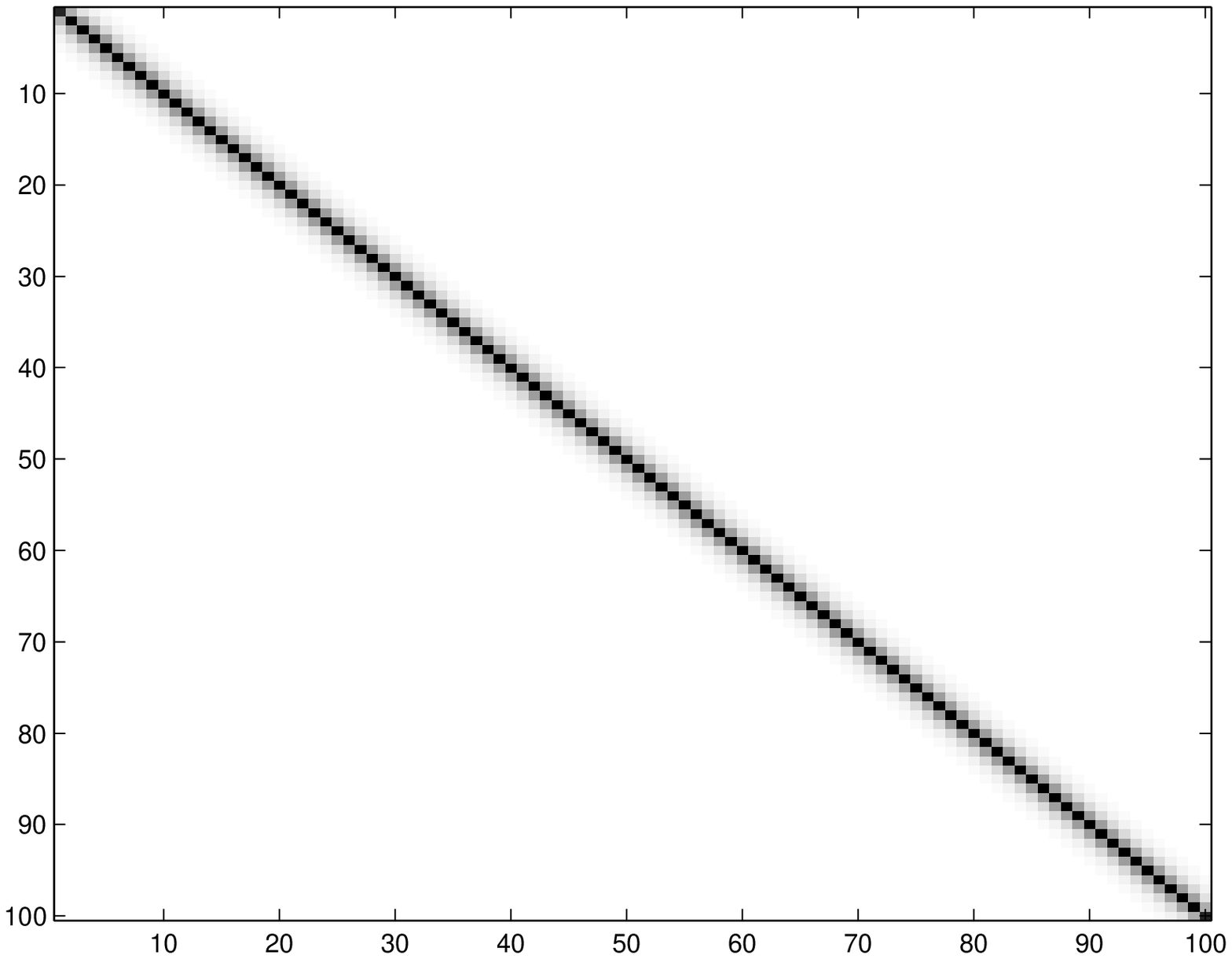} \label{fig: cholesky_RW1_InverseQq} }
    \subfigure[]{\includegraphics[width=0.3\textwidth,height=0.315\textwidth]{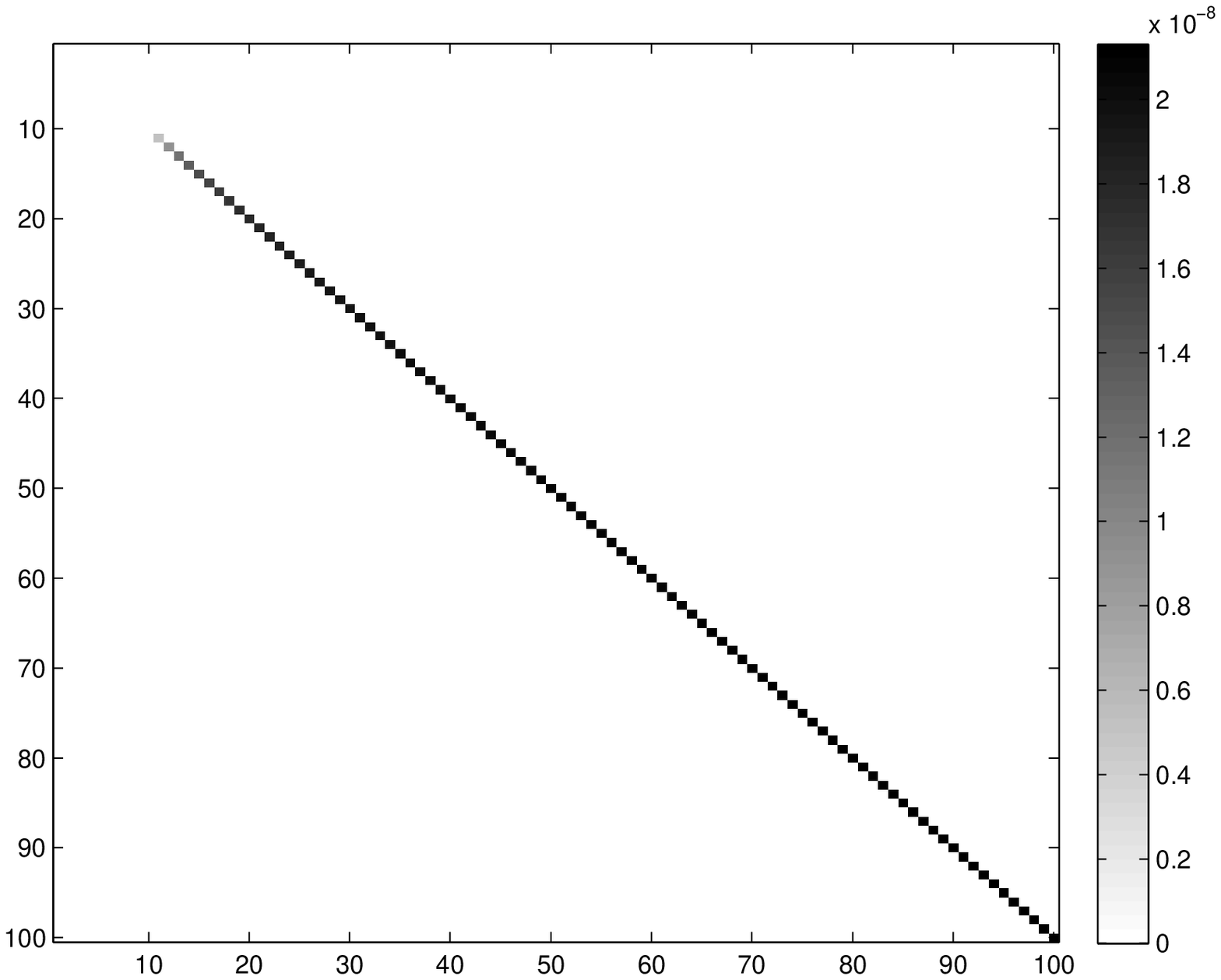} \label{fig: cholesky_RW1_Err} } 
    \caption{Images of true covariance matrix $\boldsymbol{\Sigma}$ (a), the approximated covariance matrix $\boldsymbol{\widetilde{\Sigma}}$ (b) and 
             and the error matrix $\widetilde{\boldsymbol{E}}$ (c) for RW$1$ model} \label{fig: cholesky_RW1ImQ12QqR}
    \end{figure}
      
    \begin{figure}[htb]
    \centering
    \subfigure[]{\includegraphics[width=0.3\textwidth,height=0.3\textwidth]{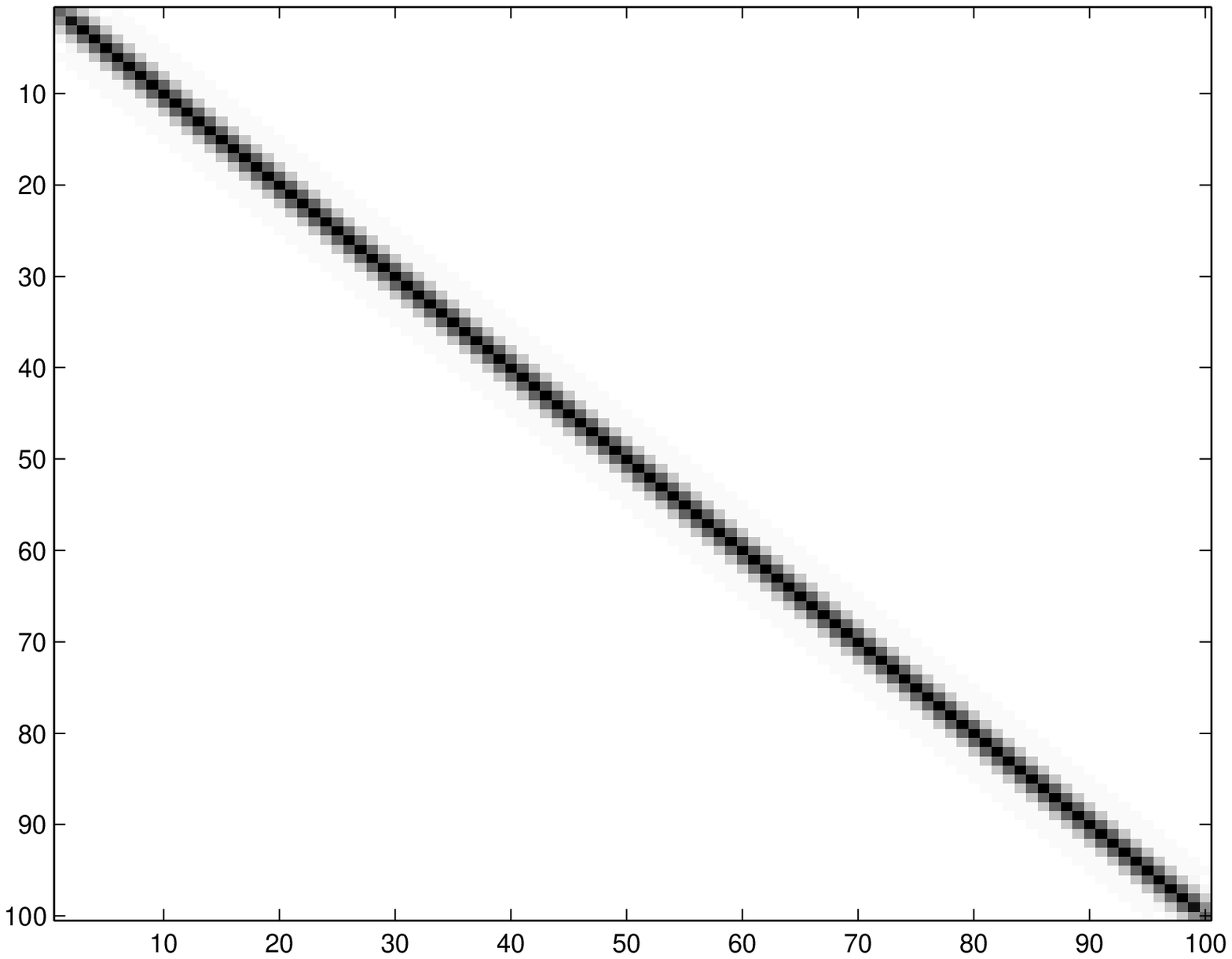} \label{fig: cholesky_RW2_InverseQ} }
    \subfigure[]{\includegraphics[width=0.3\textwidth,height=0.3\textwidth]{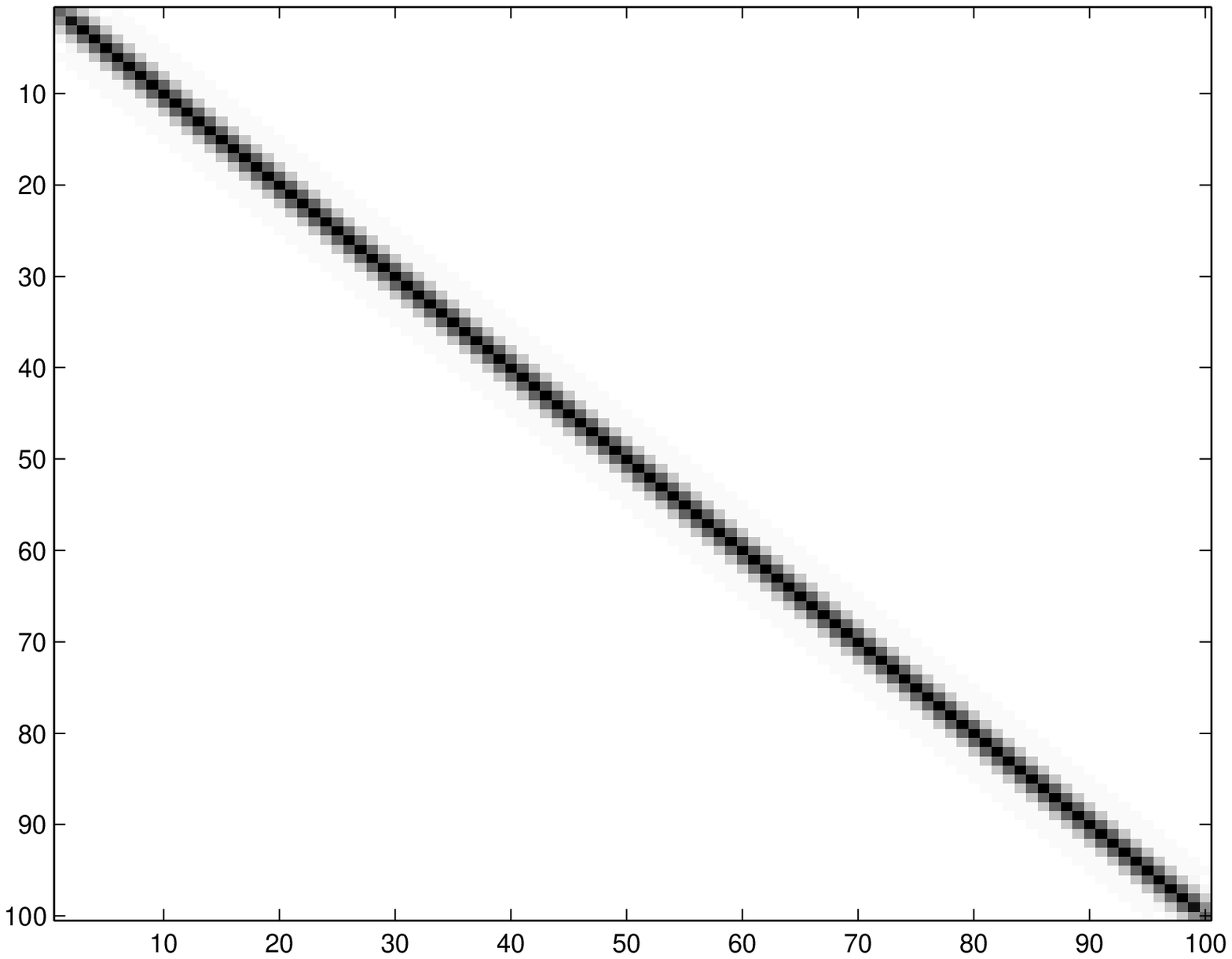} \label{fig: cholesky_RW2_InverseQq} }
    \subfigure[]{\includegraphics[width=0.3\textwidth,height=0.315\textwidth]{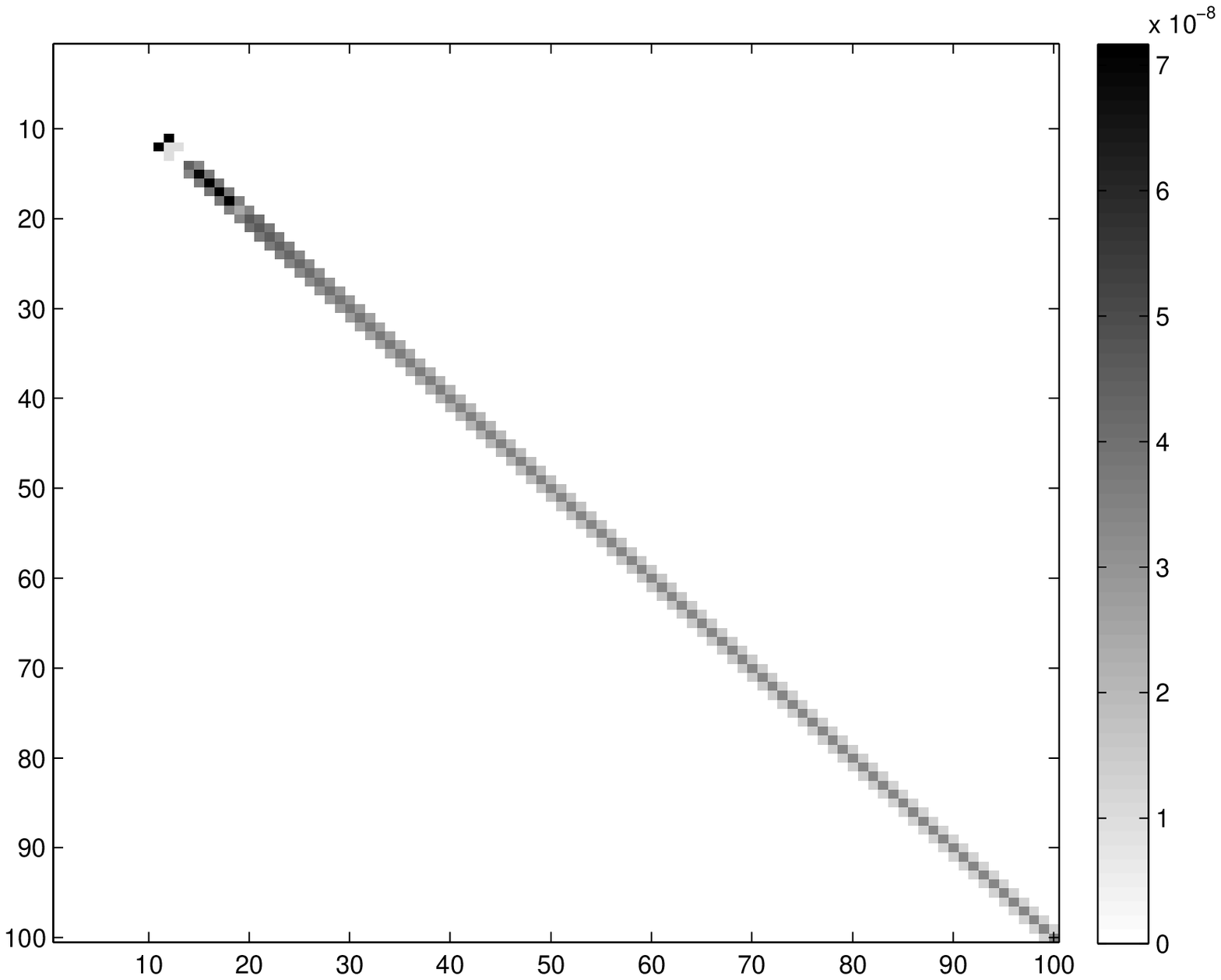} \label{fig: cholesky_RW2_Err} }
    \caption{Images of the true covariance matrix $\boldsymbol{\Sigma}$ (a), the approximated covariance matrix $\boldsymbol{\widetilde{\Sigma}}$ (b)
             and and the error matrix $\widetilde{\boldsymbol{E}}$ (c) for RW$2$ model} \label{fig: cholesky_RW2_ImQ12QqR}
    \end{figure}

    The next example we have chosen is a block tridiagonal matrix of order $n^2$ resulting from
    discretizing Poisson's equation with the $5$-point operator on an $n$-by-$n$ mesh. Thus it is called Poisson matrix in this paper. The sparsity pattern of this matrix
    is given in Figure \ref{fig: cholesky_Possion_Q}. With the Poisson matrix and $\boldsymbol{Q}_2$ as before, we find the Cholesky factors $\boldsymbol{L}_1$ and $\boldsymbol{L}_2$ and form the rectangular matrix $\boldsymbol{A}$.
    We apply the cTIGO algorithm to the matrix $\boldsymbol{A}$ with dropping tolerance $\tau = 0.0001$ to find the sparse upper triangular matrix $\boldsymbol{R}$.
    The sparsity patterns of the matrices $\boldsymbol{L}_1$, $\boldsymbol{L}_2$, $\boldsymbol{L}$, $\boldsymbol{A}$ and $\boldsymbol{R}$ are given in Figure \ref{fig: cholesky_Possion_L1} - Figure \ref{fig: cholesky_Possion_R}, respectively.
    The sparsity patterns of the true precision matrix $\boldsymbol{Q}$ and the approximated precision matrix $\boldsymbol{\widetilde{Q}}$ are given in Figure \ref{fig: Cholsky_Possion_QandQq}.
    We can notice that the upper triangular matrix $\boldsymbol{R}$ is sparser than the Cholesky factor $\boldsymbol{L}$ from the original precision matrix $\boldsymbol{Q}$.
    It can be shown that the sparseness depends on the dropping tolerance $\tau$. 
    The images of the true covariance matrices $\boldsymbol{\Sigma}$, the approximated precision matrix $\boldsymbol{\widetilde{\Sigma}}$, and the error
    matrix $\boldsymbol{\widetilde{E}}$ in this case are shown in Figure \ref{fig: cholesky_Possion_ImQ12QqR}. Note that
    the order of the numerical values in the error matrix $\boldsymbol{\widetilde{E}}$ is $10^{-5}$. This is small for practical use. More results for this band matrix are given in Section \ref{sec: cholesky_result_comparsion}.

    \begin{figure}[htb]
    \centering
    \subfigure[]{\includegraphics[width=0.4\textwidth,height=0.4\textwidth]{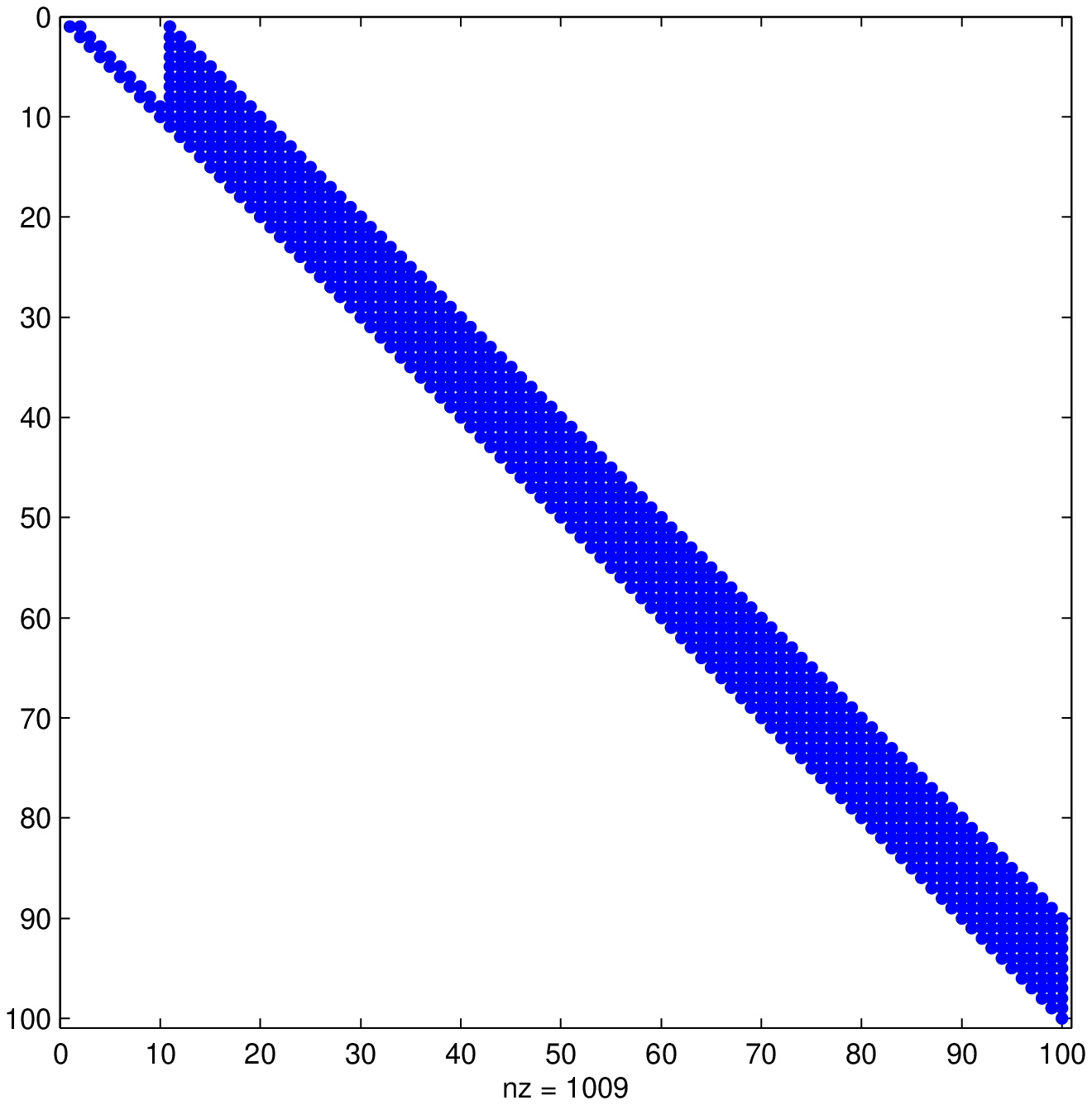} \label{fig: cholesky_Possion_L1} }
    \subfigure[]{\includegraphics[width=0.4\textwidth,height=0.4\textwidth]{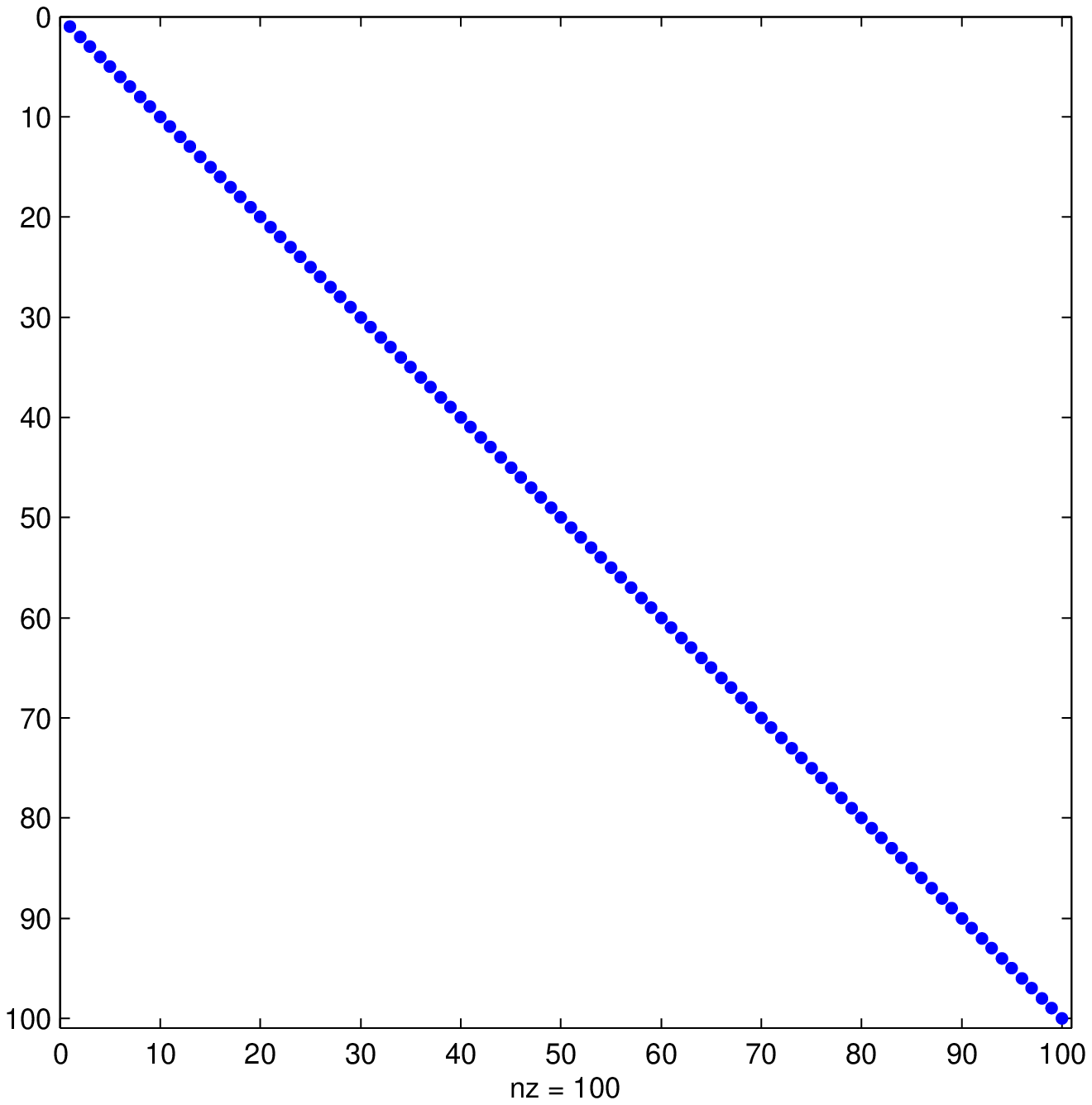} \label{fig: cholesky_Possion_L2} } 
    \subfigure[]{\includegraphics[width=0.4\textwidth,height=0.4\textwidth]{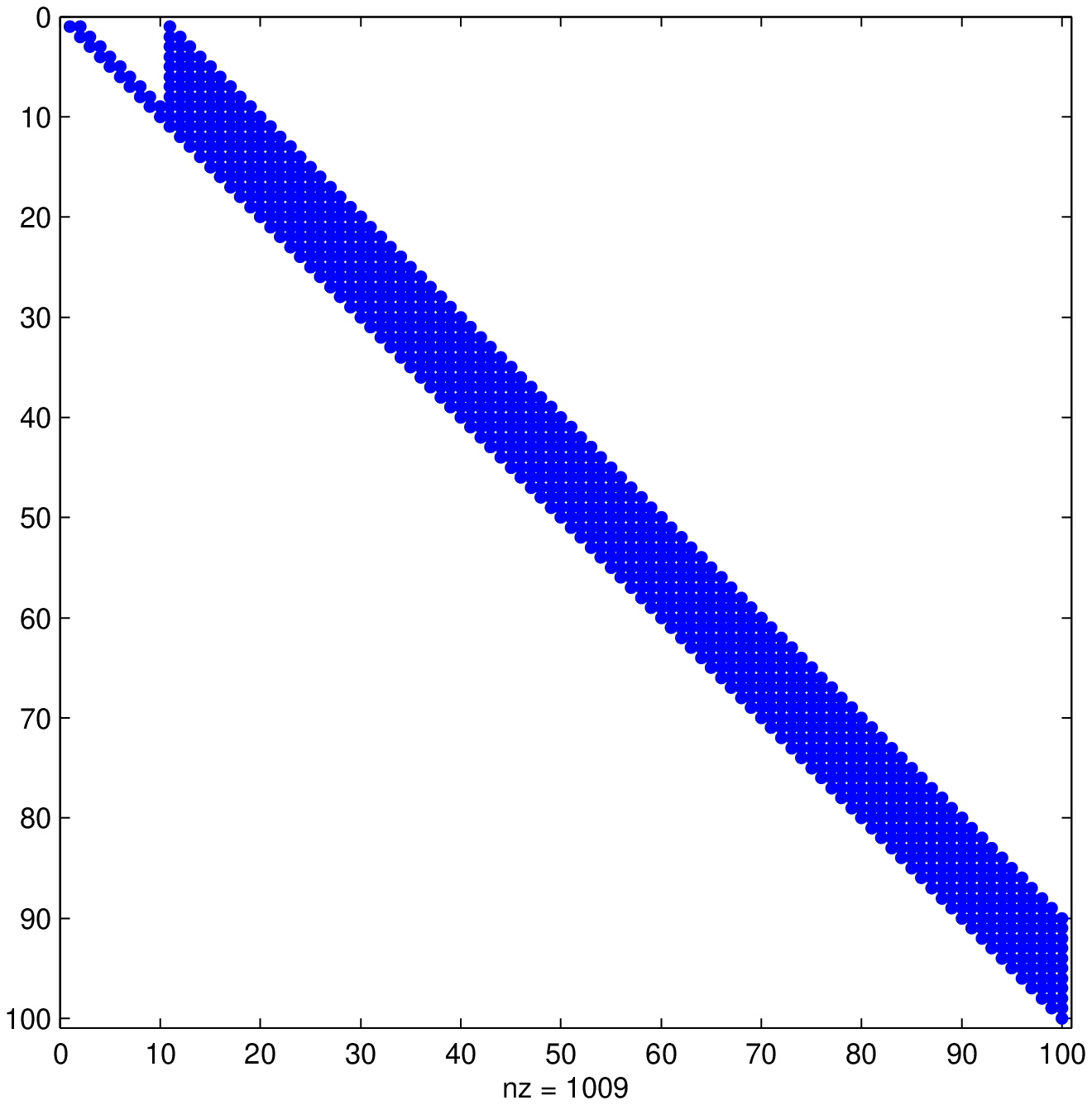} \label{fig: cholesky_Possion_L} }
    \subfigure[]{\includegraphics[width=0.2\textwidth,height=0.4\textwidth]{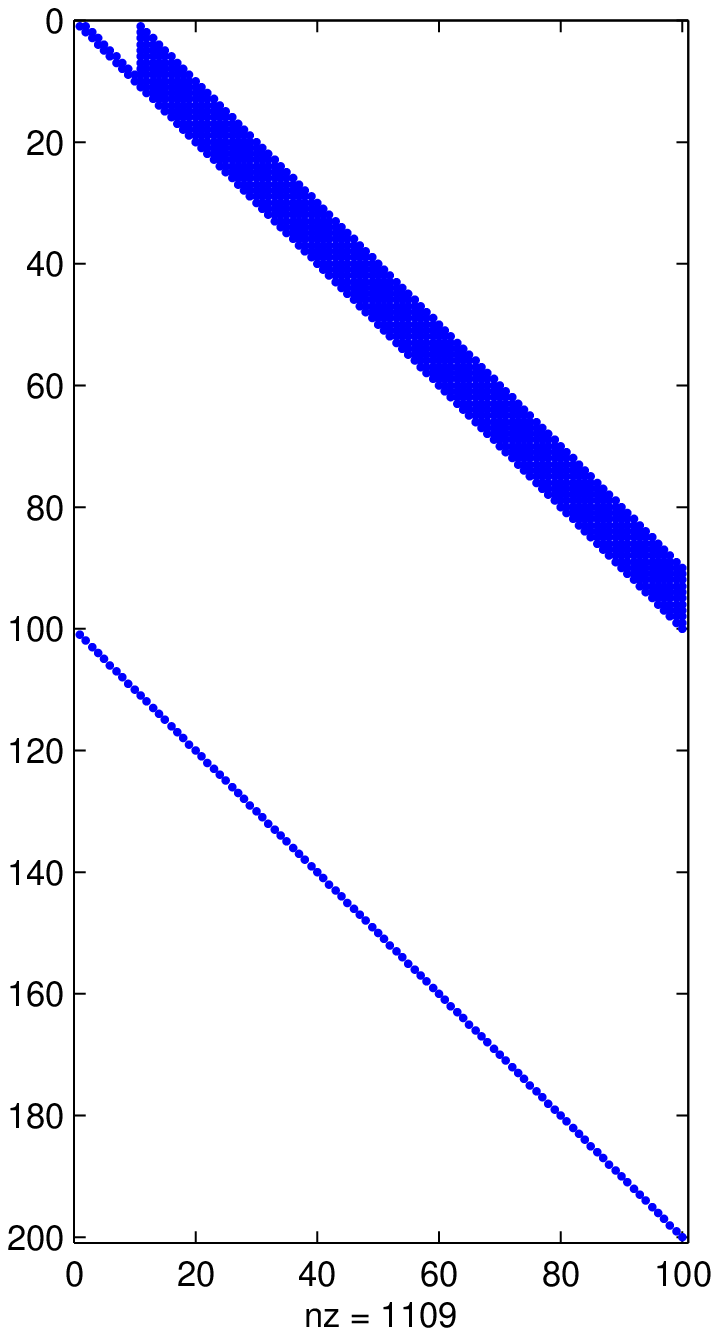} \label{fig: cholesky_Possion_A} }
    \subfigure[]{\includegraphics[width=0.2\textwidth,height=0.4\textwidth]{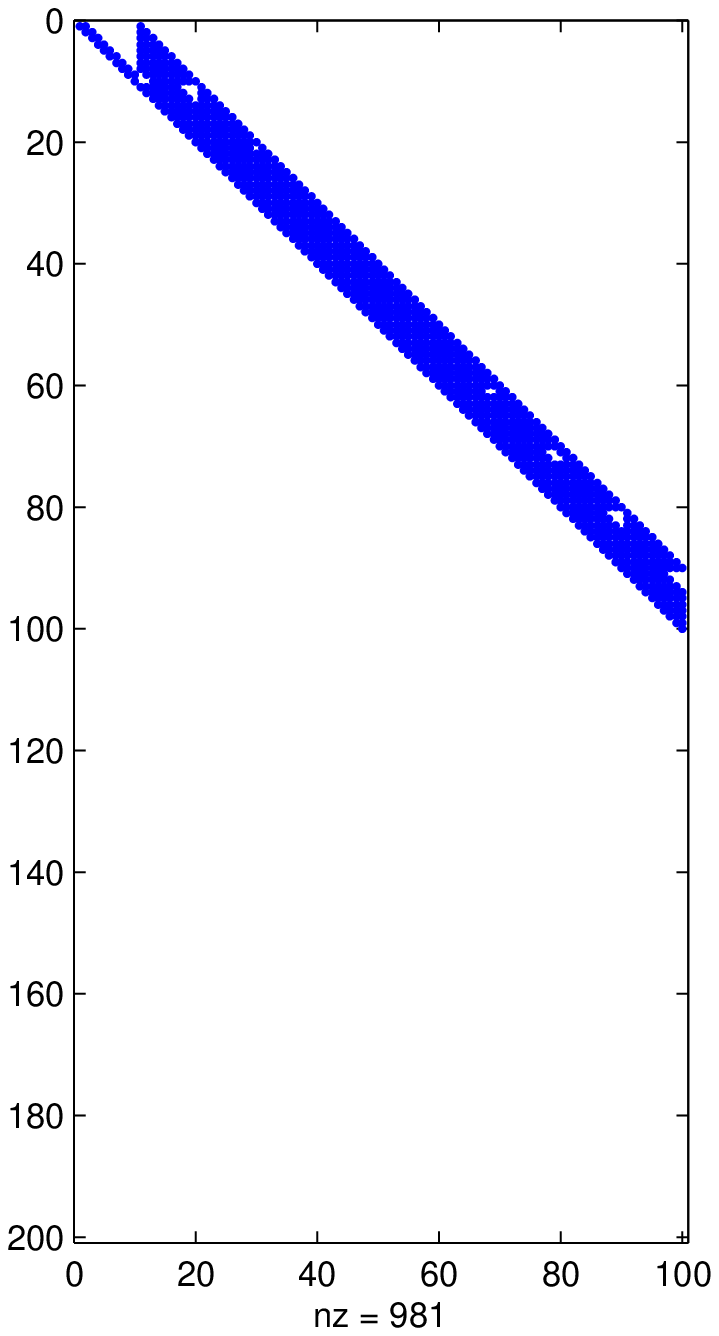} \label{fig: cholesky_Possion_R} }
    \caption{ Sparsity patterns of $\boldsymbol{L}_1$ (a), $\boldsymbol{L}_2$ (b), $\boldsymbol{L}$ (c), $\boldsymbol{A}$ (d) and $\boldsymbol{R}$ (e) with Poisson matrix }\label{fig: cholesky_Possion}
    \end{figure}

    \begin{figure}[htb]
    \centering
    \subfigure[]{\includegraphics[width=0.4\textwidth,height=0.4\textwidth]{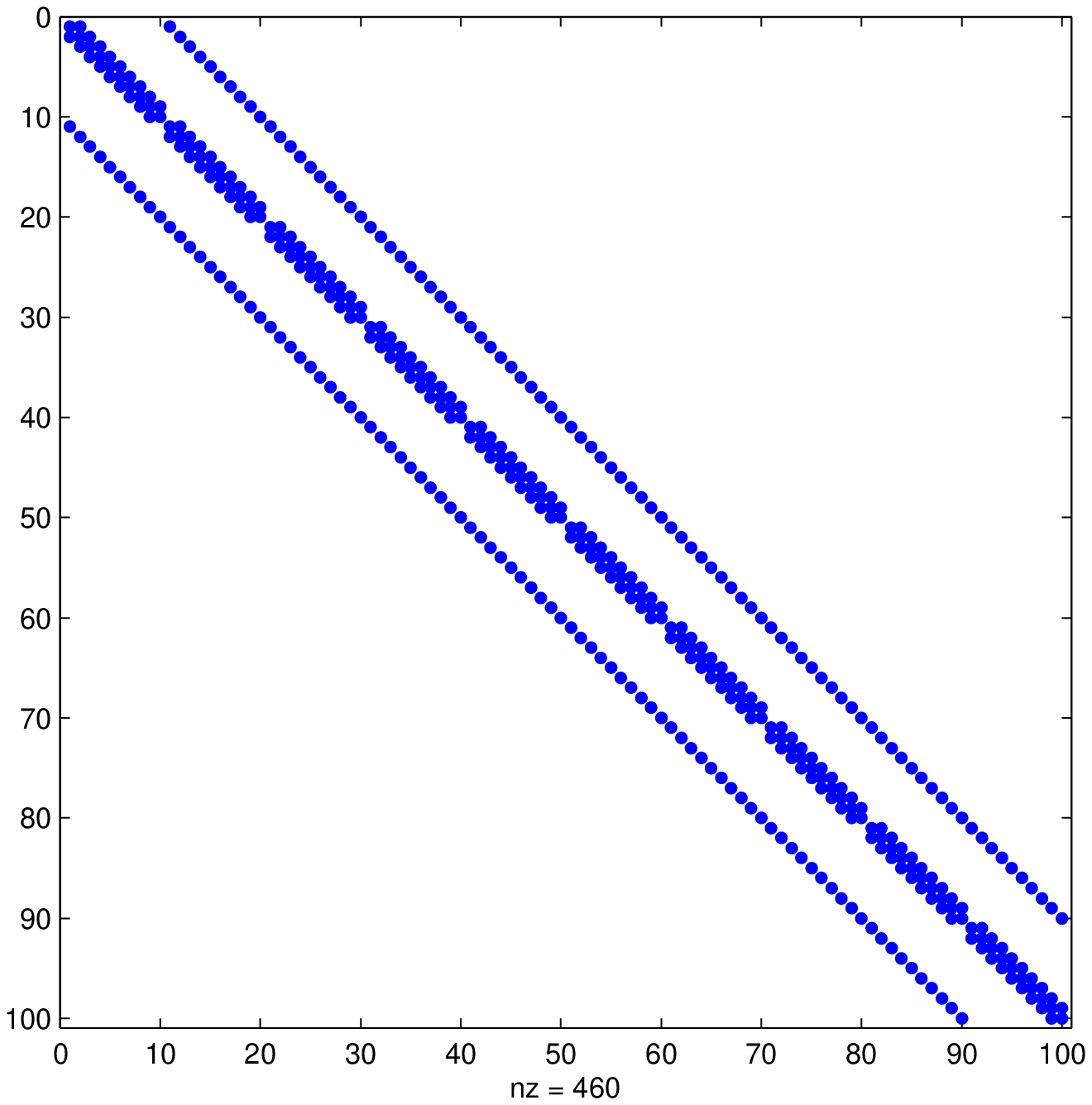} \label{fig: cholesky_Possion_Q} }
    \subfigure[]{\includegraphics[width=0.4\textwidth,height=0.4\textwidth]{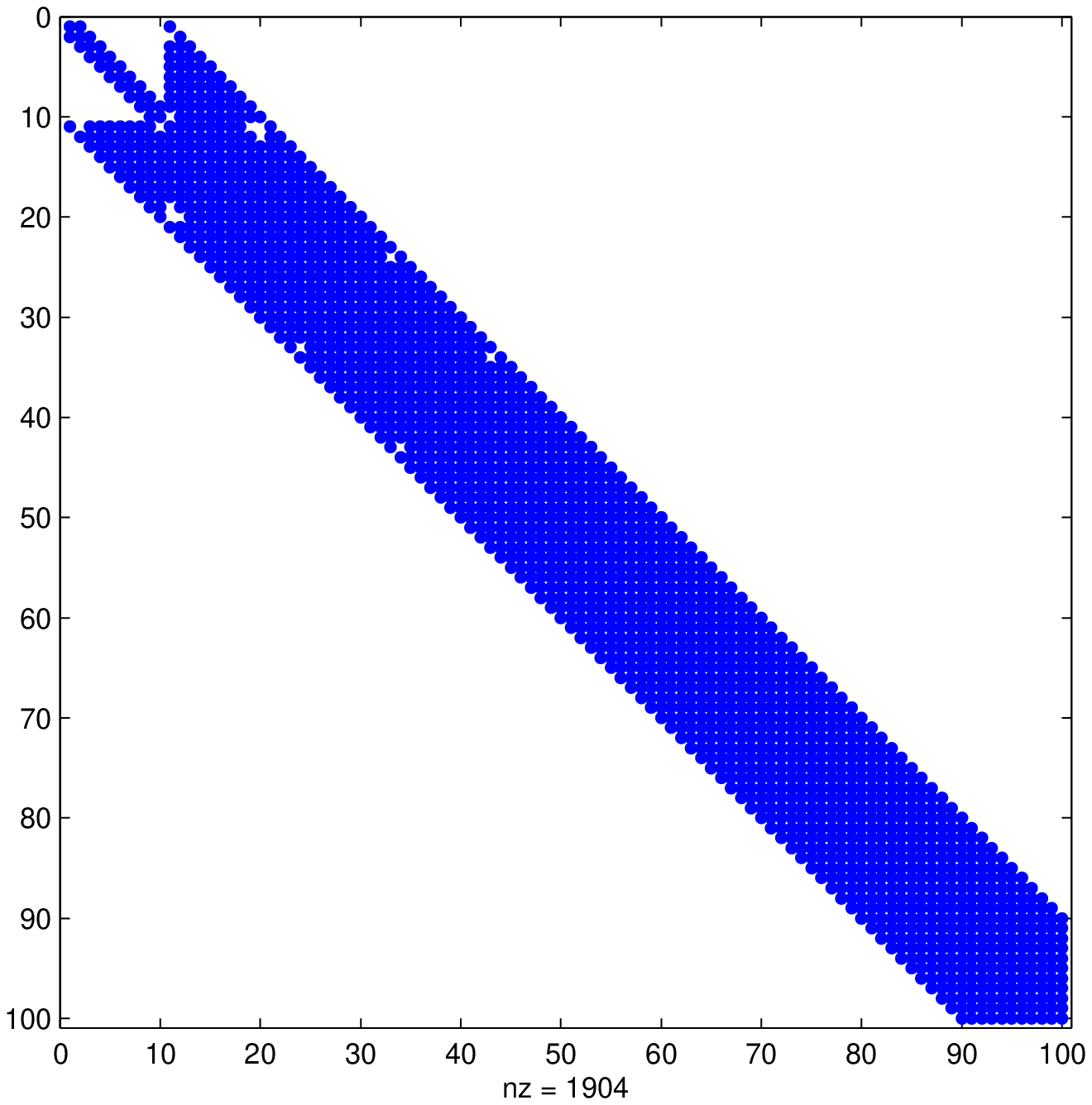} \label{fig: cholesky_Possion_Qq} }
    \caption{Sparsity patterns for the true precision matrix $\boldsymbol{Q}$ (a) and the approximated precision matrix $\boldsymbol{\widetilde{Q}}$ (b) with Poisson matrix} \label{fig: Cholsky_Possion_QandQq}
    \end{figure}

    \begin{figure}[htb]
    \centering
    \subfigure[]{\includegraphics[width=0.3\textwidth,height=0.3\textwidth]{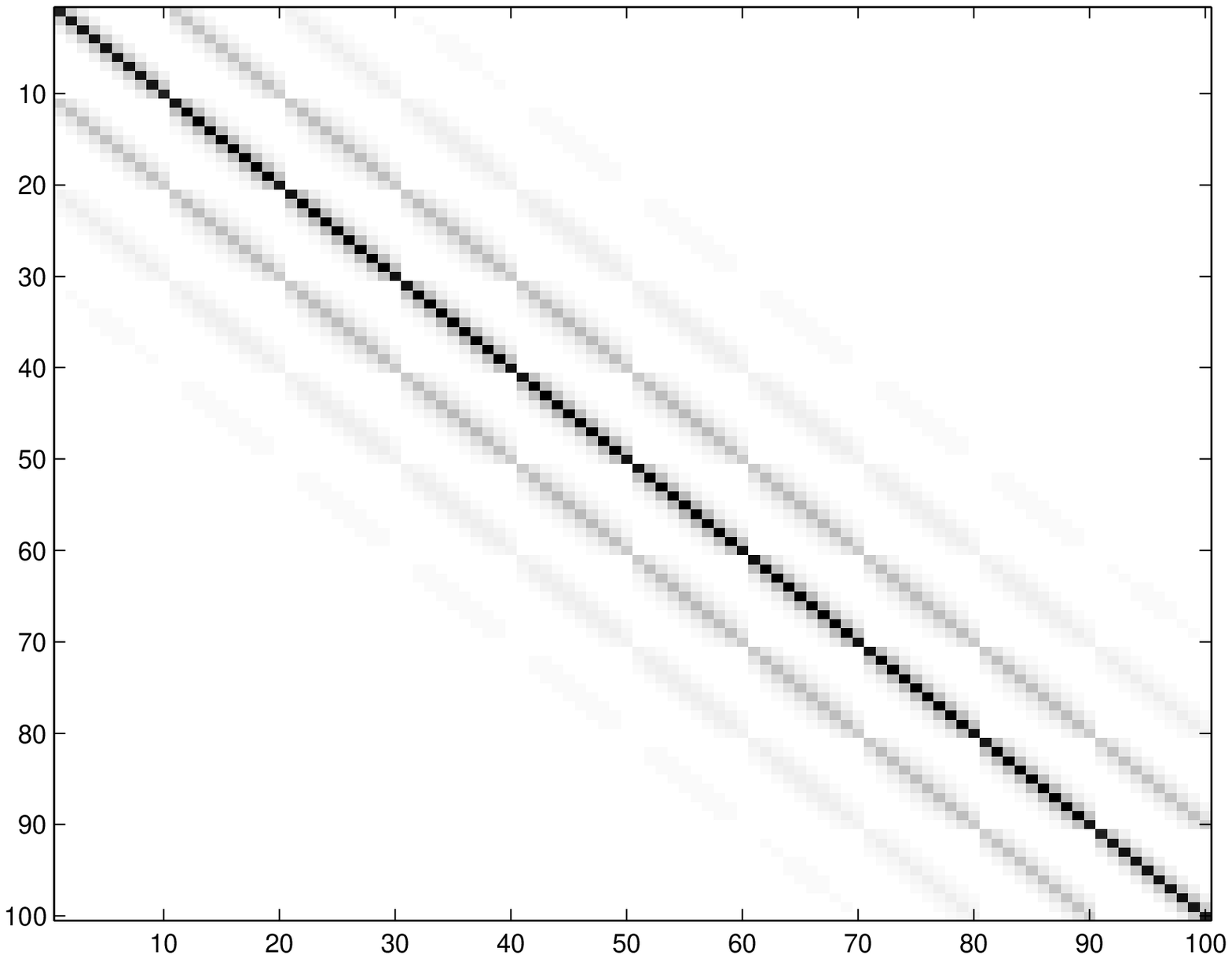} \label{fig: cholesky_Possion_InverseQ} }
    \subfigure[]{\includegraphics[width=0.3\textwidth,height=0.3\textwidth]{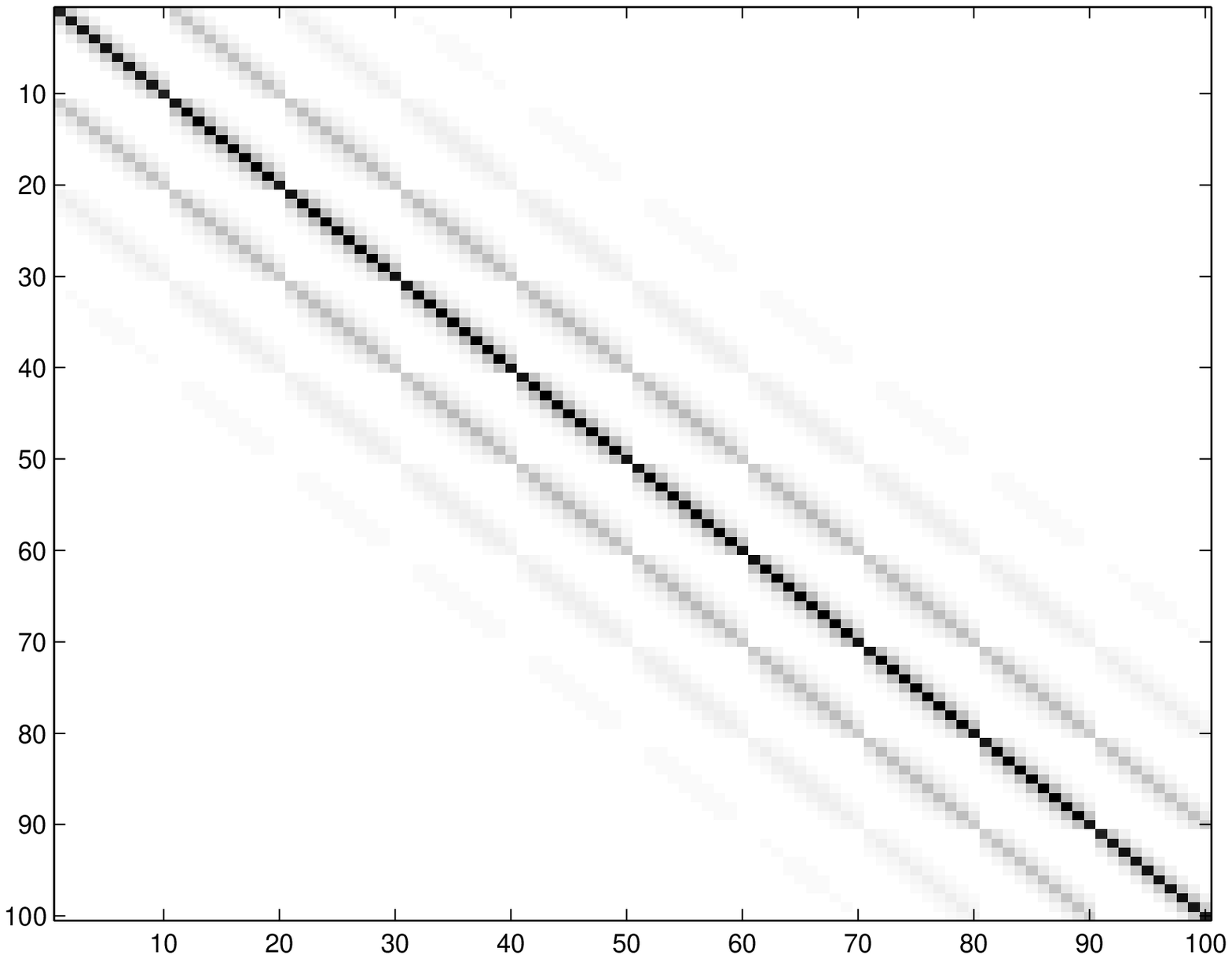} \label{fig: cholesky_Possion_InverseQq}}
    \subfigure[]{\includegraphics[width=0.3\textwidth,height=0.315\textwidth]{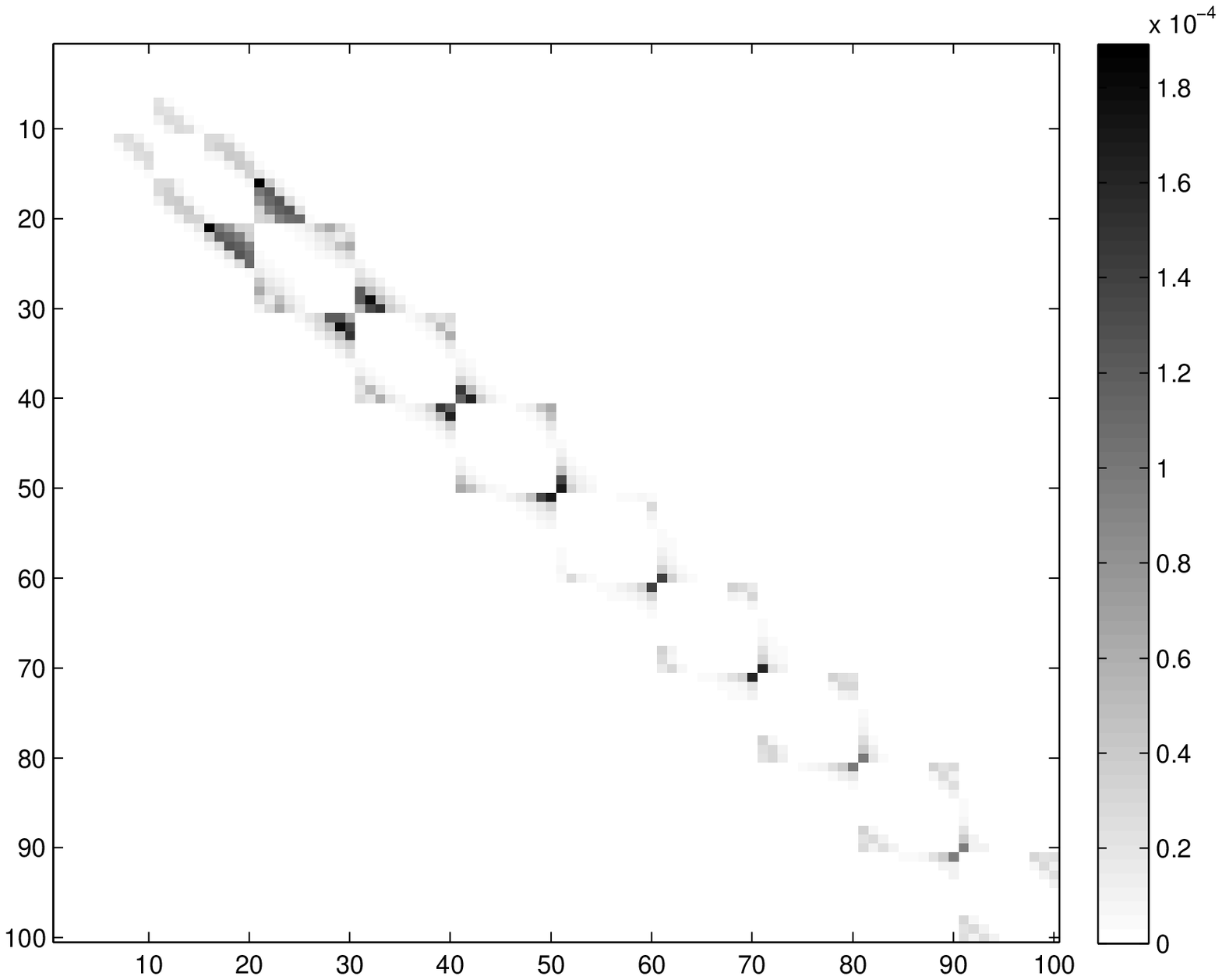} \label{fig: cholesky_Possion_Err}}
    \caption{Images of true covariance matrix $\boldsymbol{\Sigma}$ (a), the approximated covariance matrix $\boldsymbol{\widetilde{\Sigma}}$ (b), 
             and the error matrix $\widetilde{\boldsymbol{E}}$ (c) for Poisson matrix} \label{fig: cholesky_Possion_ImQ12QqR}
    \end{figure}

   The next example is a precision matrix with a nearly band matrix.   
   Assume that $\boldsymbol{Q}_1$ is a nearly banded matrix but with the values $\boldsymbol{Q}_1(1,n) = 1$ and $\boldsymbol{Q}_1(n,1) = 1$. We call this matrix as Toeplitz matrix in this paper.
   The sparsity pattern of this matrix is given in Figure \ref{fig: cholesky_Toeplitz_Q}.
   With the dropping tolerance $\tau = 0.0001$, we apply the cTIGO algorithm to the rectangular matrix $\boldsymbol{A}$.
   The sparsity patterns of $\boldsymbol{L}_1$, $\boldsymbol{L}_2$, $\boldsymbol{L}$, $\boldsymbol{A}$ and $\boldsymbol{R}$ are given in Figure \ref{fig: cholesky_Toeplitz_L1} - Figure \ref{fig: cholesky_Toeplitz_R}, respectively. 
   We notice that the upper triangular matrix $\boldsymbol{R}$ is sparser than the matrix $\boldsymbol{L}$.
   We can also notice that the sparseness of $\boldsymbol{R}$ depends on the tolerance $\tau$. The sparsity pattern of the approximated precision matrix $\boldsymbol{\widetilde{Q}}$ is given in Fig \ref{fig: cholesky_Toeplitz_Qq}.
    The image of the true covariance matrices $\boldsymbol{\Sigma}$, the approximated covariance matrix $\boldsymbol{\widetilde{\Sigma}}$, and the error
    matrix $\boldsymbol{\widetilde{E}}$ are shown in Figure \ref{fig: cholesky_Toeplitz_ImQ12QqR}. Note that the order of the numerical values in the error matrix is $10^{-5}$ with
    the given tolerance. More simulation results for this matrix are found in Section \ref{sec: cholesky_result_comparsion}.

    \begin{figure}[htb]
    \centering
    \subfigure[]{\includegraphics[width=0.4\textwidth,height=0.4\textwidth]{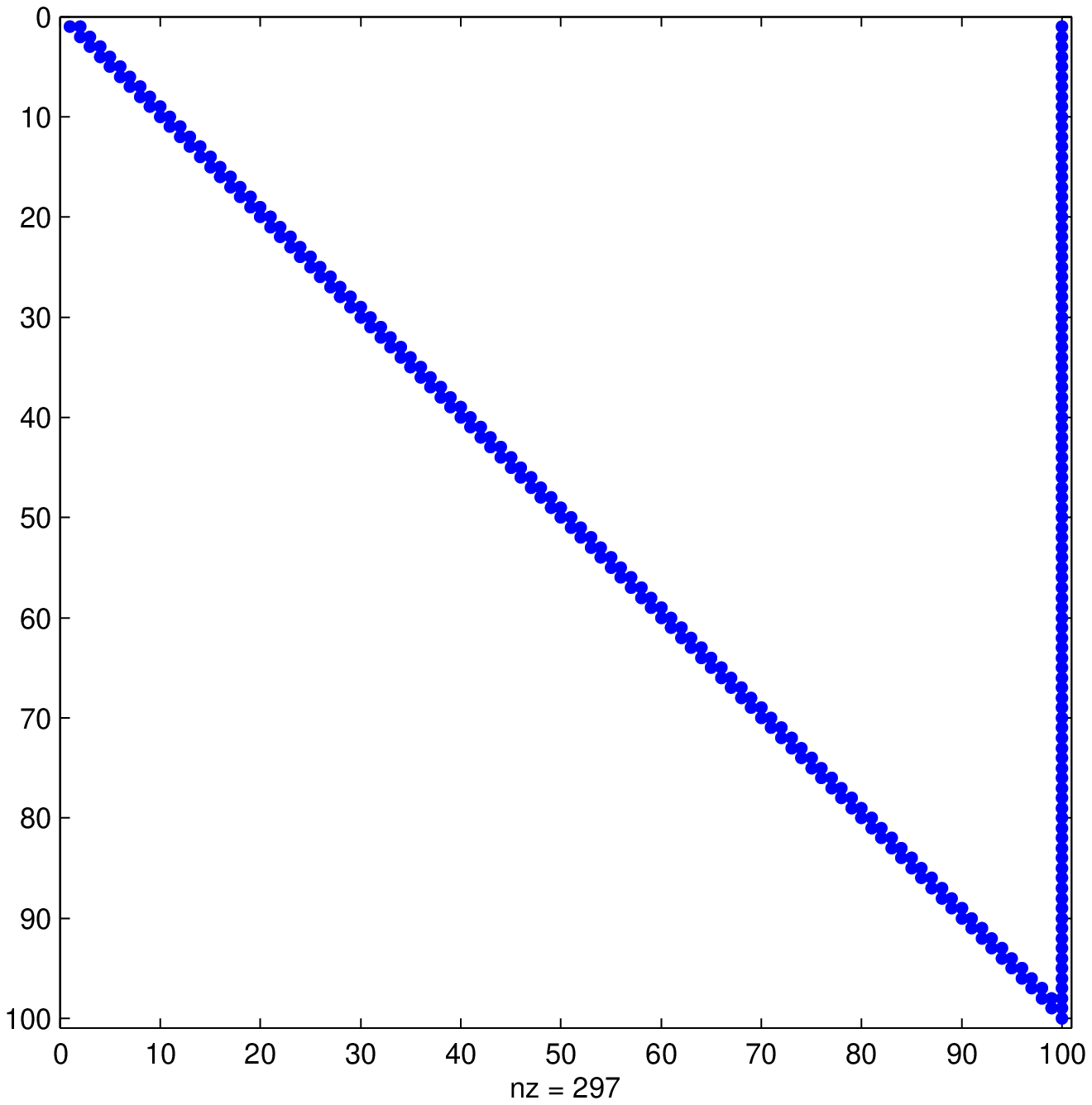} \label{fig: cholesky_Toeplitz_L1} }
    \subfigure[]{\includegraphics[width=0.4\textwidth,height=0.4\textwidth]{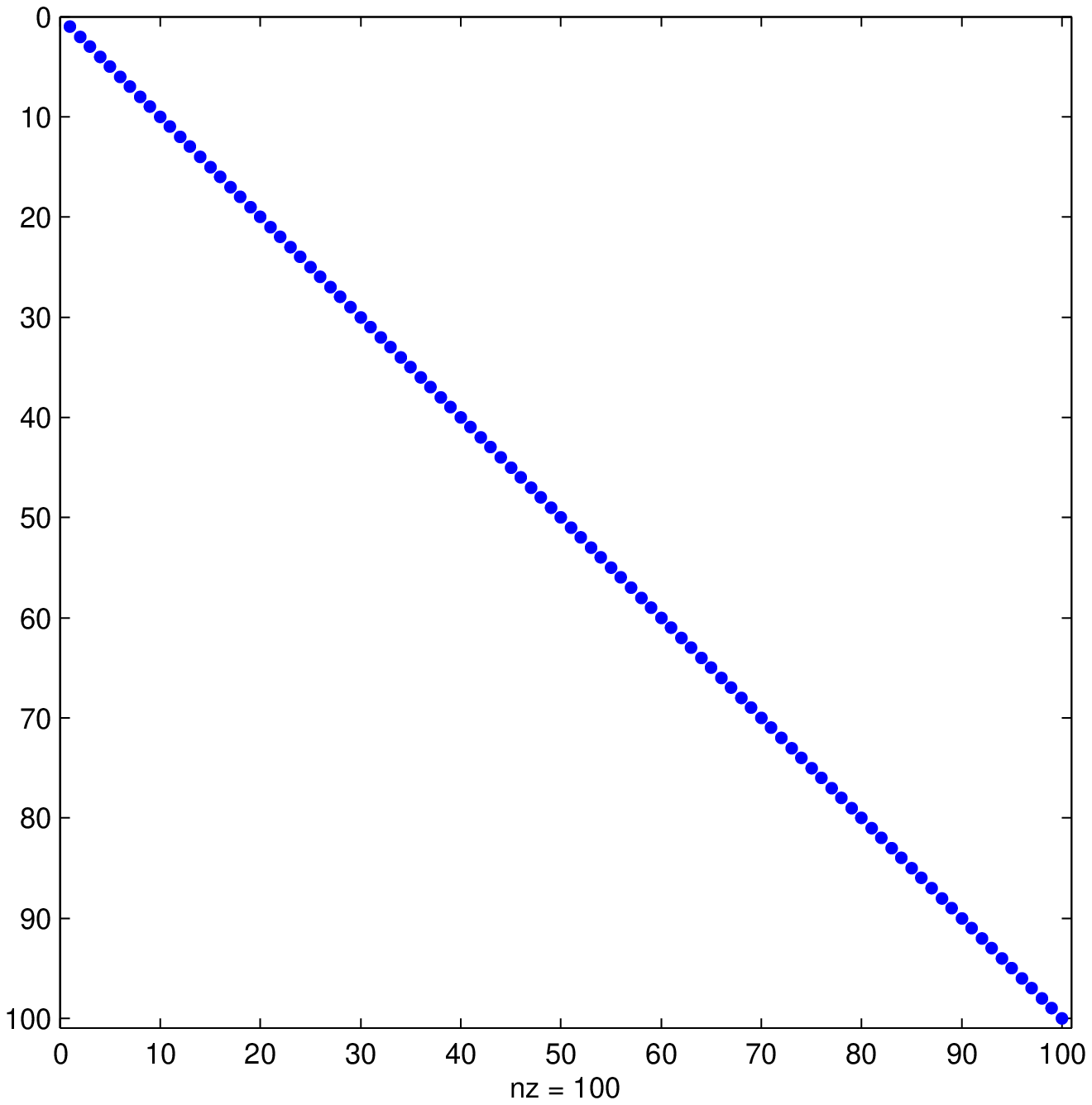} \label{fig: cholesky_Toeplitz_L2} }
    \subfigure[]{\includegraphics[width=0.4\textwidth,height=0.4\textwidth]{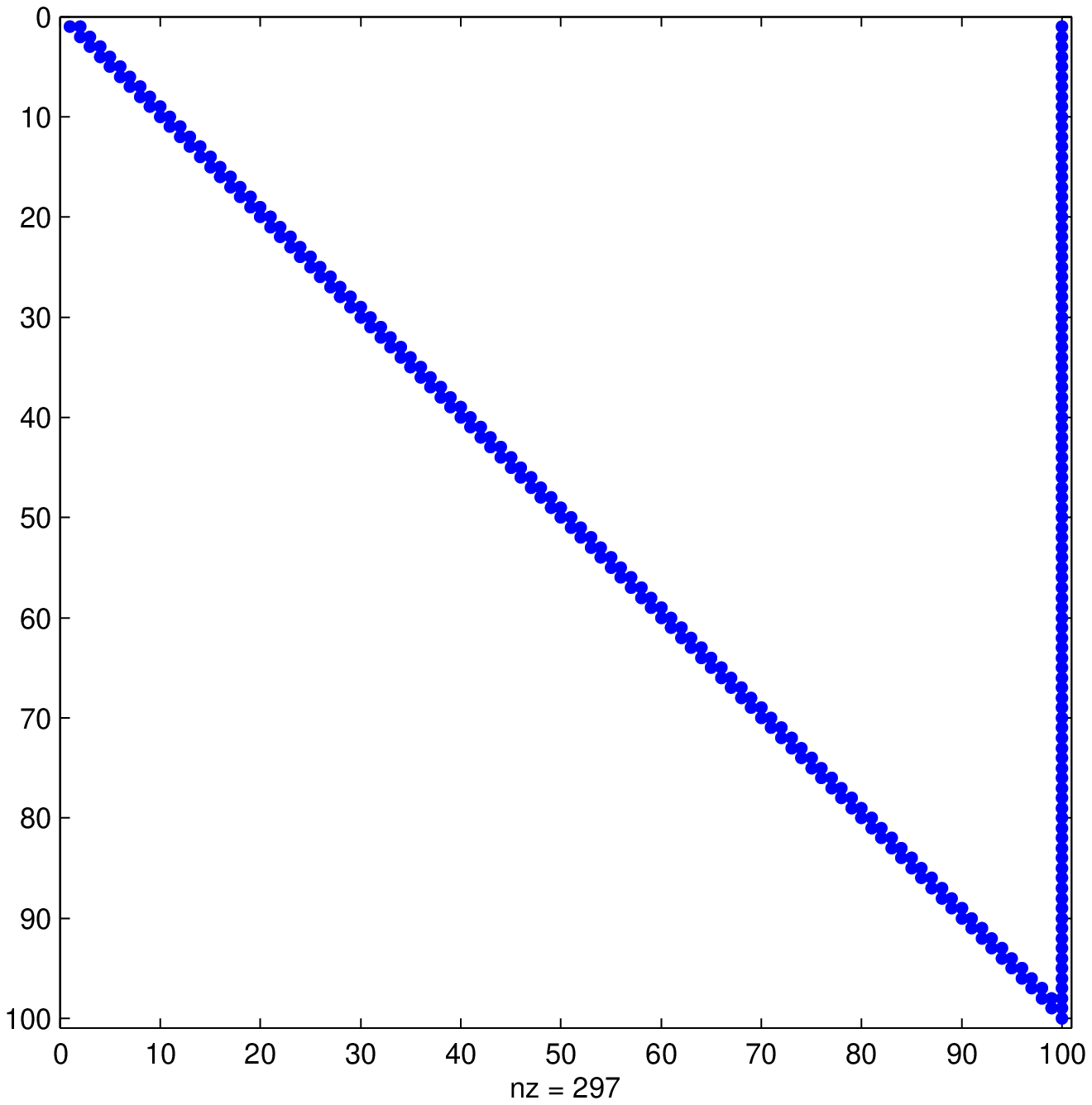} \label{fig: cholesky_Toeplitz_L} }
    \subfigure[]{\includegraphics[width=0.2\textwidth,height=0.4\textwidth]{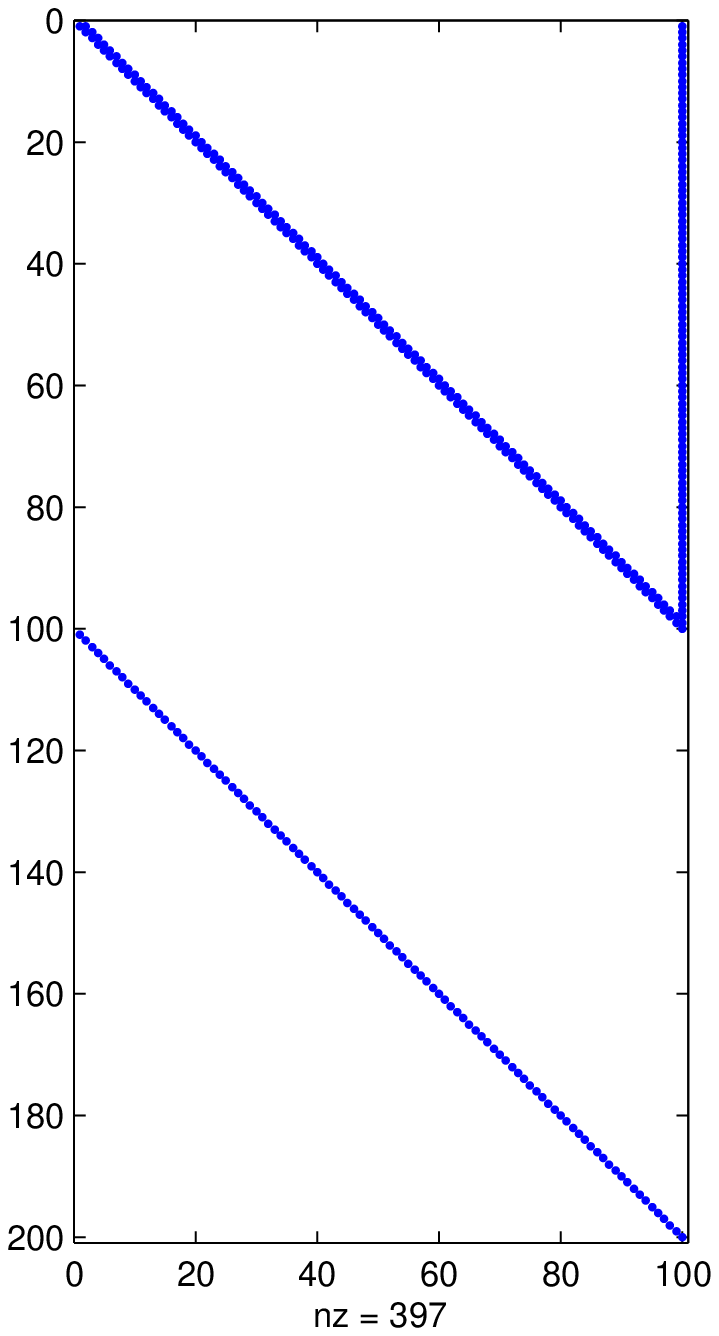} \label{fig: cholesky_Toeplitz_A} }
    \subfigure[]{\includegraphics[width=0.2\textwidth,height=0.4\textwidth]{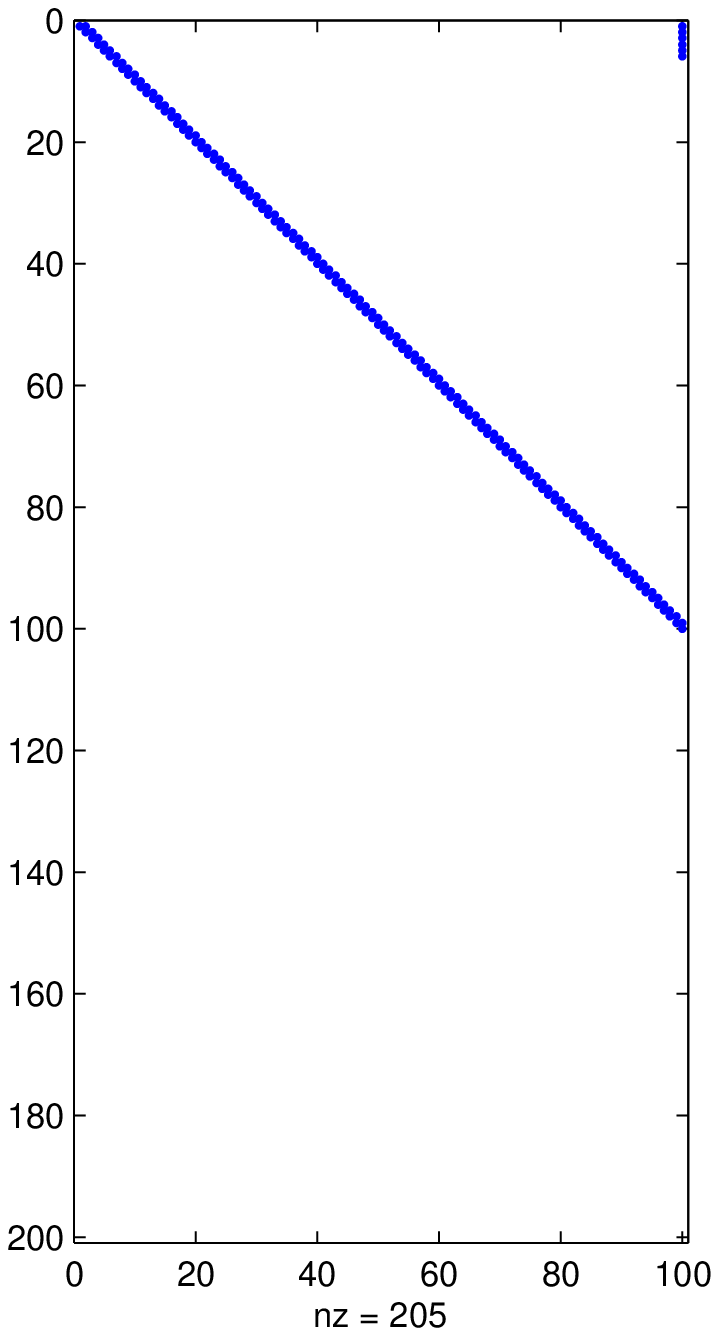} \label{fig: cholesky_Toeplitz_R} }
    \caption{ Sparsity patterns for $\boldsymbol{L}_1$ (a), $\boldsymbol{L}_2$ (b), $\boldsymbol{L}$ (c), $\boldsymbol{A}$ (d) and $\boldsymbol{R}$ for Toeplitz matrix.}\label{fig: cholesky_Toeplitz}
    \end{figure}

    \begin{figure}[htb]
    \centering
    \subfigure[]{\includegraphics[width=0.4\textwidth,height=0.4\textwidth]{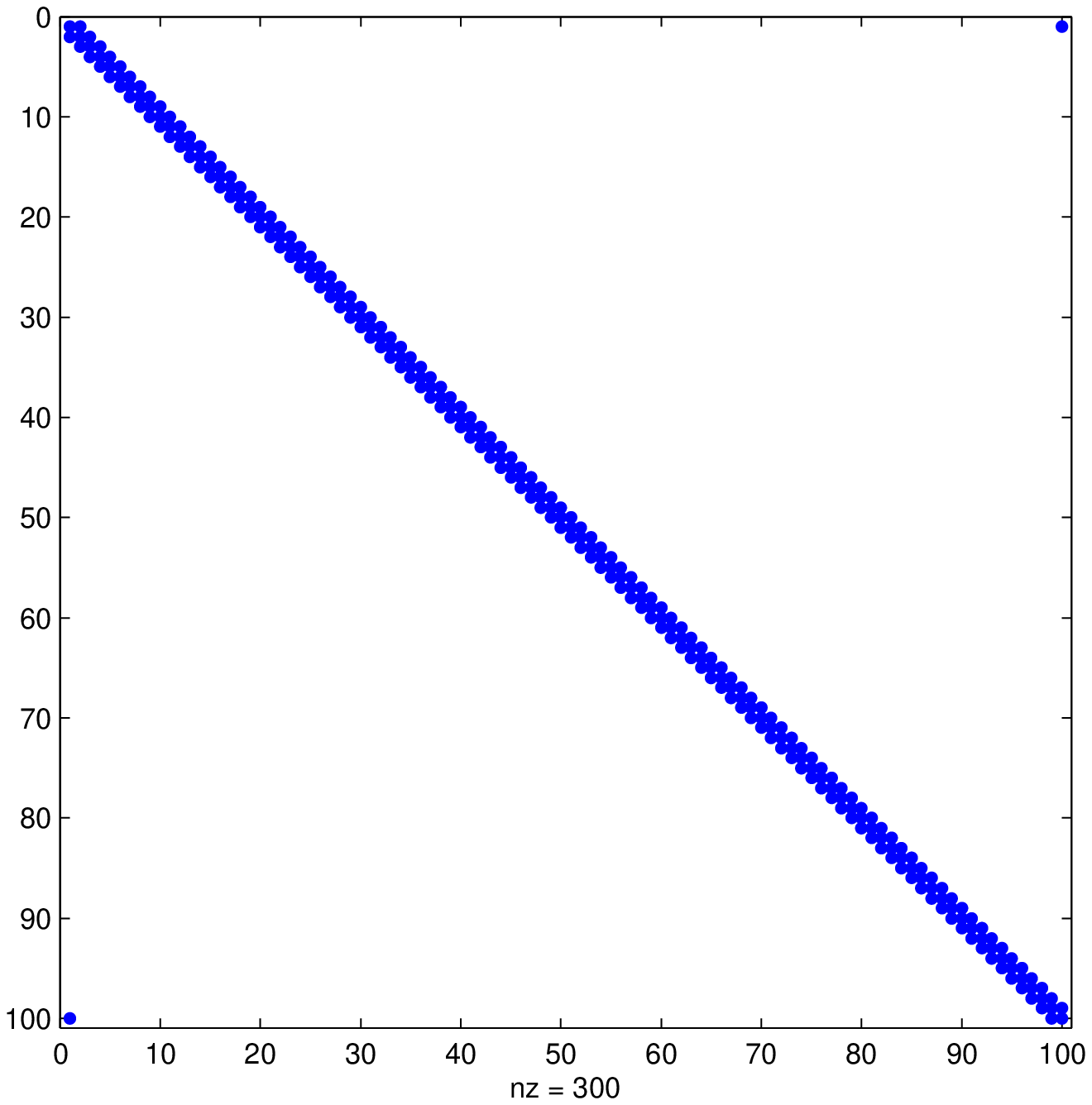} \label{fig: cholesky_Toeplitz_Q} }
    \subfigure[]{\includegraphics[width=0.4\textwidth,height=0.4\textwidth]{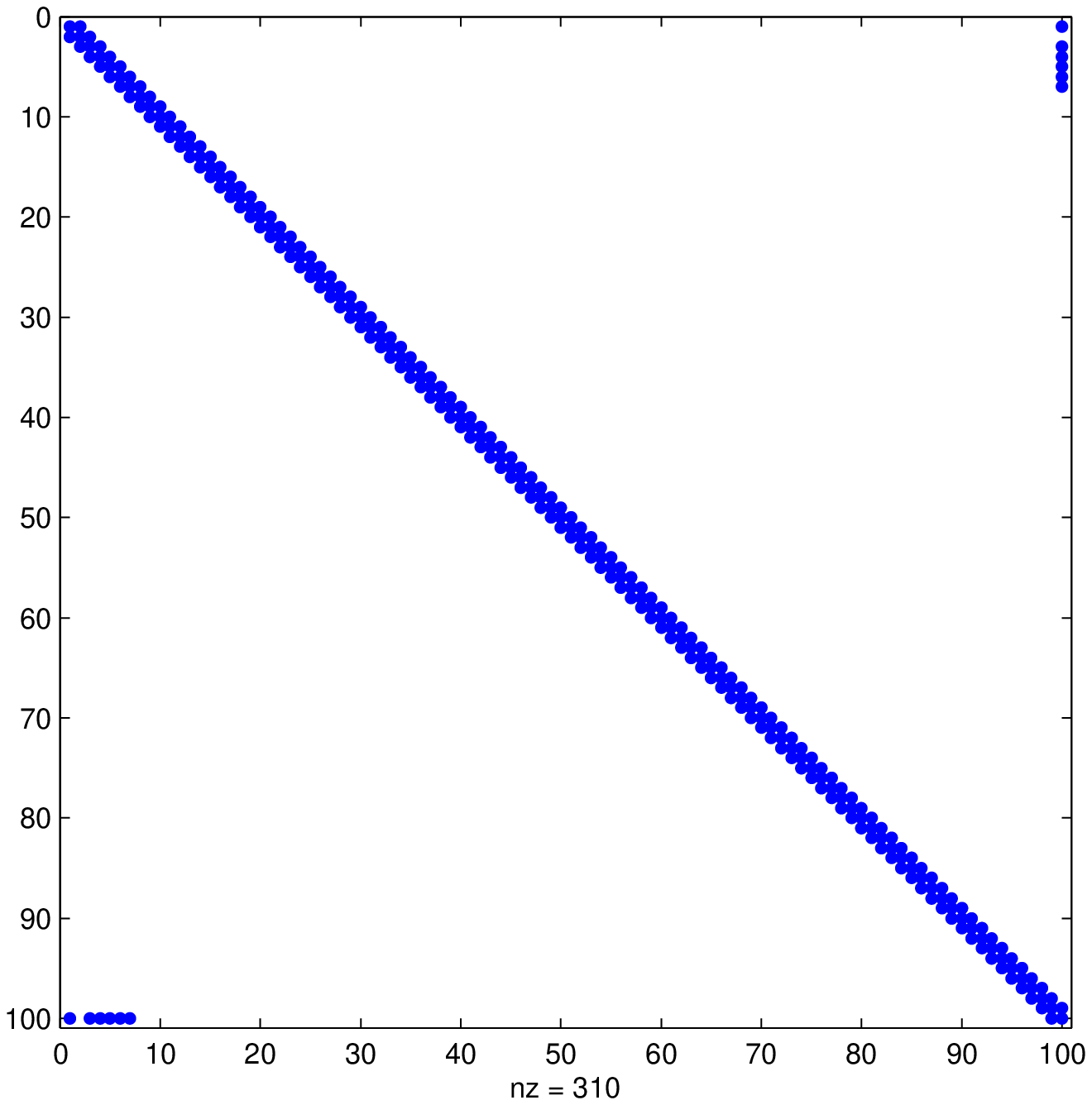} \label{fig: cholesky_Toeplitz_Qq} }
    \caption{Sparsity patterns of $\boldsymbol{Q}$ (a) and $\boldsymbol{Qq}$ (b) for Toeplitz matrix.} \label{fig: cholesky_Toeplitz_QandQq}
    \end{figure}

    \begin{figure}[htb]
    \centering
    \subfigure[]{\includegraphics[width=0.3\textwidth,height=0.3\textwidth]{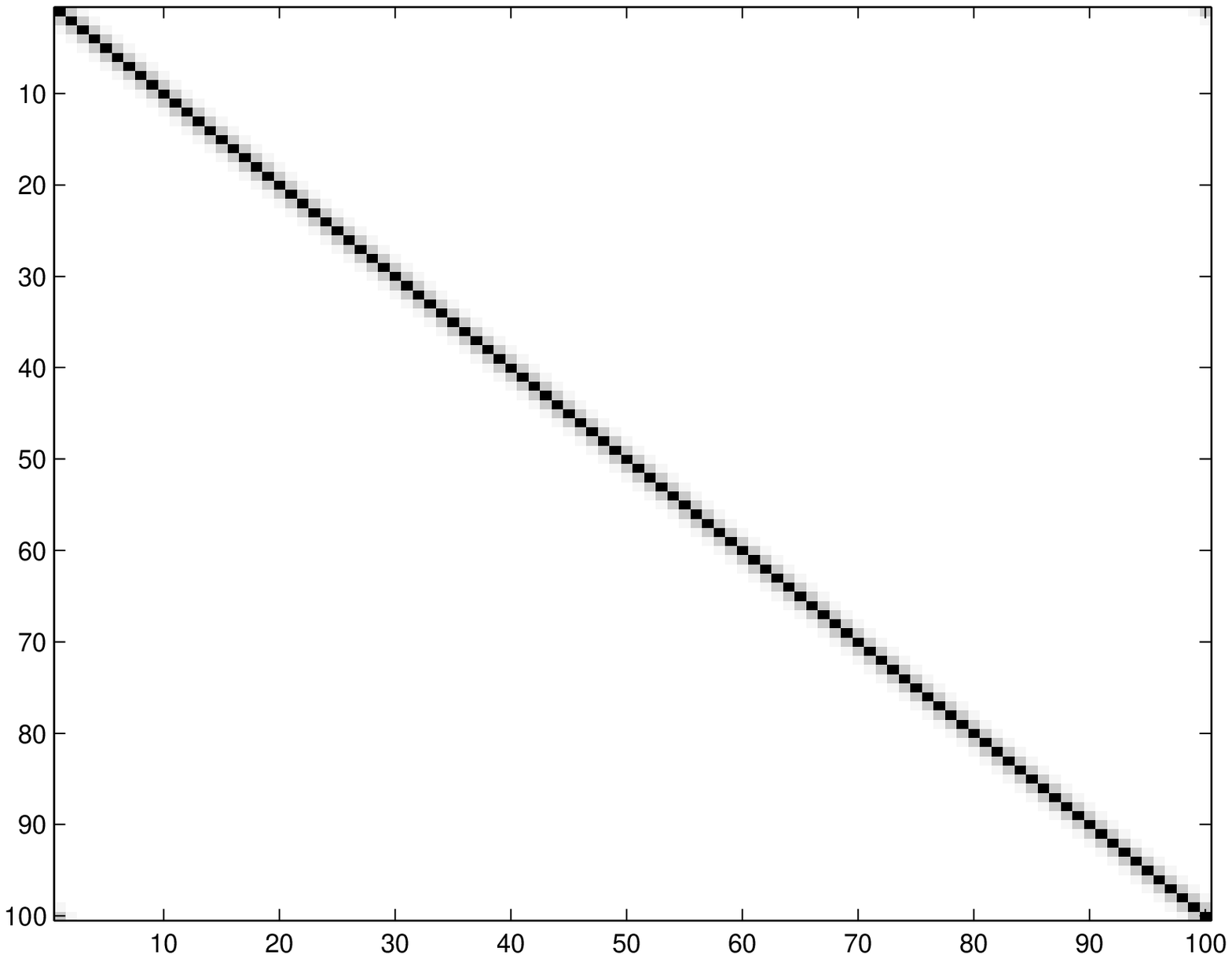} \label{fig: cholesky_Toeplitz_InverseQ} }
    \subfigure[]{\includegraphics[width=0.3\textwidth,height=0.3\textwidth]{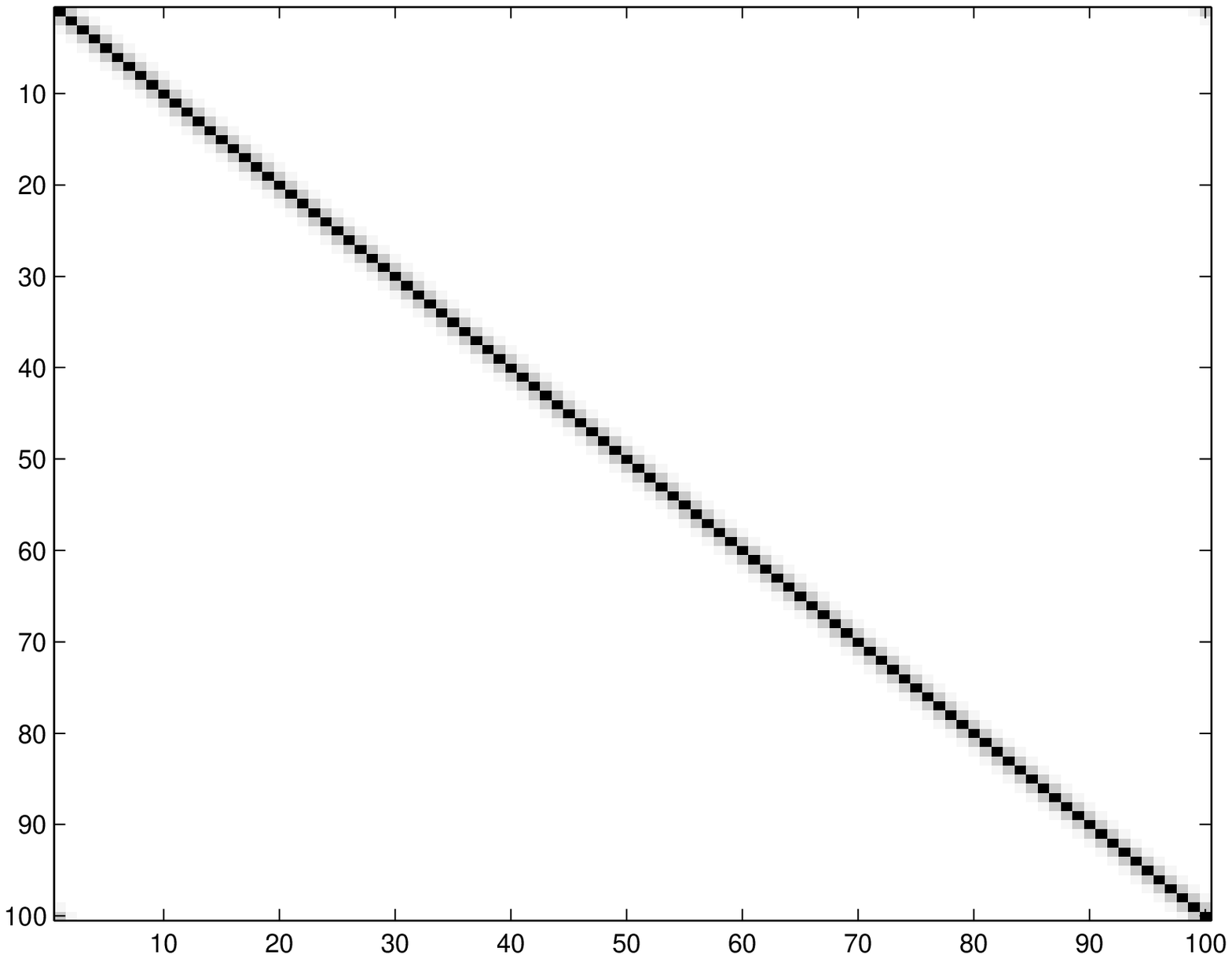} \label{fig: cholesky_Toeplitz_InverseQq} }
    \subfigure[]{\includegraphics[width=0.3\textwidth,height=0.315\textwidth]{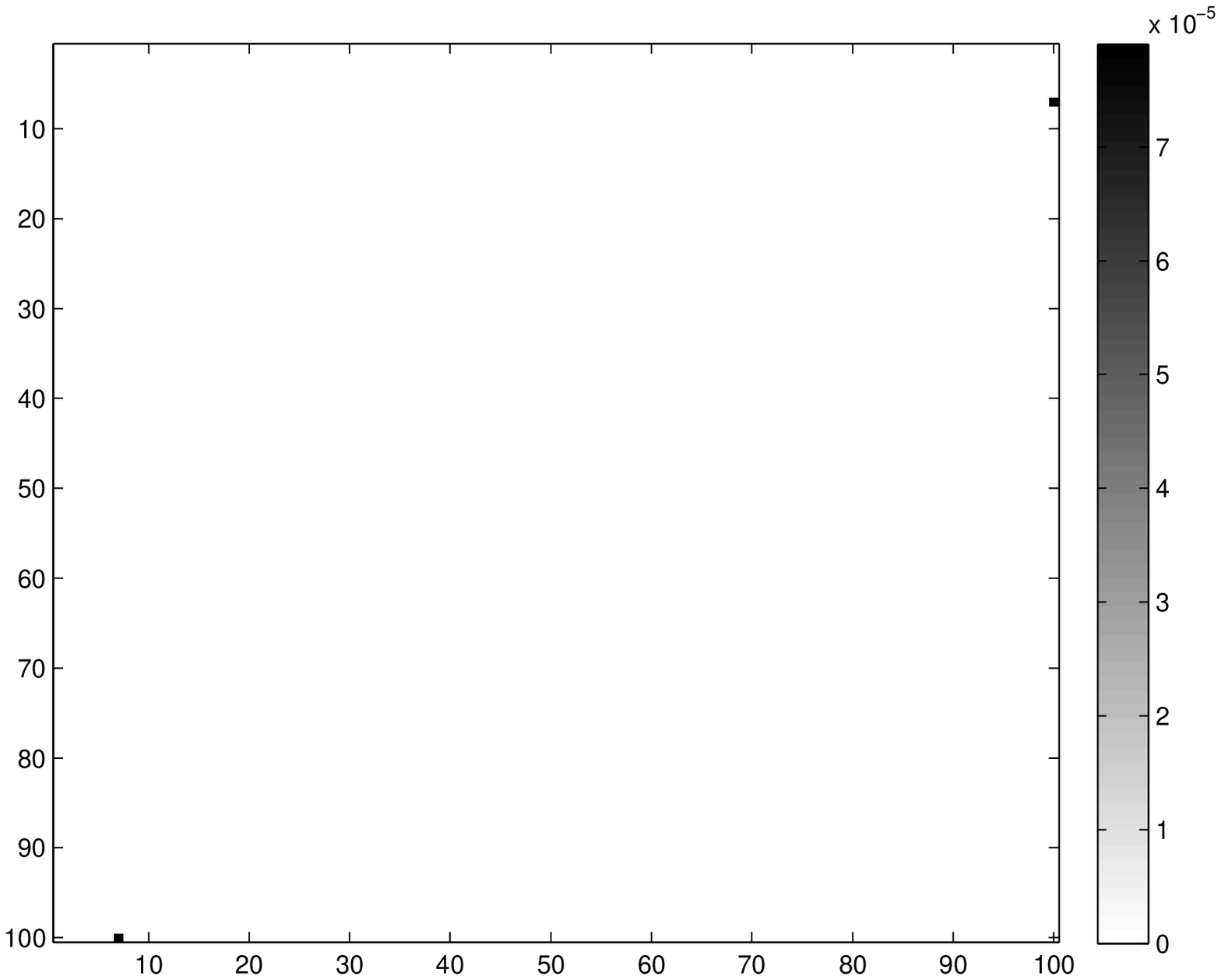} \label{fig: cholesky_Toeplitz_Err} }
    \caption{Images of the true covariance matrix $\boldsymbol{\Sigma}$ (a), the approximated covariance matrix $\boldsymbol{\widetilde{\Sigma}}$ (b) and the error matrix $\boldsymbol{\widetilde{E}}$ (c) 
             for Toeplitz matrix.} \label{fig: cholesky_Toeplitz_ImQ12QqR}
    \end{figure}

\subsection{Simulation results for particular precision matrices} \label{sec: cholesky_result_spde}
   In this section we emphasize on some particular precision matrices, namely the precision matrices from the stochastic partial differential equations (SPDEs) approach discussed by 
   \citet{lindgren2011explicit} and \citet{fuglstad2011spatial}. As pointed out by \citet{lindgren2011explicit} there is an explicit link between GRFs and GMRFs through SPDEs.
   The important relationship which was initially used by \citet{lindgren2011explicit} is that the solution $\boldsymbol{x}(\boldsymbol{u})$  to the following SPDE  
   is a GRF with Mat\'ern covariance function,
   \begin{equation} \label{eq: cholesky_spde1}
    (\kappa^2 - \Delta)^{\alpha/2} x(\boldsymbol{s}) = \mathcal{W}(\boldsymbol{s}), \hspace{2mm} \boldsymbol{s} \in \mathbb{R}^{d}, \hspace{2mm} \alpha = \nu +d/2, \hspace{2mm} \kappa > 0, \hspace{2mm} \nu > 0,
   \end{equation}
    where $\Delta = \sum_{i=1}^d \frac{\partial}{\partial x_i^2}$ is the Laplacian, $(\kappa^2 - \Delta)^{\alpha/2}$ is a differential operator and $d$ is the dimension of the field $x(\boldsymbol{s}) $. 
    \citet{fuglstad2011spatial} extended this approach to construct anisotropic and inhomogeneous fields with the SPDE 
    \begin{equation} \label{eq: cholesky_spde2}
      \kappa^2(\boldsymbol{u}) x(\boldsymbol{u}) - \nabla \cdot \boldsymbol{H}(\boldsymbol{u}) \nabla x(\boldsymbol{u}) = \mathcal{W} (\boldsymbol{u}), \
    \end{equation}
    where $\kappa$ and $\boldsymbol{H}$ control the local range and anisotropy, and $\nabla =\left( \frac{\partial}{\partial{x}}, \frac{\partial}{\partial{y}} \right)$.
   One important difference between \citet{lindgren2011explicit} and \citet{fuglstad2011spatial} is that \citet{lindgren2011explicit} have chosen the Neumann boundary condition but \citet{fuglstad2011spatial} has chosen the 
    periodic boundary condition. With Neumann boundary condition the precision matrix $\boldsymbol{Q}_1$ is a band matrix. However, the periodic boundary condition gives elements `` in the corners''  of the precision matrix.
   \citet{hu2012multivariate} extended the approach to multivariate settings by using systems of SPDEs.
    For more information about the SPDE approach, We refer to \citet{lindgren2011explicit}, \citet{fuglstad2011spatial} and \citet{hu2012multivariate}. 

    First,  choose the precision matrix for $\boldsymbol{Q}_1$ that results from the discretization of the SPDE \eqref{eq: cholesky_spde1} with $\alpha = 2, \, d = 2$ and $\kappa = 0.3$. 
    The sparsity pattern of $\boldsymbol{Q}_1$ is given in Figure \ref{fig: cholesky_SPDE1_Q}. We still assume $\boldsymbol{Q}_2 = \boldsymbol{I}$. 
    The sparsity patterns of $\boldsymbol{L}_1$, $\boldsymbol{L}_2$, $\boldsymbol{L}$, $\boldsymbol{A}$ and $\boldsymbol{R}$ are given in Figure \ref{fig: cholesky_SPDE1_L1} -\ref{fig: cholesky_SPDE1_R}, respectively. 
    We notice that the upper triangular matrix $\boldsymbol{R}$ is sparser than the matrix $\boldsymbol{L}$.
    The sparsity pattern of the approximated precision matrix $\boldsymbol{\widetilde{Q}}$ is given in Figure \ref{fig: cholesky_SPDE1_Qq}.
    The images of the true covariance matrix $\boldsymbol{\Sigma}$, the approximated covariance matrix $\widetilde{\boldsymbol{\Sigma}}$
    and the error matrix $\widetilde{\boldsymbol{E}}$ are shown in Figure \ref{fig: cholesky_result_spde1_comparation}. We can notice that the elements of in the error matrix $\widetilde{\boldsymbol{E}}$ 
    are reasonably small.

    The second precision matrix for $\boldsymbol{Q}_1$ is generated from the SPDE \eqref{eq: cholesky_spde2} with $\kappa = 0.1$ and 
    \begin{displaymath}
     \boldsymbol{H} = 0.1 \times
     \begin{pmatrix}
       1   &     0.5   \\
       0.5 &      1
     \end{pmatrix}.
    \end{displaymath}
    The sparsity pattern of the precision matrix $\boldsymbol{Q}_1$ is given in 
    Figure \ref{fig: cholesky_SPDE2_Q}. We use the same $\boldsymbol{Q}_2$ as previous examples.  The sparsity patterns of $\boldsymbol{L}_1$, $\boldsymbol{L}_2$, $\boldsymbol{L}$, $\boldsymbol{A}$ and $\boldsymbol{R}$ 
    are given in Figure \ref{fig: cholesky_SPDE2_L1} - Figure \ref{fig: cholesky_SPDE2_R}, respectively. We can notice that the upper triangular matrix $\boldsymbol{R}$ is sparser than the matrix $\boldsymbol{L}$. 
    The sparsity pattern of the approximated precision matrix $\boldsymbol{\widetilde{Q}}$ is given in Figure \ref{fig: cholesky_SPDE2_Qq}.
    The images of the true covariance matrix $\boldsymbol{\Sigma}$, the approximated covariance matrix $\widetilde{\boldsymbol{\Sigma}}$
    and the error matrix $\widetilde{\boldsymbol{E}}$ are illustrated in Figure \ref{fig: cholesky_result_spde2_comparation}. We could notice that the order of the numerical values in the error matrix
    $\widetilde{\boldsymbol{E}}$ are also reasonably small in this case.
    
    \begin{figure}[htbp]
    \centering
    \subfigure[]{\includegraphics[width=0.4\textwidth,height=0.4\textwidth]{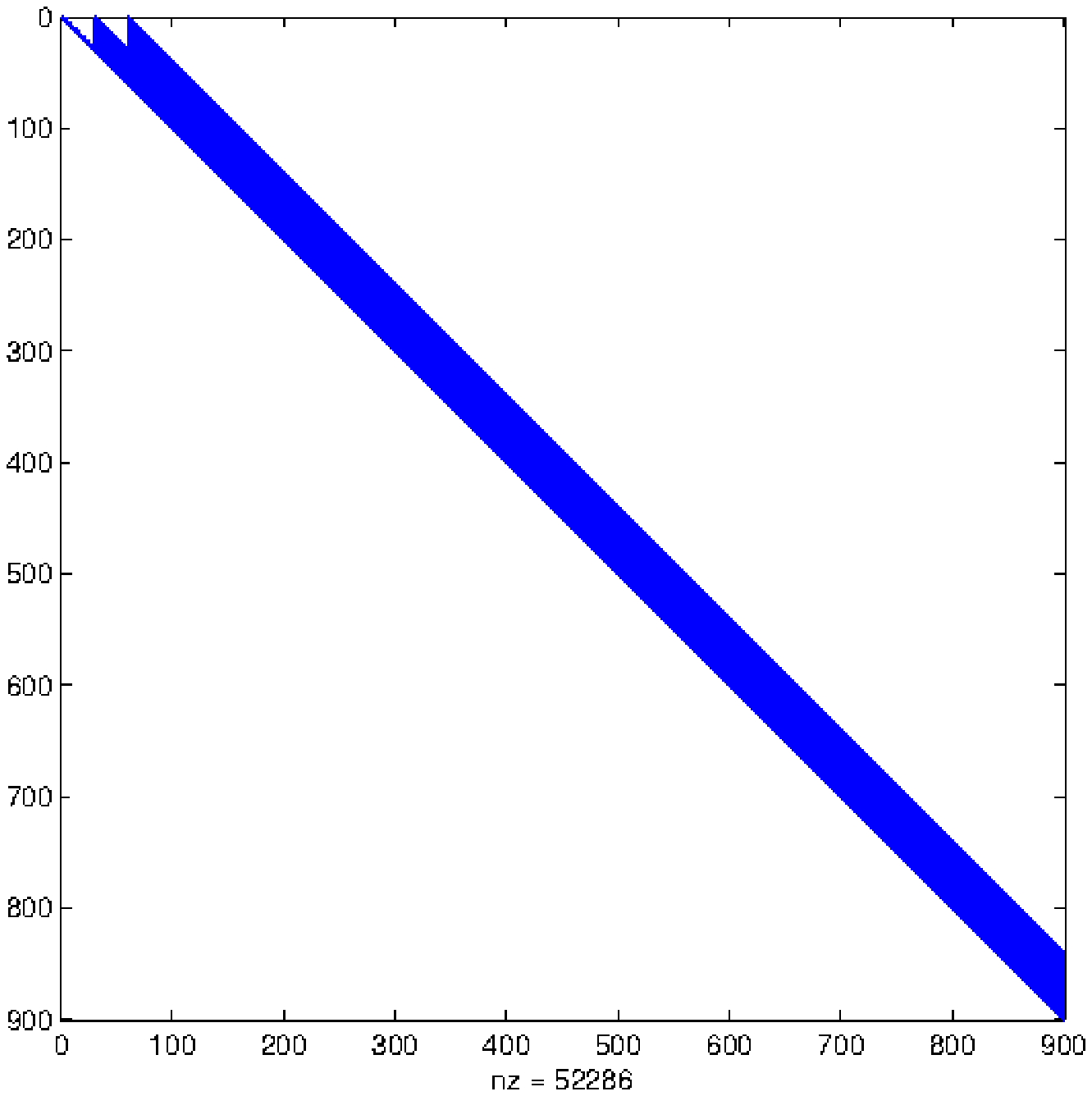} \label{fig: cholesky_SPDE1_L1}}
    \subfigure[]{\includegraphics[width=0.4\textwidth,height=0.4\textwidth]{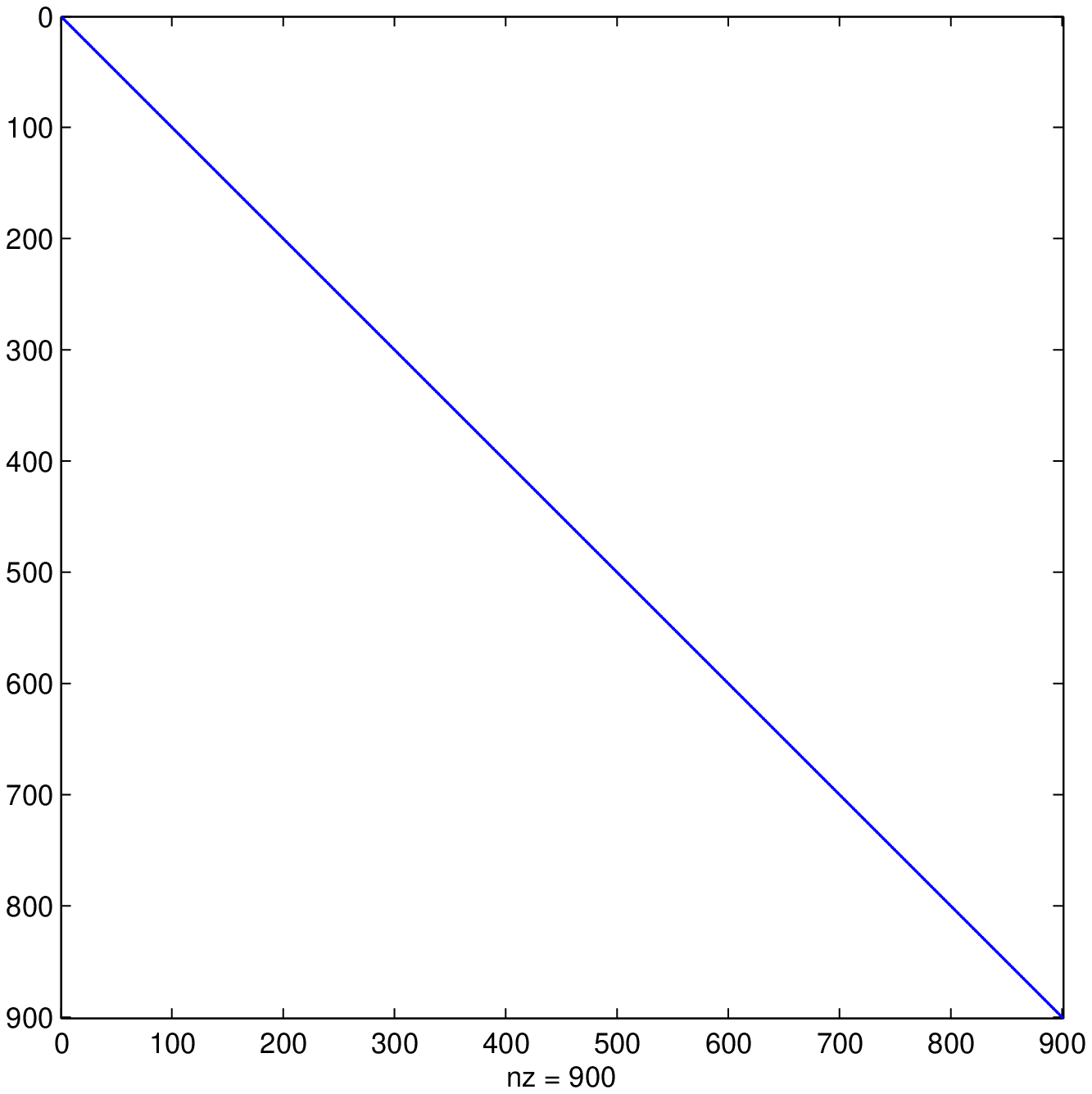} \label{fig: cholesky_SPDE1_L2}} 
    \subfigure[]{\includegraphics[width=0.4\textwidth,height=0.4\textwidth]{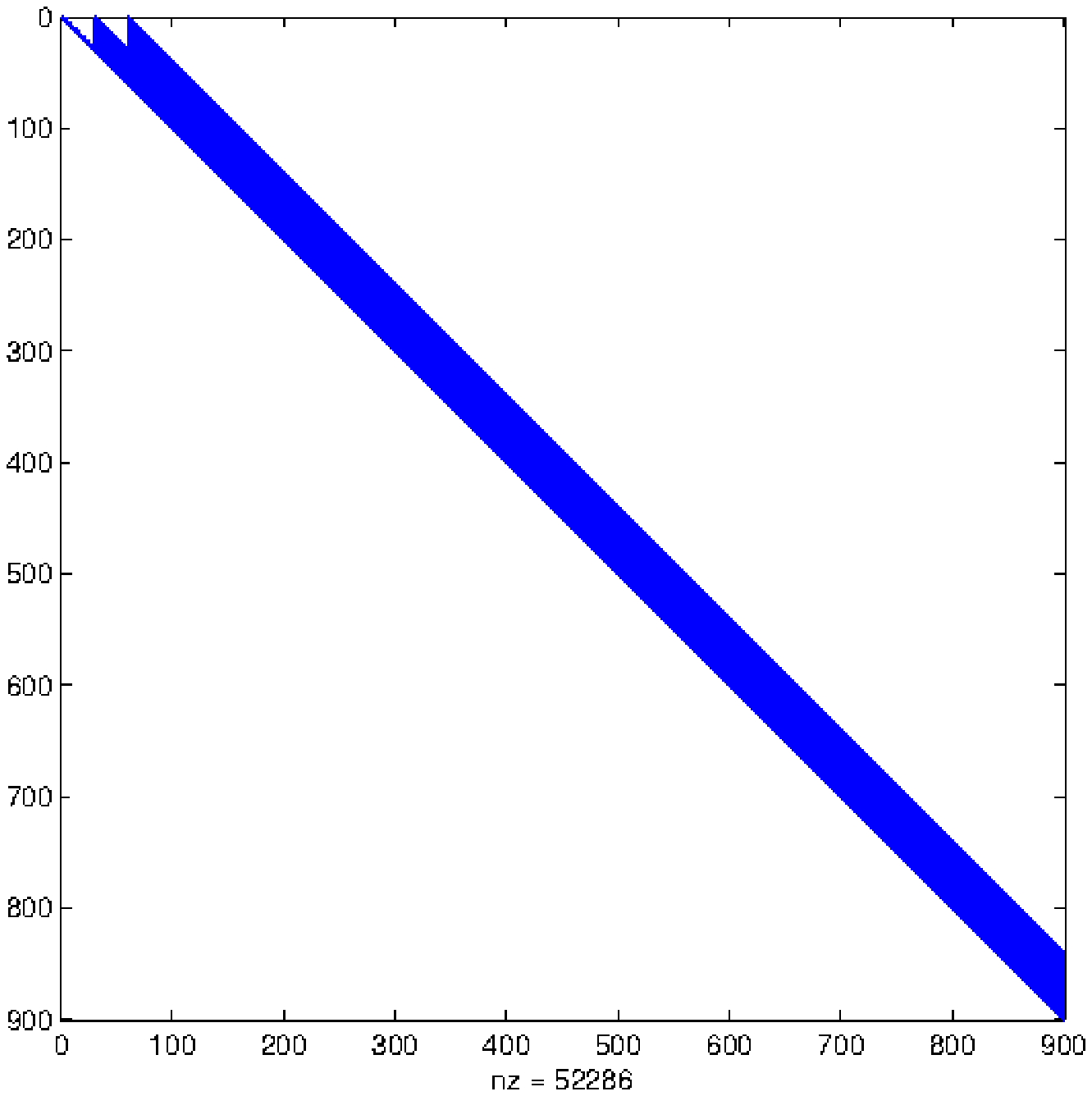} \label{fig: cholesky_SPDE1_L}}
    \subfigure[]{\includegraphics[width=0.2\textwidth,height=0.4\textwidth]{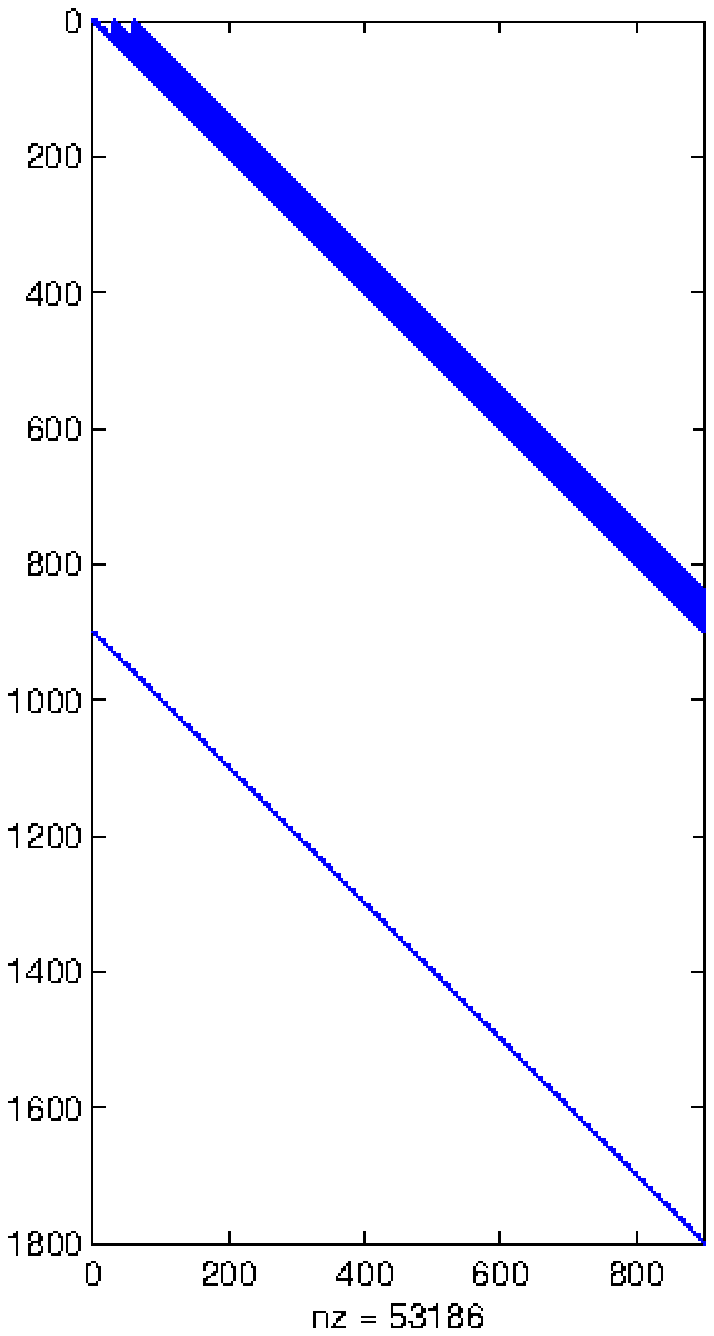} \label{fig: cholesky_SPDE1_A}}
    \subfigure[]{\includegraphics[width=0.2\textwidth,height=0.4\textwidth]{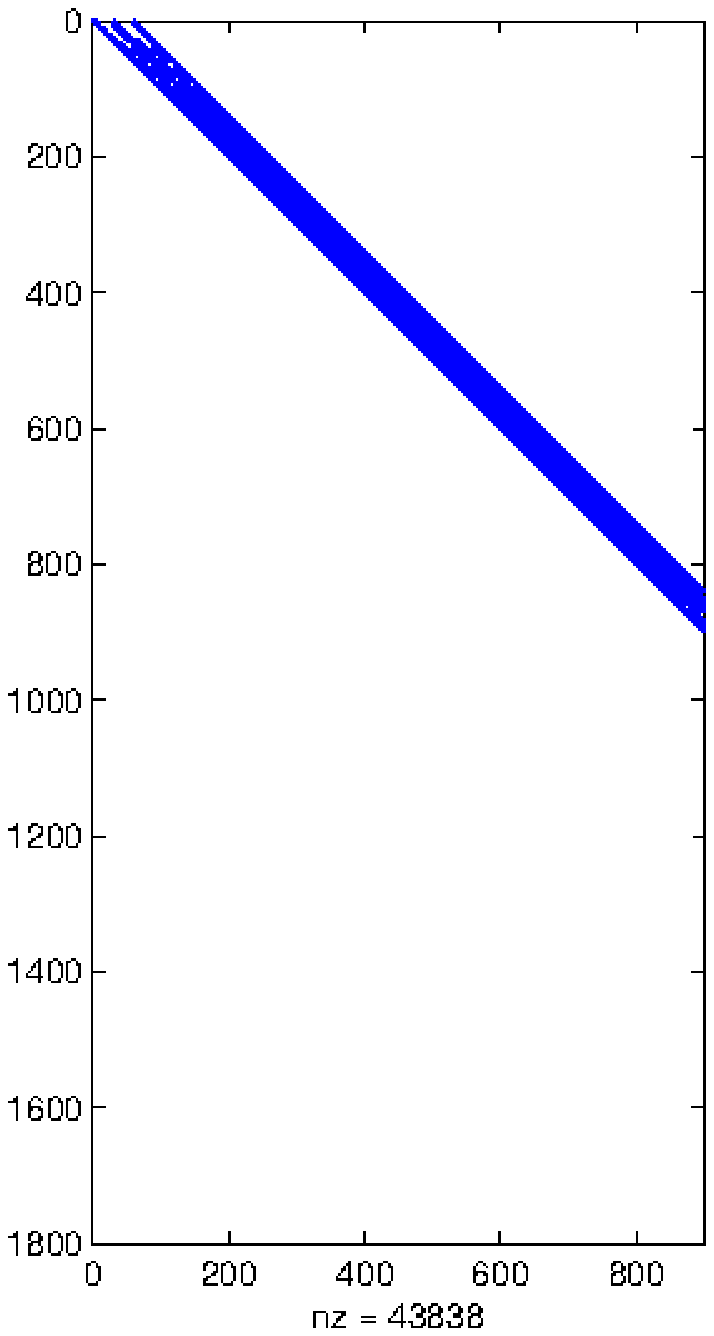} \label{fig: cholesky_SPDE1_R}}
    \caption{ Sparsity pattern for $\boldsymbol{L}_1$(a), $\boldsymbol{L}_1$(b), $\boldsymbol{L}_1$(c), $\boldsymbol{A}$ (d) and $\boldsymbol{R}$(e) for the random field from the SPDE \eqref{eq: cholesky_spde1}.}\label{fig: cholesly_SPDE1}
    \end{figure}

    \begin{figure}[htbp]
    \centering
    \subfigure[]{\includegraphics[width=0.4\textwidth,height=0.4\textwidth]{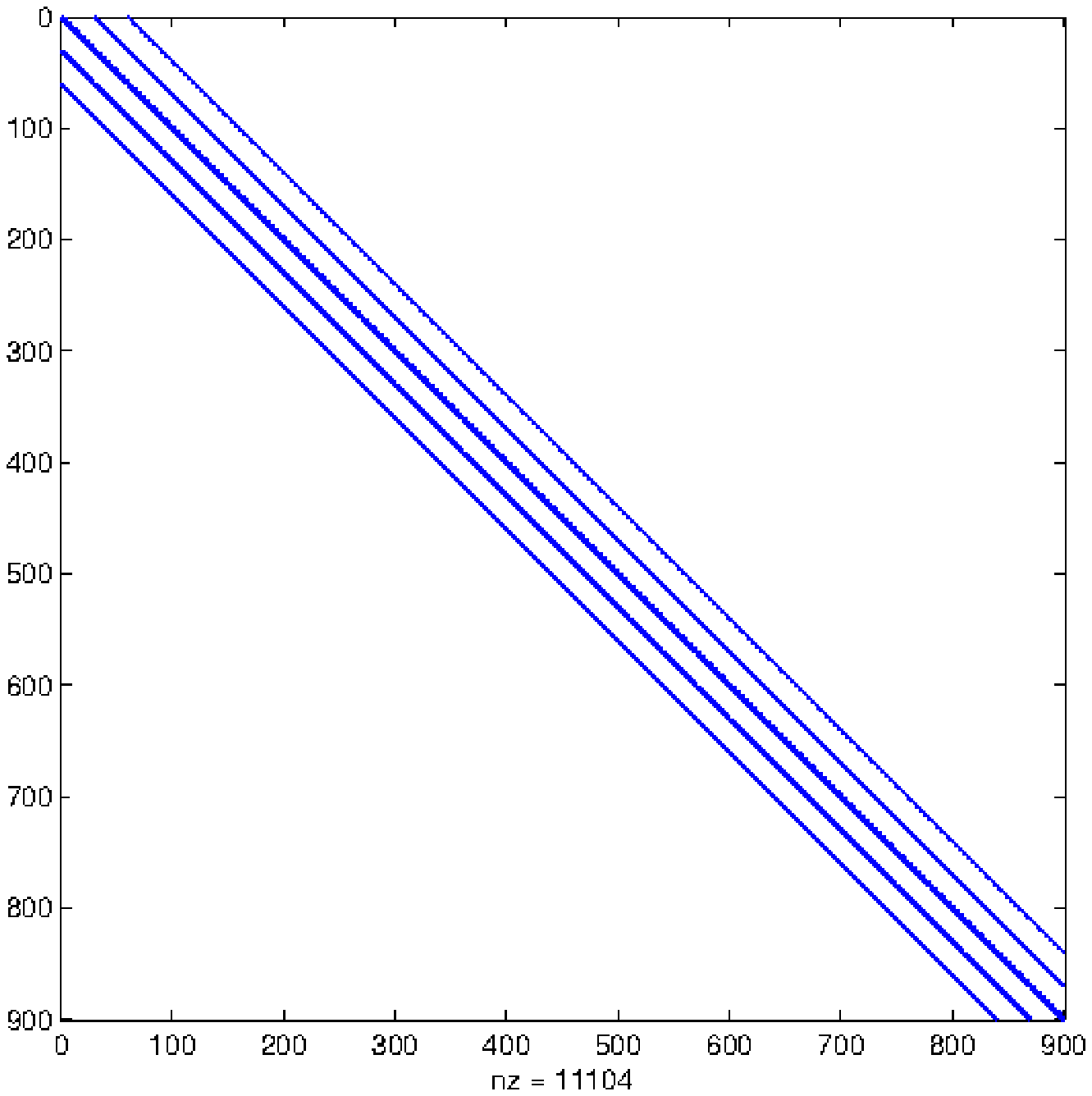} \label{fig: cholesky_SPDE1_Q} }
    \subfigure[]{\includegraphics[width=0.4\textwidth,height=0.4\textwidth]{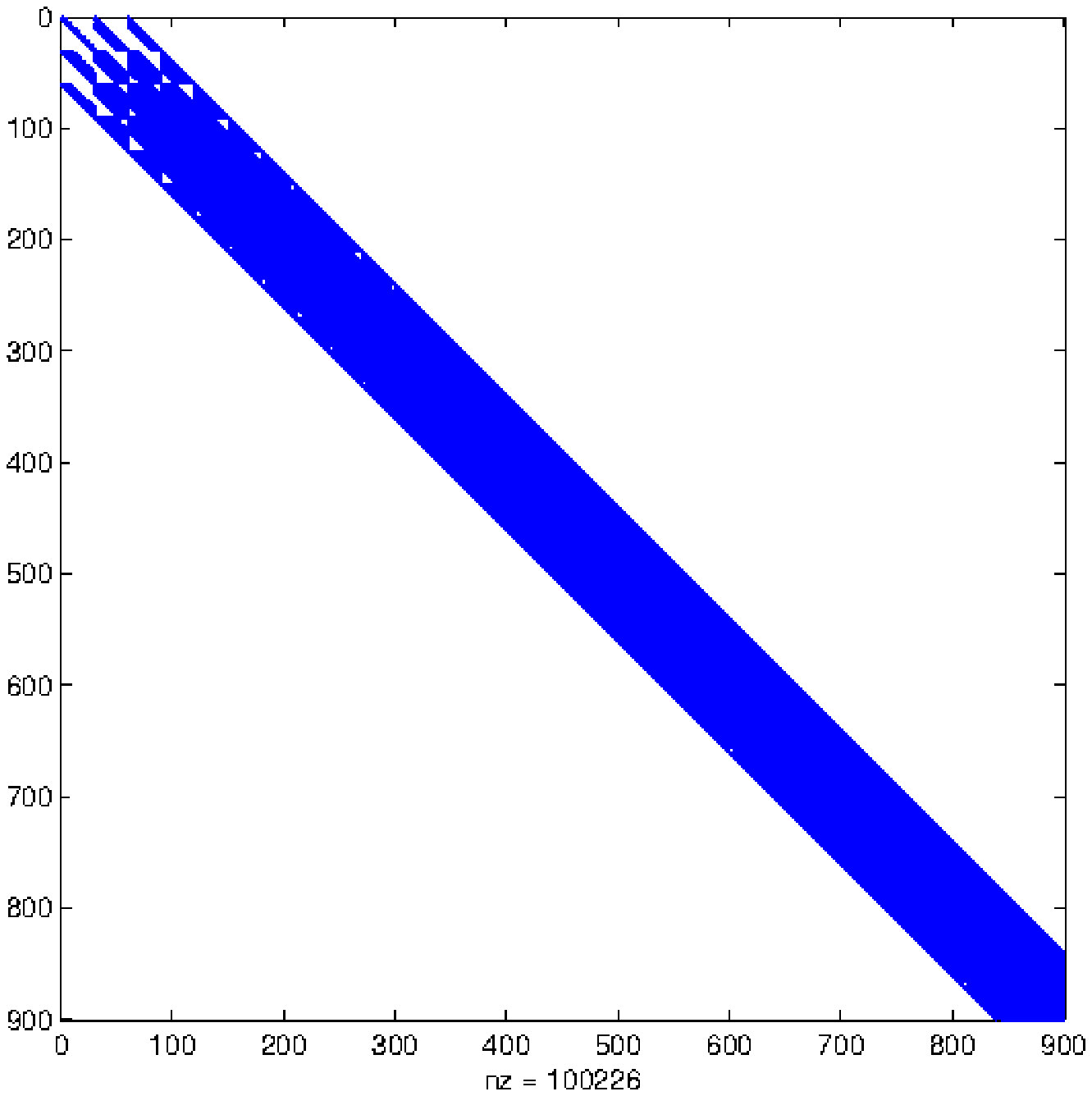} \label{fig: cholesky_SPDE1_Qq} }
    \caption{Sparsity patterns of $\boldsymbol{Q}$ (a) and $\boldsymbol{Qq}$ (b) for the random field generated from the SPDE \eqref{eq: cholesky_spde1}.} \label{fig: cholesky_spde1_QandQq}
    \end{figure}

    \begin{figure}[htbp]
    \centering
    \subfigure[]{\includegraphics[width=0.3\textwidth,height=0.3\textwidth]{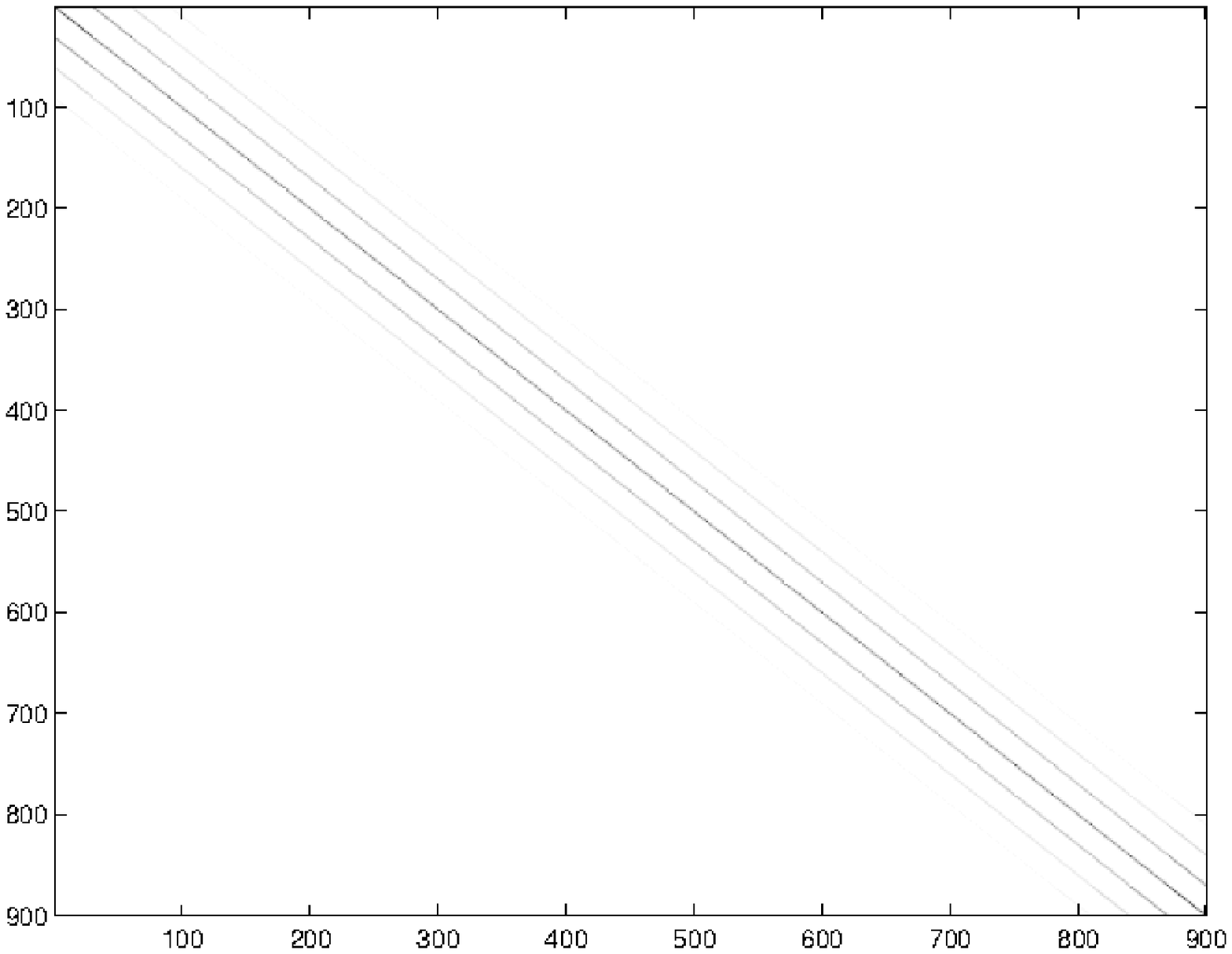} \label{fig: cholesky_SPDE1_InverseQ}}
    \subfigure[]{\includegraphics[width=0.3\textwidth,height=0.3\textwidth]{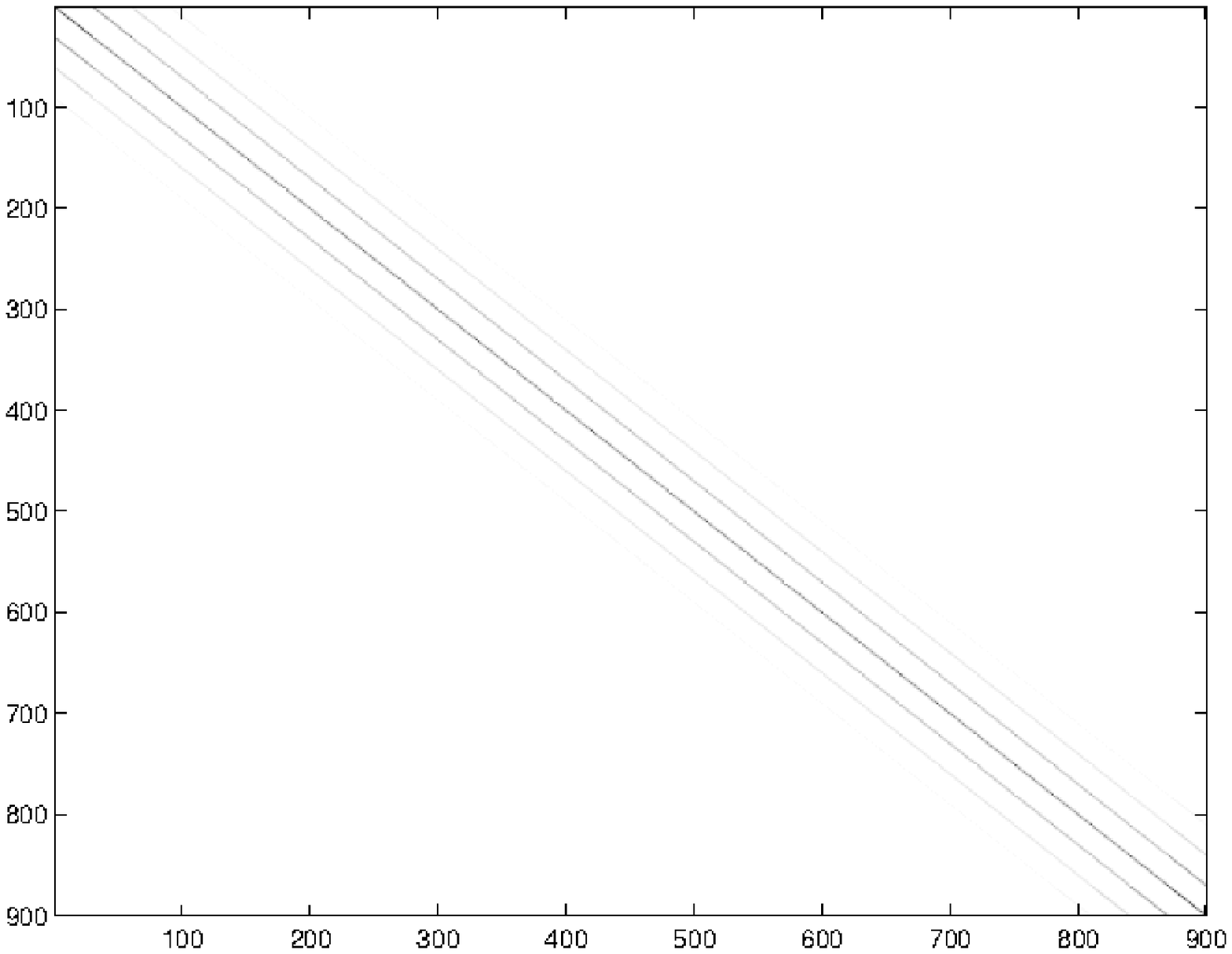} \label{fig: cholesky_SPDE1_InverseQq}}
    \subfigure[]{\includegraphics[width=0.3\textwidth,height=0.315\textwidth]{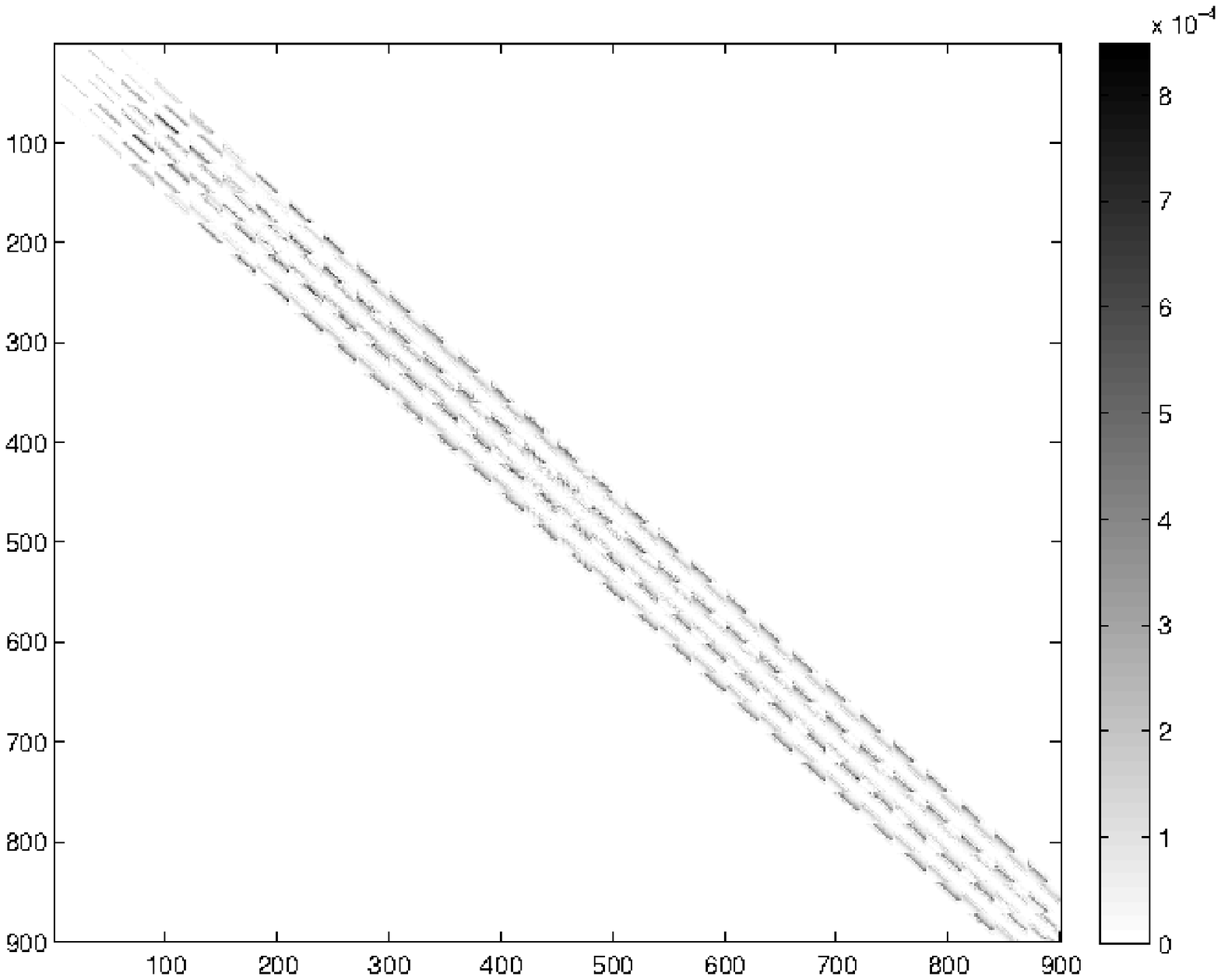} \label{fig: cholesky_SPDE1_Err}}
    \caption{Images of the true covariance matrix $\boldsymbol{\Sigma}$ (a), the approxiamted covariance matrix$\widetilde{\boldsymbol{\Sigma}}$ (b) 
            and the error matrix $\widetilde{\boldsymbol{E}}$ (c) for the random field generated from the SPDE \eqref{eq: cholesky_spde1}.} \label{fig: cholesky_result_spde1_comparation}
    \end{figure}

   \begin{figure}[htbp]
    \centering
    \subfigure[]{\includegraphics[width=0.4\textwidth,height=0.4\textwidth]{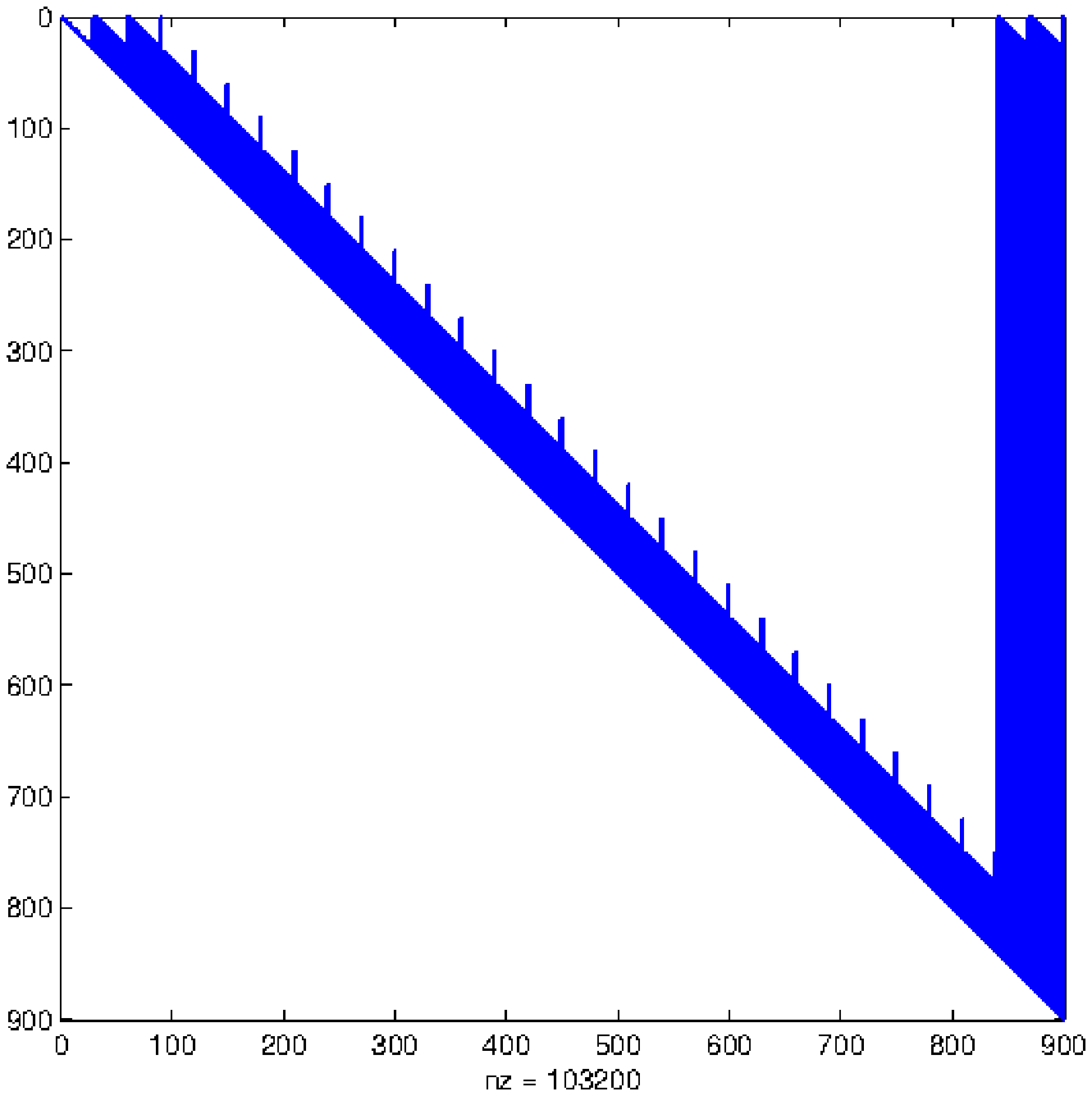} \label{fig: cholesky_SPDE2_L1}}
    \subfigure[]{\includegraphics[width=0.4\textwidth,height=0.4\textwidth]{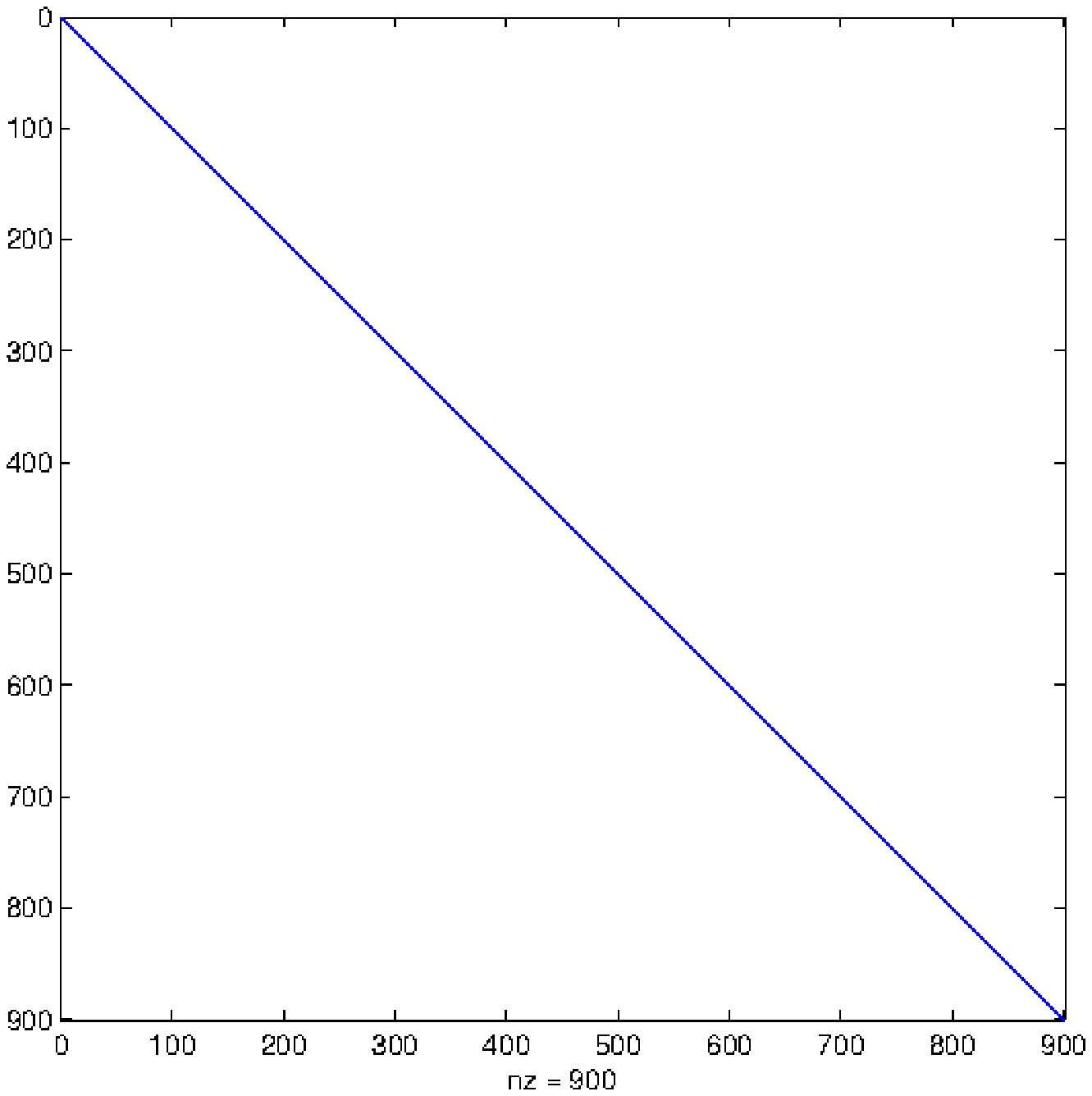} \label{fig: cholesky_SPDE2_L2}} 
    \subfigure[]{\includegraphics[width=0.4\textwidth,height=0.4\textwidth]{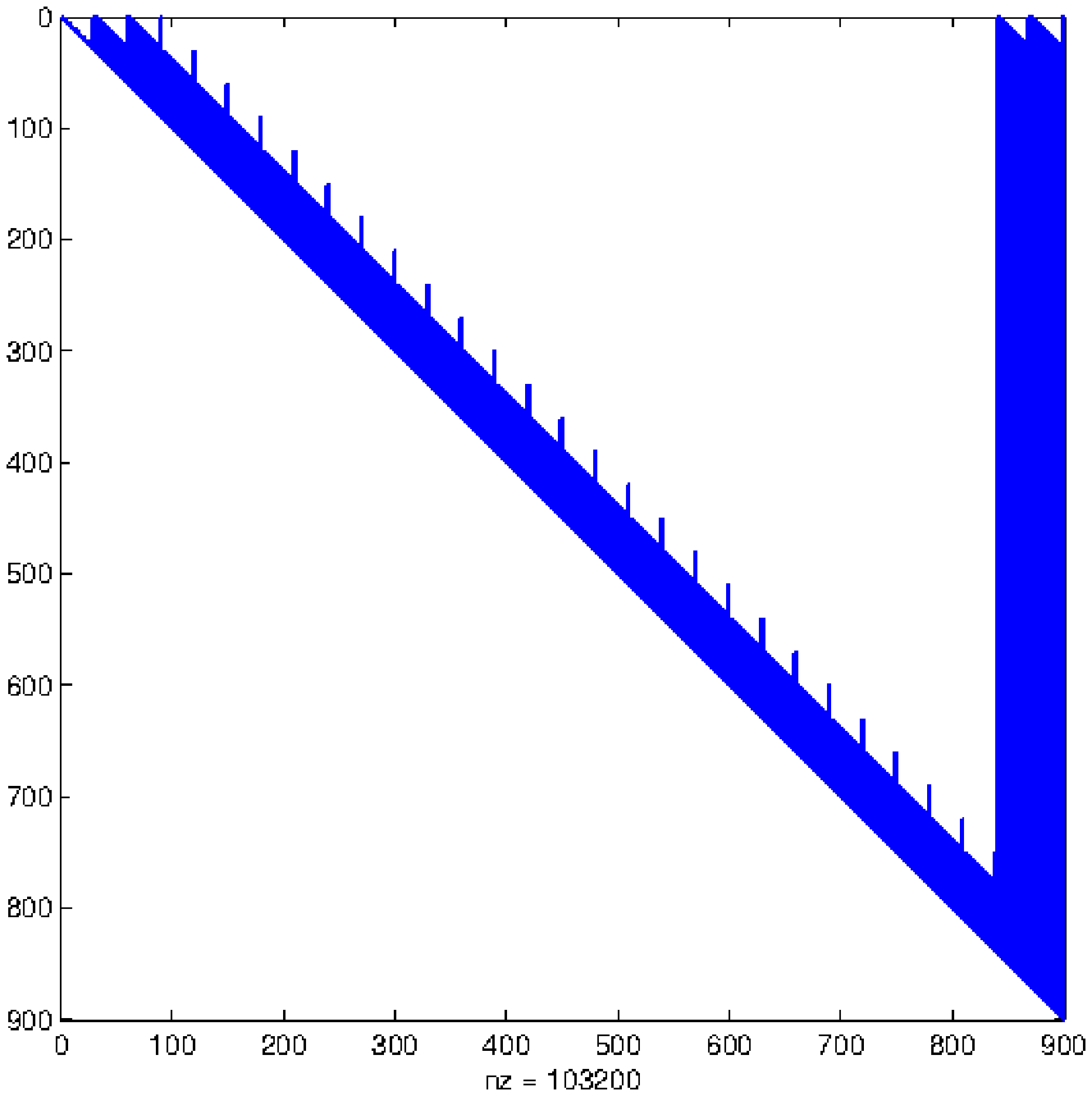} \label{fig: cholesky_SPDE2_L}}
    \subfigure[]{\includegraphics[width=0.2\textwidth,height=0.4\textwidth]{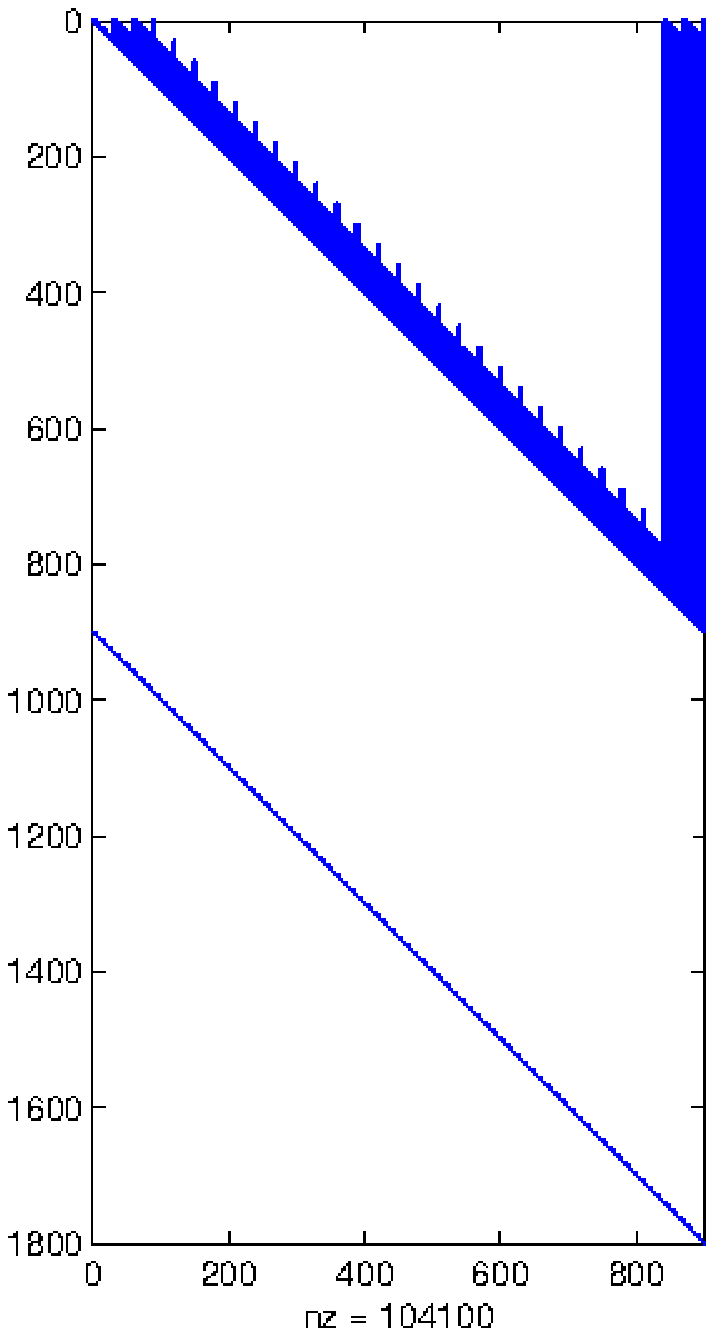} \label{fig: cholesky_SPDE2_A}}
    \subfigure[]{\includegraphics[width=0.2\textwidth,height=0.4\textwidth]{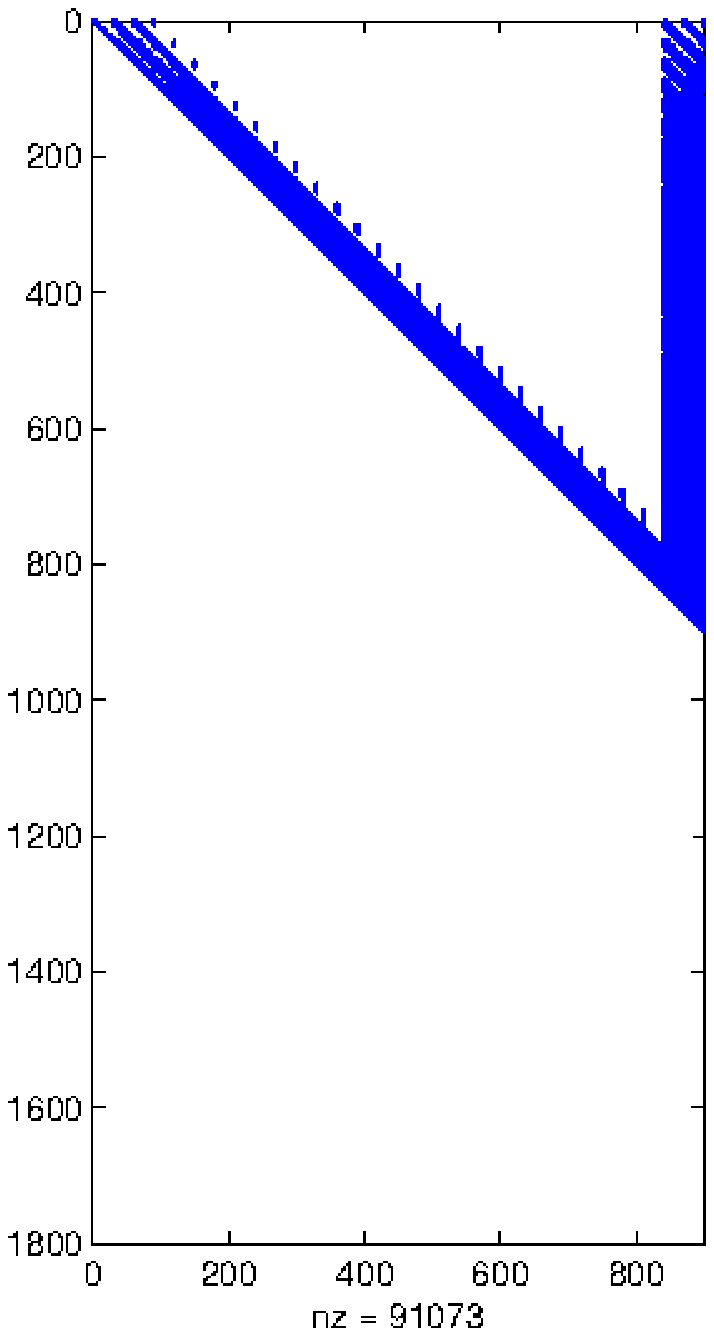} \label{fig: cholesky_SPDE2_R}}
    \caption{ Sparsity pattern for $\boldsymbol{L}_1$(a), $\boldsymbol{L}_1$(b), $\boldsymbol{L}_1$(c), $\boldsymbol{A}$ (d) and 
              $\boldsymbol{R}$(e) for the random field generated from SPDE given in \eqref{eq: cholesky_spde2}. }\label{fig: cholesly_SPDE2}
    \end{figure}

    \begin{figure}[htbp]
    \centering
    \subfigure[]{\includegraphics[width=0.4\textwidth,height=0.4\textwidth]{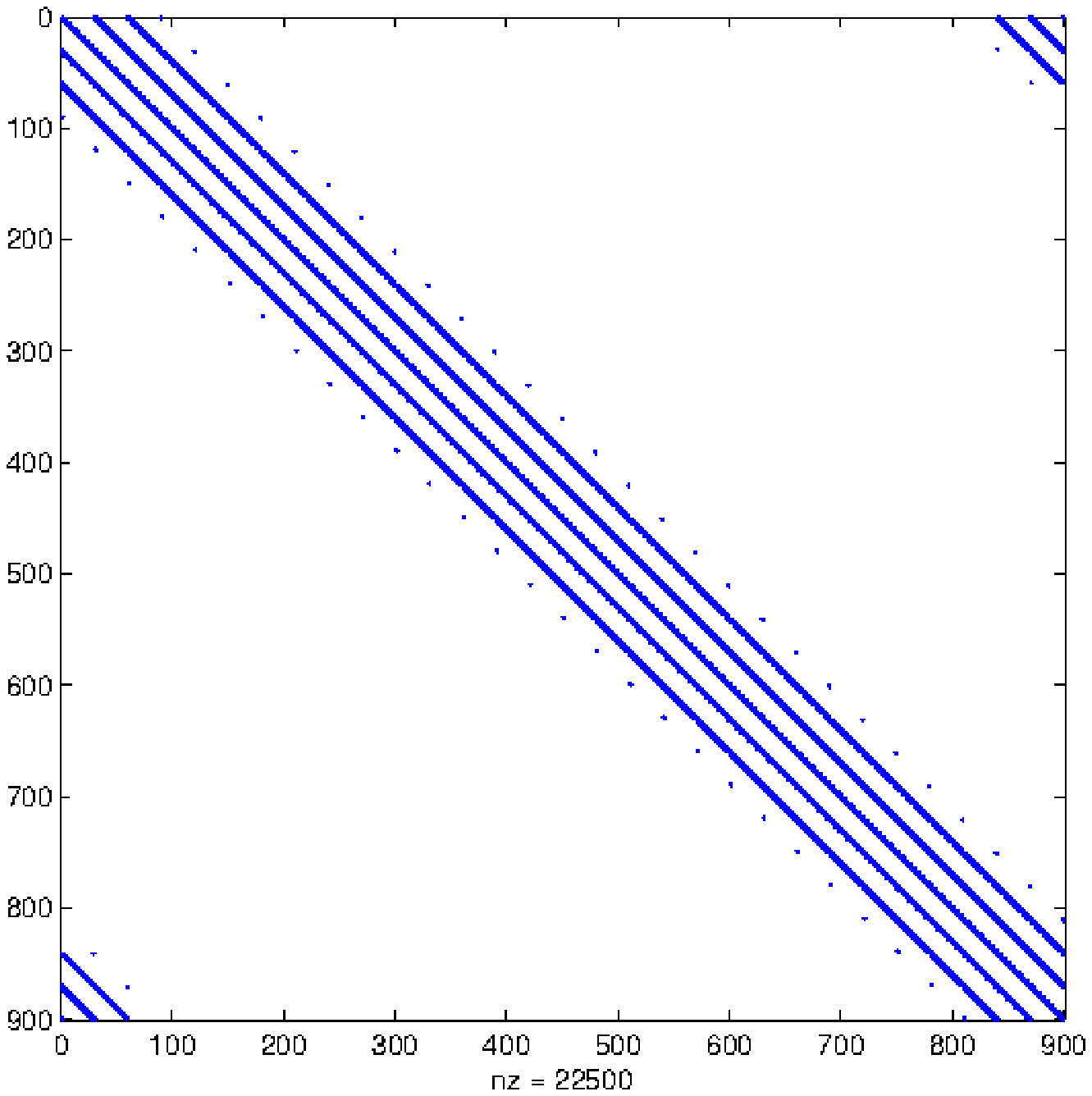} \label{fig: cholesky_SPDE2_Q} }
    \subfigure[]{\includegraphics[width=0.4\textwidth,height=0.4\textwidth]{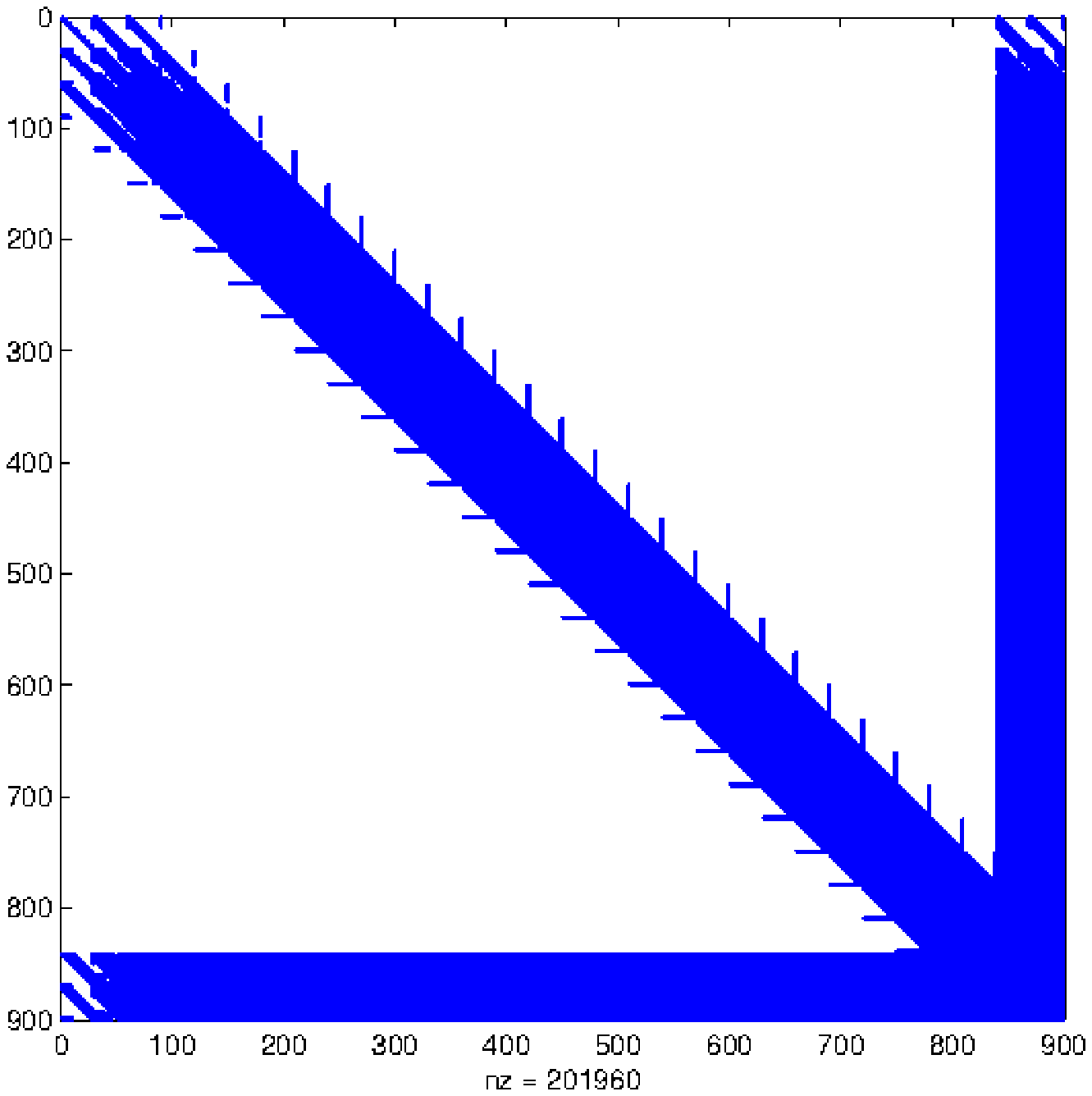} \label{fig: cholesky_SPDE2_Qq} }
    \caption{Sparsity patterns of $\boldsymbol{Q}$ (a) and $\boldsymbol{Qq}$ (b) for the random field generated from SPDE given in \eqref{eq: cholesky_spde2}.} \label{fig: cholesky_spde2_QandQq}
    \end{figure}

    \begin{figure}[htbp]
    \centering
    \subfigure[]{\includegraphics[width=0.3\textwidth,height=0.3\textwidth]{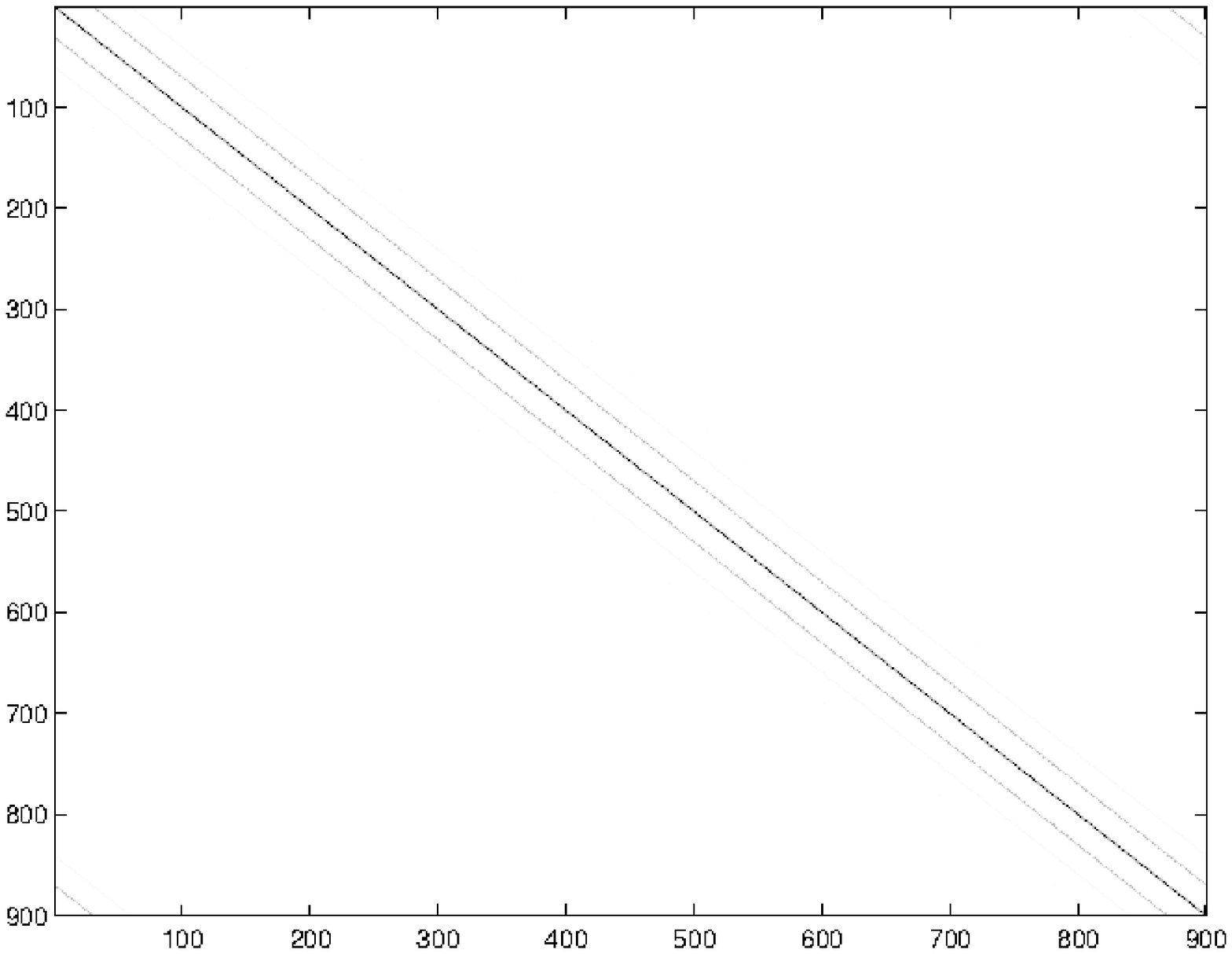} \label{fig: cholesky_SPDE2_InverseQ}}
    \subfigure[]{\includegraphics[width=0.3\textwidth,height=0.3\textwidth]{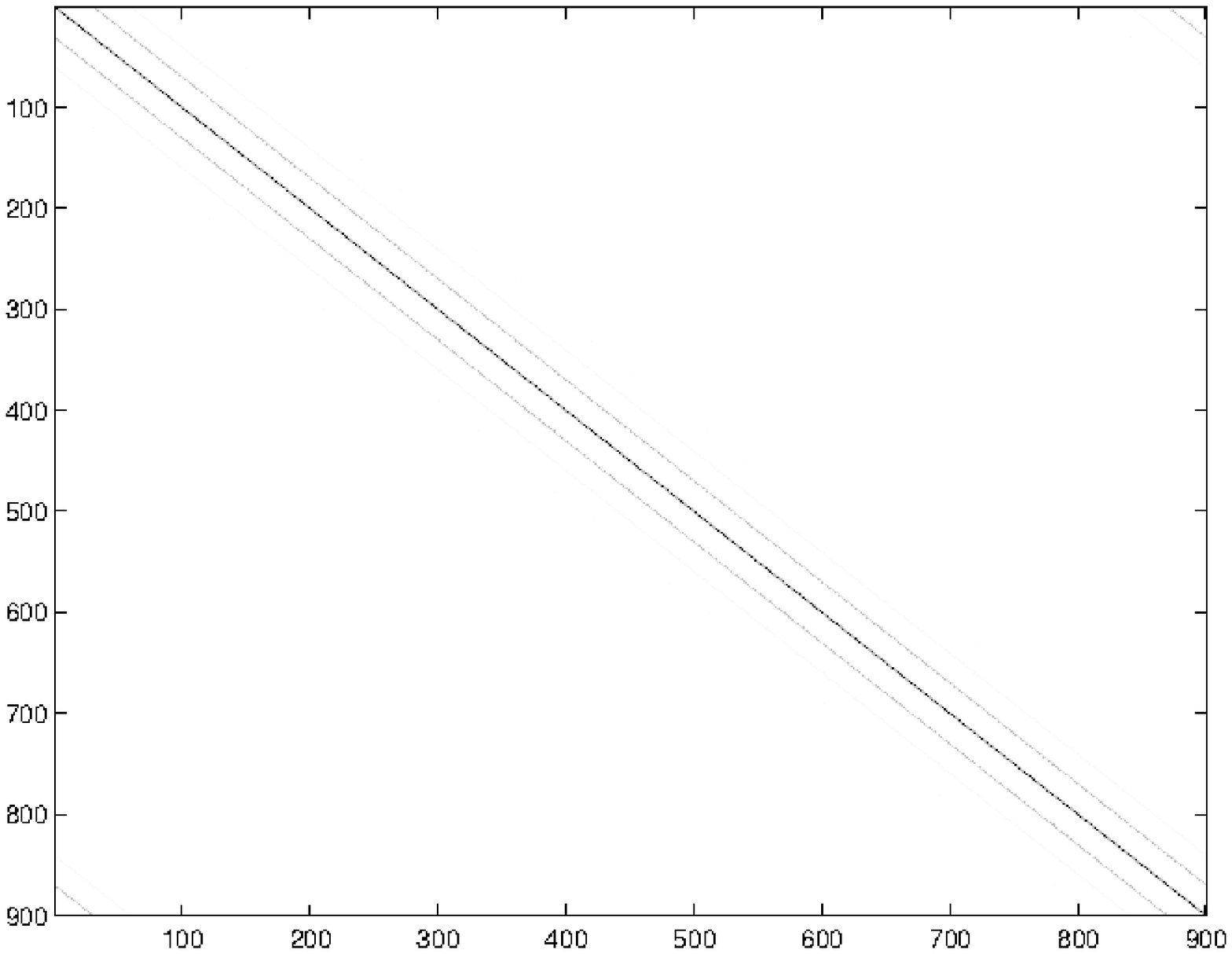} \label{fig: cholesky_SPDE2_InverseQq}}
    \subfigure[]{\includegraphics[width=0.3\textwidth,height=0.315\textwidth]{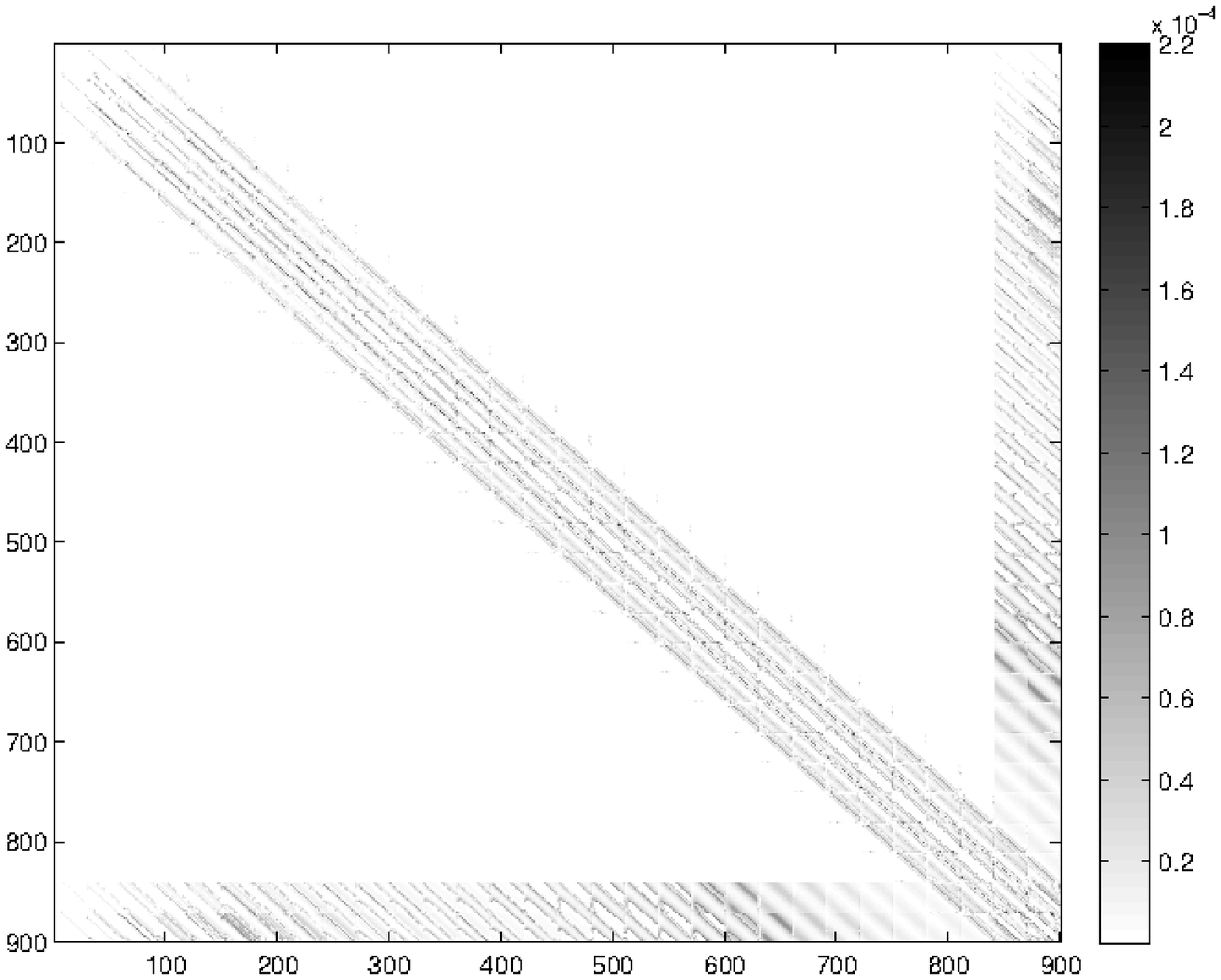} \label{fig: cholesky_SPDE2_Err}}
    \caption{Images of the true covariance matrix $\boldsymbol{\Sigma}$ (a), the approxiamted covariance matrix$\widetilde{\boldsymbol{\Sigma}}$ (b) 
            and the error matrix $\widetilde{\boldsymbol{E}}$ (c) for the random field generated from SPDE given in \eqref{eq: cholesky_spde2}.} \label{fig: cholesky_result_spde2_comparation}
    \end{figure}

\subsection{Sampling from GMRFs} \label{sec: cholesky_sampling}
   In this section samples from a GMRF are obtained using the sparse upper triangular matrices $\boldsymbol{R}$ and the Cholesky factors $\boldsymbol{L}$ for the precision matrices $\boldsymbol{Q}$.
   Let the precision matrix $\boldsymbol{Q} = \boldsymbol{Q}_1 + \boldsymbol{Q}_2$, where $\boldsymbol{Q}_1$ is from the SPDE \eqref{eq: cholesky_spde1} or \eqref{eq: cholesky_spde2},
   and $\boldsymbol{Q}_2$ is a diagonal matrix. The sampling is done as follows.
  \begin{itemize}
   \item Compute the Cholesky factor $\boldsymbol{L}$ with a Cholesky factorization or compute the sparse upper triangular matrix $\boldsymbol{R}$ from the cTIGO algorithm;
   \item Sample $\boldsymbol{z} \sim \mathcal{N}(\boldsymbol{0}, \boldsymbol{I})$;
   \item Solve the equation $\boldsymbol{L}^{\mbox{T}}\boldsymbol{x} = \boldsymbol{z}$ or $\boldsymbol{R}\boldsymbol{x} = \boldsymbol{z}$;
   \item $\boldsymbol{x}$ is the sample of the GMRF with precision matrix $\boldsymbol{Q}$ or $\widetilde{\boldsymbol{Q}}$.
  \end{itemize}
  If the mean $\boldsymbol{\mu}$ of the field is not zero, then we just need a last step $\boldsymbol{x} = \boldsymbol{\mu} + \boldsymbol{x}$ to correct the mean. With $\boldsymbol{L}$ the field $\boldsymbol{x}$ has the true covariance matrix 
  $\boldsymbol{Q}$ because
  \begin{displaymath}
   \Cov(\boldsymbol{x}) = \Cov(\boldsymbol{L^{-T}\boldsymbol{z}}) = (\boldsymbol{L}\boldsymbol{L}^{T})^{-1} = \boldsymbol{Q}^{-1}.
  \end{displaymath}
  Similarly, with $\boldsymbol{R}$ the field $\boldsymbol{x}$ has the approximated covariance matrix $\tilde{\boldsymbol{Q}}$.
   Many other sampling algorithms are provided by \citet[Chapter 2]{rue2005gaussian} for different parametrization of the GMRF. 
    We cannot notice any large differences between the samples using the Cholesky factor $\boldsymbol{L}$ and the samples using the sparse matrix $\boldsymbol{R}$ 
    based on Figure \ref{fig: cholesky_sampling1} and Figure \ref{fig: cholesky_sampling2}. 
    
    \begin{figure}[htbp]
    \centering
    \subfigure[]{\includegraphics[width=0.4\textwidth,height=0.4\textwidth]{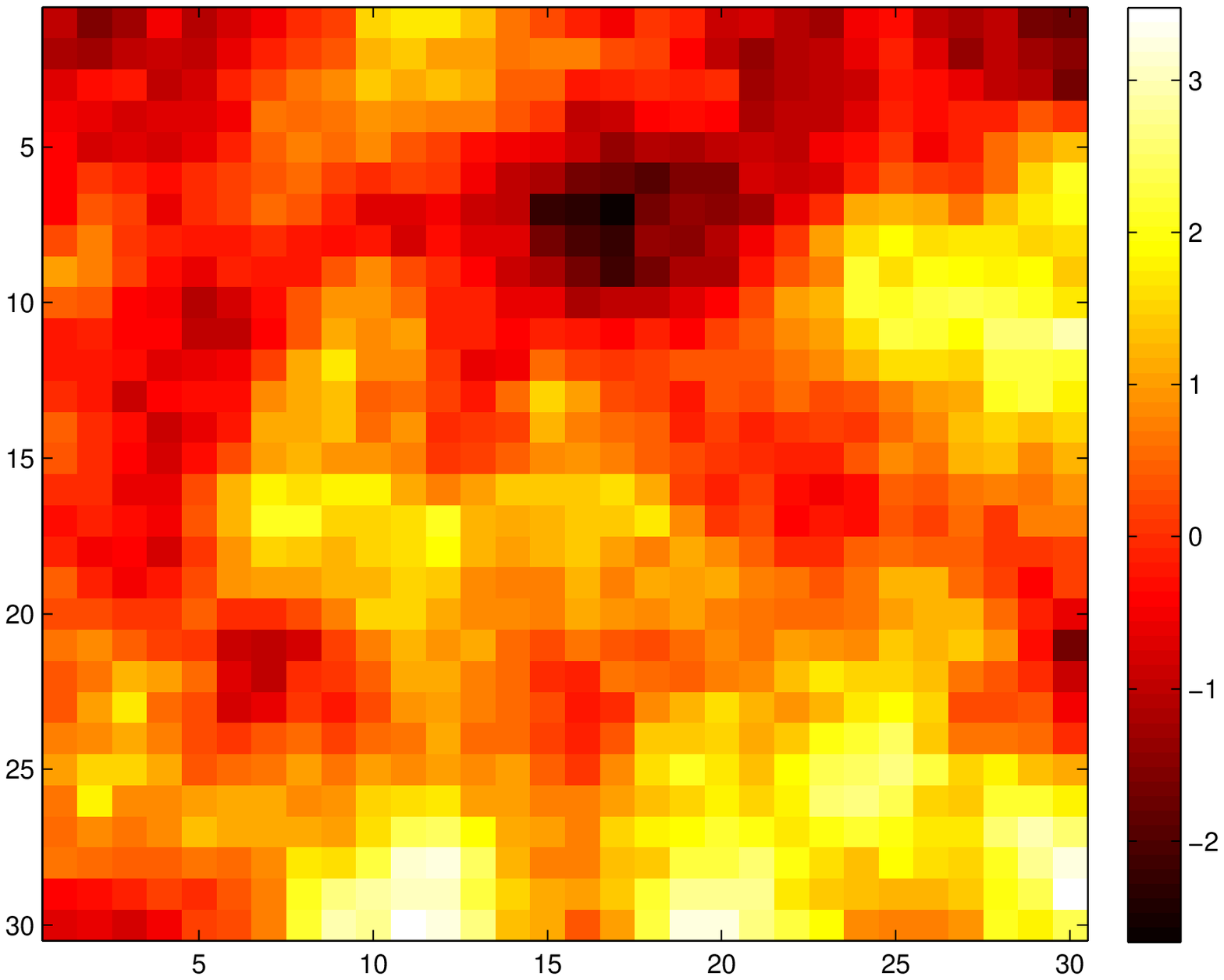}}
    \subfigure[]{\includegraphics[width=0.4\textwidth,height=0.4\textwidth]{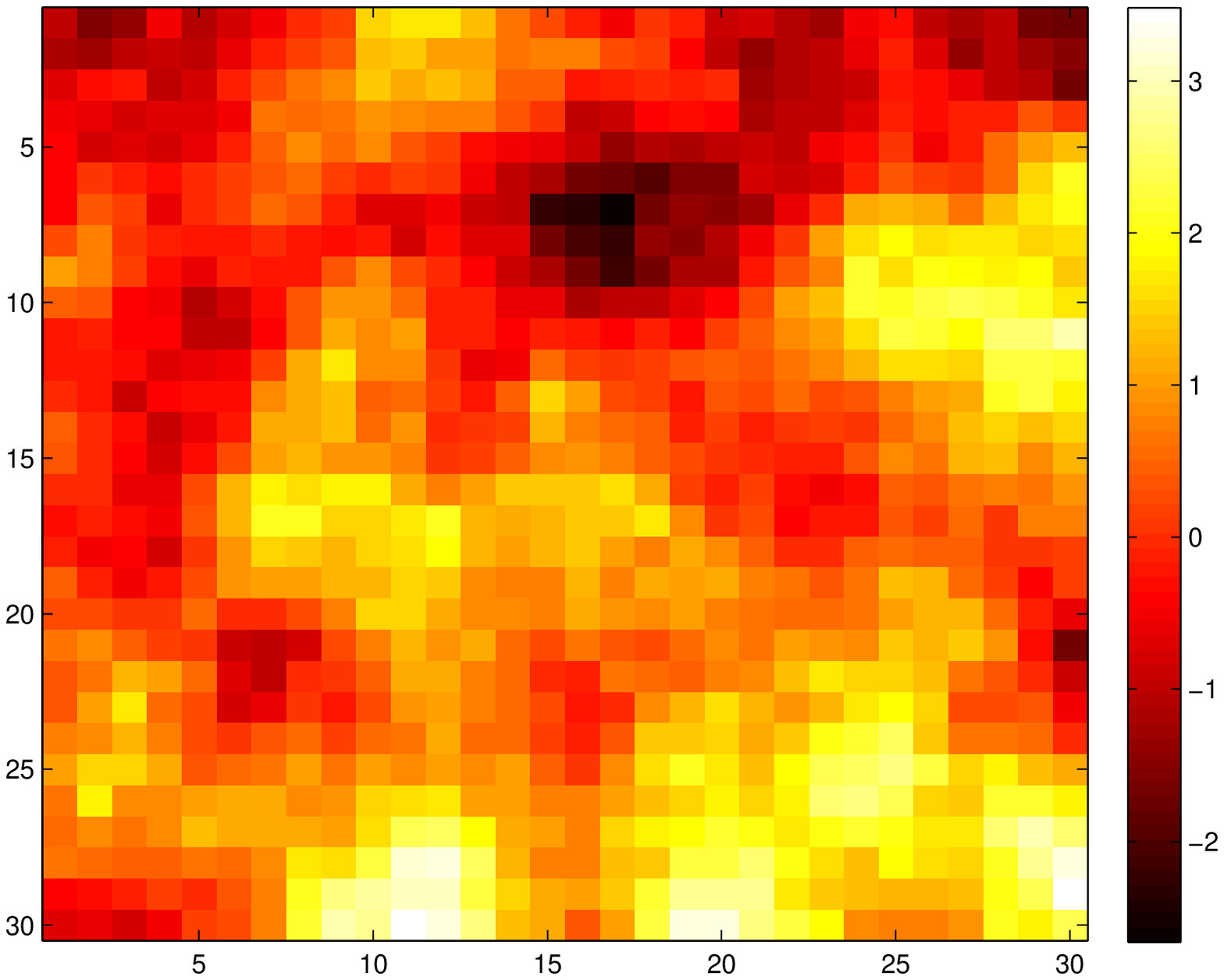}}
    \caption{Samples from the GMRF with the common Cholesky factor $\boldsymbol{L}$ with \eqref{eq: cholesky_spde1}
             (a) and the upper triangular matrix $\boldsymbol{R}$ (b) from the cTIGO algorithm.} \label{fig: cholesky_sampling1}
    \end{figure}

    \begin{figure}[htbp]
    \centering
    \subfigure[]{\includegraphics[width=0.4\textwidth,height=0.4\textwidth]{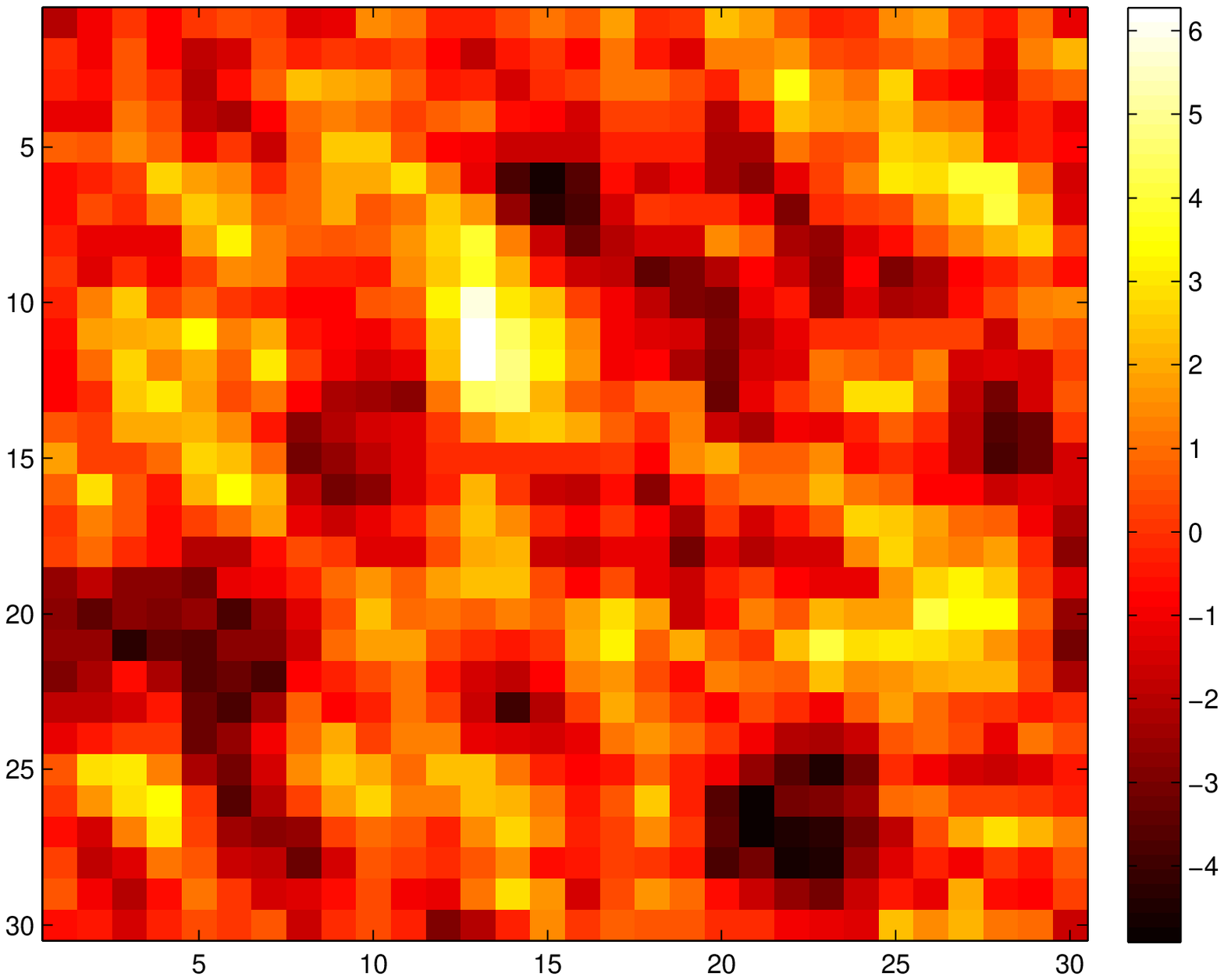}}
    \subfigure[]{\includegraphics[width=0.4\textwidth,height=0.4\textwidth]{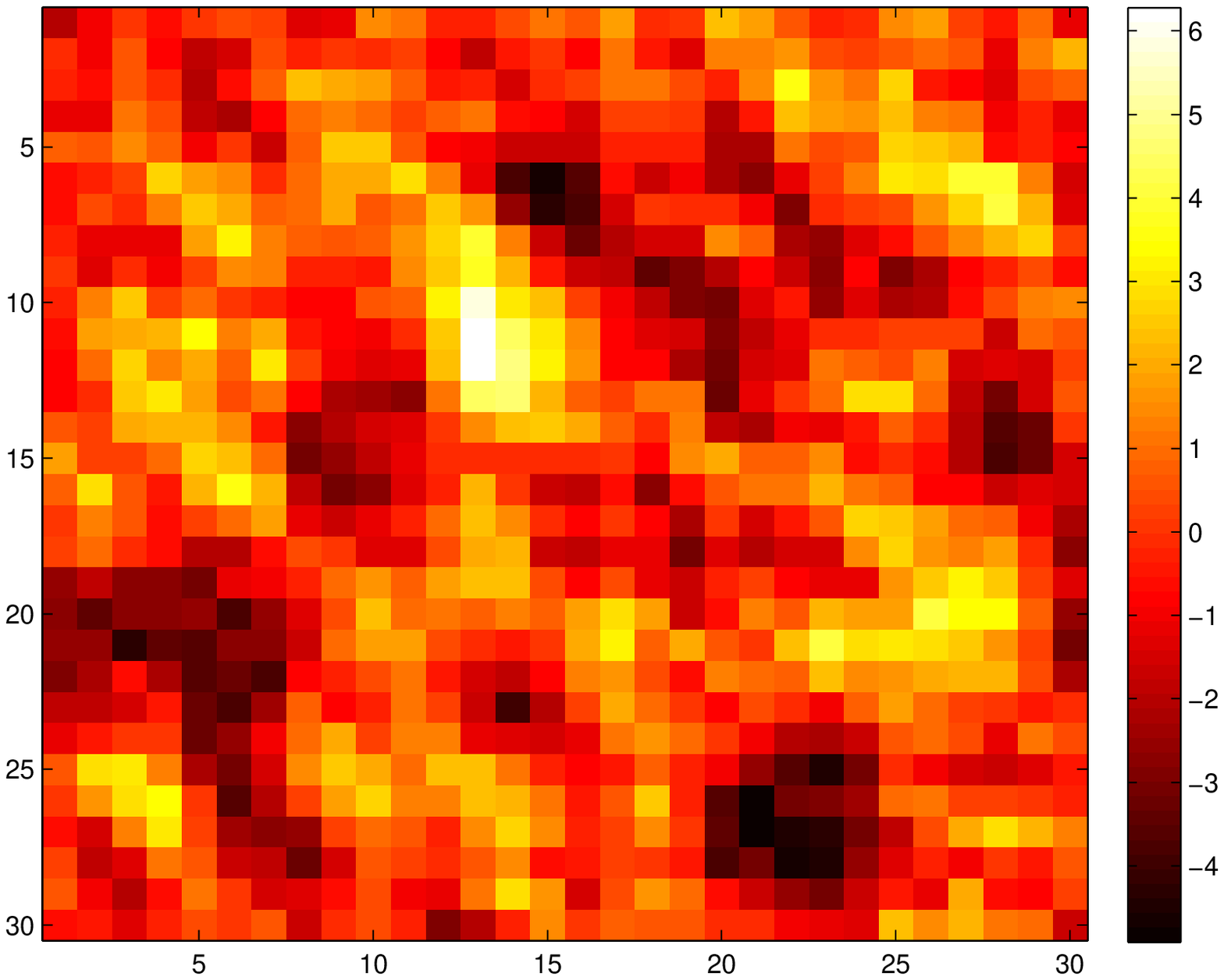}}
    \caption{Samples from the GMRF with the common Cholesky factor $\boldsymbol{L}$ (a) with \eqref{eq: cholesky_spde2}
              and the upper triangular matrix $\boldsymbol{R}$ (b) from the cTIGO algorithm.} \label{fig: cholesky_sampling2}
    \end{figure}
 
\subsection{Effect of dropping tolerance} \label{sec: cholesky_result_comparsion}
    In this section we choose different values of $\tau$ in order to know the effect of dropping tolerance.  We use the same kinds of structures for the precision matrices as discussed in Section
    \ref{sec: cholesky_result_band_matrices} with dropping tolerances 
    $ \tau = \{0, 0.000001, 0.00001, 0.0001, 0.001, 0.01\} $. The $1$-norm for the error matrix $\boldsymbol{E} = \boldsymbol{Q} - \boldsymbol{\widetilde{Q}}$
    is used for the comparisons. The results are given in Table \ref{RWcomparison} and Table \ref{othercomparison}. 
    From these tables, we can see that as the dropping tolerance $\tau$ becomes smaller and smaller, the error becomes smaller and smaller. Notice that by choosing $10^{-5}$
    as the dropping tolerance, the error reaches a level acceptable in many applications. If the dropping tolerance is equal to $0$, then the error is also equal to zero which means no element
    has been zeroed out during the Givens rotations and it returns the common Cholesky factor.

   \begin{table}[htb]
   \caption{Comparisons for RW$1$ Model and RW2 Model} \label{RWcomparison}
   \centering
   \begin{minipage}[b]{0.35\textwidth}
   \begin{tabular}{r|r}
    \hline
      \multicolumn{2}{c} {Random Walk 1} \\
      \hline
      tolerance & error \\
     \hline
      0.01    & 2.55E-04 \\
      0.001   & 1.66E-06 \\
      0.0001  & 2.13E-08 \\
      0.00001 & 2.50E-10 \\
      1.00E-6 & 2.34E-12 \\
            0 & 4.00E-15 \\
      \hline
    \end{tabular}
    \end{minipage}
   \begin{minipage}[b]{0.35\textwidth}
   \begin{tabular}{r|r}
    \hline
      \multicolumn{2}{c} {Random Walk 2} \\
      \hline
      tolerance & error \\
     \hline
      0.01    & 0.33 \\
      0.001   & 1.80E-05 \\
      0.0001  & 1.48E-07 \\
      0.00001 & 1.07E-08 \\
      1.00E-6 & 2.15E-09 \\
            0 & 1.73E-14 \\
      \hline
   \end{tabular}
   \end{minipage}
\end{table}

\begin{table}[htb]
   \caption{Comparisons for Poisson matrix and Toeplitz matrix} \label{othercomparison}
   \centering
   \begin{minipage}[c]{0.35\textwidth}
   \begin{tabular}{r|r}
    \hline
      \multicolumn{2}{c} {Poisson matrix} \\
      \hline
      tolerance & error \\
     \hline
      0.01    & 0.11 \\
      0.001   & 8.51E-03 \\
      0.0001  & 6.33E-04 \\
      0.00001 & 5.91E-05 \\
      1.00E-6 & 3.44E-06 \\
            0 & 1.49E-14 \\
      \hline
    \end{tabular}
    \end{minipage}
   \begin{minipage}[c]{0.35\textwidth}
   \begin{tabular}{r|r}
    \hline
      \multicolumn{2}{c} {Toeplitz matrix} \\
      \hline
      tolerance & error \\
     \hline
      0.01    & 9.15E-03 \\
      0.001   & 9.59E-04 \\
      0.0001  & 7.91E-05 \\
      0.00001 & 8.70E-06 \\
      1.00E-6 & 7.19E-07 \\
            0 & 5.66E-15 \\
      \hline
    \end{tabular}
    \end{minipage}
   \end{table}

\section{Discussion and Conclusion} \label{sec: cholesky_conclusion}
   In this paper we use the cTIGO algorithm to find sparse Cholesky factors for specifying GMRFs. 
   Some commonly used structures of the precision matrices and two precision matrices generated from SPDEs have been tested.
   By using the incomplete orthogonal factorization with Givens rotations, a sparse incomplete Cholesky factor can be found and it is usually sparser than the Cholesky factor from the standard Cholesky factorization. 
   The sparsity of the incomplete Cholesky factor depends on the value of the tolerance. With a good choice for the dropping tolerance, the error between the true covariance matrix and the approximated covariance matrix becomes negligible. 

   One advantage of this approach is that it is robust. It always produces a sparse incomplete Cholesky factor. Since the algorithm works both for square matrices and for rectangular matrices,
   this approach can be applied to GMRFs conditioned on observed data or a subset of the variable.
   On the negative side, it seems that our current implementation of the approach is slow when the dimension of the matrix becomes large. 
   We believe that this is due to the nature of the incomplete orthogonal factorization with dynamic dropping strategy.
   The orthogonal factorization is usually slower than the Cholesky factorization. Further, Givens rotations only zero out values to zeros one at a time.
   This leads to the slowness of the algorithm. When the computation resources are limited, 
   we might need to use the fixed pattern dropping strategy. However, to implement a fast cTIGO algorithm is out the scope of this paper and it is for further research.

\clearpage
\bibliographystyle{plainnat}
\linespread{0.8}{\small {\bibliography {Reference/Ref}}}
\end{document}